\documentclass[12pt]{report}
\usepackage{latexsym}
\usepackage{appendix}  

\usepackage{amssymb}
\usepackage{amsmath}
\usepackage{amscd}
\usepackage{amsthm}
\usepackage{emptypage} % to remove headers and footers for pages between chapters
\usepackage[english]{babel}
\usepackage[utf8x]{inputenc}
\usepackage{amsmath}

\usepackage{graphicx}
\usepackage[colorinlistoftodos]{todonotes}
\usepackage{amsthm}
\usepackage{bigints} 

\usepackage[a4paper, total={6in, 8.5in}]{geometry}

\usepackage[english]{babel}
\usepackage[utf8x]{inputenc}
\usepackage{amsfonts}
\usepackage{graphicx}
\usepackage{float}
\usepackage{hyperref}
\usepackage{url}

\def\thetitle{Vortices, Painlev\'e integrability and projective geometry}
\def\theauthor{Felipe Contatto}

\def\theyear{March 2018}

\newcommand{\koniec}{\begin{flushright}  $\Box $ \end{flushright}}
\newtheorem{defn}{Definition}[chapter]
\newtheorem{theo}{Theorem}[chapter] 
\newtheorem{prop}[theo]{Proposition}  
\newtheorem{lemma}[theo]{Lemma}

\newtheorem{coll}[theo]{Corollary}

\newcounter{mnotecount}[section]
\numberwithin{equation}{chapter} %number equations according to chapters not sections
\renewcommand{\themnotecount}{\thesection.\arabic{mnotecount}}

\newcommand{\mnote}[1]%{}%
{\protect{\stepcounter{mnotecount}}$^{\mbox{\footnotesize
$%\!\!\!\!\!\!\,
\bullet$\themnotecount}}$ \marginpar{%\color{red}%
\raggedright\tiny\em
$\!\!\!\!\!\!\,\bullet$\themnotecount: #1} }

\newcommand{\CP}{\mathbb{CP}}
\newcommand{\C}{\mathbb{C}}

\newcommand{\R}{\mathbb{R}}
\newcommand{\HH}{\mathbb{H}}

\newcommand{\Rho}{\mathrm{P}}
\def\p{\partial}
\def\be{\begin{equation}}
\def\ee{\end{equation}}

\def\bea{\begin{eqnarray}}
\def\eea{\end{eqnarray}}

\def\ov{\overline}

\newcommand{\dzbar}{\partial_{\bar z}}
\DeclareMathOperator{\tr}{Tr}
\DeclareMathOperator{\arctanh}{arctanh}
\usepackage{geometry}

\begin{document}

\pagenumbering{roman}
%\doublespaced
\large\newlength{\oldparskip}\setlength\oldparskip{\parskip}\parskip=.3in
\thispagestyle{empty}
\begin{center}
%\vspace*{\fill}
{\Huge\textbf\thetitle}

\vspace{0.8in}

{\huge\theauthor}

\vspace{-0.15in}

Queens' College
 
\end{center}
\vspace{0.1in}
\begin{figure}[H]
\begin{center}
\includegraphics[scale=1]{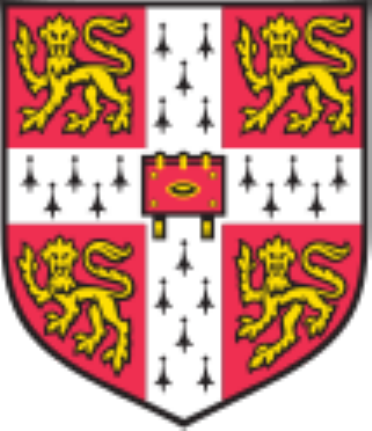}
\end{center}
  %\end{tabular}
\end{figure}

\noindent%\singlespaced\large
\begin{center}
\textit {This dissertation is submitted for the degree of Doctor of Philosophy.}
\end{center}

\begin{center}
Department of Applied Mathematics and Theoretical Physics

\vspace{-0.1in}

University of Cambridge

Supervisor: Dr. Maciej Dunajski
\end{center}

%\doublespaced
%\large
\begin{center}
\theyear
\end{center}

%\noindent Supervisor of Dissertation

\bigskip
%\noindent\makebox[-0.06in][l]{\rule[2ex]{3in}{.3mm}}
%\singlespaced

%\noindent Graduate Group Chairperson

%\bigskip
%\noindent\makebox[-0.06in][l]{\rule[2ex]{3in}{.3mm}} 
%%\singlespaced
%Graduate Chair's Name, Graduate Chair's Faculty Title\\

%\noindent
%%\singlespaced
%Dissertation Committee:\\
%Member 1, Title\\
%Member 2, Title\\
%Member 3, Title
%\vspace*{\fill}
%
%\normalsize\parskip=\oldparskip

\pagenumbering{gobble}
\newpage
%\doublespaced

\newpage
\vspace*{\fill}
\newpage

%\vspace*{\fill}
\begin{center}
 \textbf{Vortices, Painlev\'e integrability and projective geometry}\\
%\thetitle\\
%\vspace{.5in}
%  \theauthor\\
%  \theadvisor
\end{center}

\begin{center} Felipe Contatto \end{center}

The first half of the thesis concerns Abelian vortices and Yang--Mills theory. It is proved that the $5$ types of vortices recently proposed by Manton are actually symmetry reductions of (anti-)self-dual Yang--Mills equations with suitable gauge groups and symmetry groups acting as isometries in a $4$-manifold. As a consequence, the twistor integrability results of such vortices can be derived. It is presented a natural definition of their kinetic energy and thus the metric of the moduli space was calculated by the Samols' localisation method. Then, a modified version of the Abelian--Higgs model is proposed in such a way that spontaneous symmetry breaking and the Bogomolny argument still hold. The Painlev\'e test, when applied to its soliton equations, reveals a complete list of its integrable cases. The corresponding solutions are given in terms of third Painlev\'e transcendents and can be interpreted as original vortices on surfaces with conical singularity. 

The last two chapters present the following results in projective differential geometry and Hamiltonians of hydrodynamic-type systems. It is shown that the projective structures defined by the Painlev\'e equations are not metrisable unless either the corresponding equations admit first integrals quadratic in first derivatives or they define projectively flat structures. The corresponding first integrals can be derived from Killing vectors associated to the metrics that solve the metrisability problem. Secondly, it is given a complete set of necessary and sufficient conditions for an arbitrary affine connection in $2D$ to admit, locally, $0$, $1$, $2$ or $3$ Killing forms. These conditions are tensorial and simpler than the ones in previous literature. By defining suitable affine connections, it is shown that the problem of existence of Killing forms is equivalent to the conditions of the existence of Hamiltonian structures for hydro\-dynamic-type systems of two components.

%\doublespaced
\noindent

% Put the content of the abstract here

%\vspace*{\fill}

\newpage
\vspace*{\fill}
\newpage

\chapter*{Declaration}

This thesis is the result of my own work and includes nothing which is the outcome of work done in collaboration except as specified in the text. It is not substantially the same as any work that I have submitted, or, is being concurrently submitted for a degree or diploma or other qualification at the University of Cambridge or any other university or similar institution except as specified in the text. I further state that no substantial part of my thesis has already been submitted, or, is being concurrently submitted for any such degree, diploma or other qualification at the University of Cambridge or any other university or similar institution except as specified in the text.

Chapters \ref{chapSDYM}, \ref{chapPainMet} and \ref{chapterKilling} are based on (but are not the same as) my papers in collaboration with my supervisor \cite{ConDun2017,CD16,FCMD2015}, respectively, except for Section \ref{secDegSol}. Chapter \ref{chapVorLike} is essentially the same as my individual paper \cite{Contatto2017}.

\newpage
\vspace*{\fill}
\newpage

%\begin{center}
% \textbf{Acknowledgements}\\
%\thetitle\\
%\vspace{.5in}
%  \theauthor\\
%  \theadvisor
%\end{center}

\textbf{Acknowledgements}. I am very grateful for the friendship and support of my close friends Fiona Doherty, Fernando Oliveira, Felipe Pellison, Taynan Ferreira, Joshua Shutter, Fr. Mark Langham, Fr. Bruno Clifton and the Fisher House community in general. I appreciate very much the patience and helpful hands of my parents, Jacir and Inês, and of my brother Fábio while I am away during my PhD and pursuing future plans. For the past year, I have been living in the Adams Road community. I am grateful for the company I enjoyed with them. In particular, I am thinking of Robin Cunnah and Helen Holmes. My special thanks go to Rosemary Summers, who passed away in the beginning of last year and had allowed me to live there, and to Geoffrey Roughton, her brother, for keeping the community alive. My officemates Tom Cridge, Sebastián C\'espedes and Kai Roehrig made my time at DAMTP much more enjoyable, especially at coffee time.

I am grateful to Prof. Nick Manton, Prof. Jenya Ferapontov, Dr. David Stuart and Dr. Daniele Dorigoni for the support they have given for my career and the reference letters they were so kind to write for me. I would also like to thank Prof. Nick Manton and his PhD student, Edward Walton, for numerous interesting discussions about Abelian vortices. Also, it is fair to say that they independently derived the moduli space metric described in Section \ref{secEnMod} at the same time as me. I thank Prof. Robert Bryant for the useful discussions about the Liouville invariants and for bringing reference \cite{Liouville1889} to our attention.

Finally, I want to show my most sincere gratitude to my supervisor Dr. Maciej Dunajski, for his high quality supervision, patience and attention, besides being so kind to revise this thesis. His encouragement and support during the most difficult times of my career is something I will never forget. Moreover, I would like to thank him for the various mathematical ideas he shared with me leading to the academic production we have had together. In particular, it is fair to attribute solely to him the ideas and proofs of Theorem \ref{theo_2} and Proposition \ref{prop_rank1}.

\tableofcontents

\newpage
%\draftspaced
\pagenumbering{arabic}
\include{introdept}
\include{back}
\include{finitedept}
\include{infinitedept}

\chapter{Introduction}

%\begin{flushright}
%\textit{God has placed in the human heart}
%\vspace{-0.7cm}
%
%\textit{a desire to know the truth}
%
%\vspace{-0.7cm}
%St. John Paul II
%\end{flushright}

This thesis approaches two areas related to each other by Painlev\'e integrability. The first half deals with Abelian vortices and the second with differential and projective geometry. Each of the chapters is self-contained and has its own notation clearly explained. In general, parenthesis $(\; )$ and brackets $[\; ]$ around tensorial indices mean normalised symmetrisation and antisymmetrisation, respectively. The symbol $\mathcal L$ means lie derivative, except in section \ref{sec_mec}. Equalities expressing definitions are denoted with $\equiv$. Finally, tensorial indices within an unspecified coordinate system correspond to the abstract index notation \cite{PenRind1984}. 

Abelian vortices are topological solitons in $2$ dimensions arising from the Abelian Higgs model. This is a relativistic $U(1)$-gauge field theory with a $\phi^4$-potential whose vacuum manifold $S^1$ is given by the Higgs fields that annihilate the potential, that is to say, satisfying $|\phi|^2=1$. Its non-trivial topology provides stability to the vortices. In fact, the topological charge is interpreted as the winding number of the Higgs field at infinity or, equivalently, as the number of zeros of the Higgs field, counting multiplicities. Moreover, since vortices are a result of the Bogomolny argument, they are given as solutions to a system of two first order equations that also satisfy the Euler--Lagrange equations. In practice, vortices are given as solutions to the so called ``Taubes equation" \cite{Taubes1980} subject to boundary conditions, a single second order equation as opposed to the two first order ones. The Taubes equation is not easy to solve and is usually tackled numerically. The first analytic solution was given by Witten \cite{Witten1977} and was interpreted as an $SU(2)$-instanton with rotational symmetry. After that, two more explicit solutions were found by Dunajski \cite{Dunajski2012} involving Painlev\'e transcendents.

Another four similar types of solitons were proposed by Manton \cite{Manton2017} as a generalisation of the Taubes equation which includes the Popov vortices \cite{Popov2012}, a result of the symmetry reduction of anti-self-dual Yang--Mills (ASDYM) equations with gauge group $SU(1,1)$ and symmetry group $SU(1,1)$ acting as isometries on a $4$-manifold. The other three types are known as Jackiw--Pi, Ambj{\o}rn--Olesen and Bradlow vortices.

It is a natural question to ask whether these last three vortices arise as symmetry reductions of ASDYM equations, like the Taubes and Popov vortices. In Chapter \ref{chapSDYM} we give an affirmative answer to this problem written as Theorem \ref{main_theo}. In fact, Jackiw--Pi vortices arise as symmetry reduction of the ASDYM with gauge group $SU(1,1)$ and symmetry group being the Euclidean group $E(2)$, Ambj{\o}rn--Olesen vortices come from ASDYM with gauge group $SU(1,1)$ and symmetry group $SU(2)$ and finally, Bradlow vortices arise from ASDYM with gauge group $E(2)$ and symmetry group $SU(2)$. The symmetry groups are all isometry groups of surfaces of constant curvature. Theorem \ref{main_theo}	relies on the construction of an ansatz of the gauge field that is equivariant  (invariant up to gauge transformations) under the symmetry group. The procedure for such construction is summarised in Appendix \ref{secConstSym}. The interpretation of these vortices as instantons allows us to determine, in Section \ref{subsecInt}, on which surfaces they are integrable under the twistor point of view. In this case, the vortex equations reduce to a Liouville equation in a similar way as in Witten's work. In Section \ref{sec_sup}, the notion of superposition of vortices is explained, generalising the superposition of Taubes vortices originally derived by Baptista \cite{Baptista2014}. In Section \ref{secEnMod}, we use the result of Theorem \ref{main_theo} to derive a stationary energy functional for these vortices from the non-Abelian pure Yang--Mills energy. A suitable extension of the $4$-dimensional theory to $5$ dimensions allows us to derive a time-dependent theory for the vortices, so that they can be interpreted as stationary solutions to first order equations that also solve the dynamical Euler--Lagrange equations. This dynamical theory provides a natural definition of kinetic energy and thus the moduli space metric can be calculated using Samols' localisation method \cite{Samols1992}.

The problem approached in Chapter \ref{chapVorLike} relates to the other two integrable Taubes vortices described in \cite{Dunajski2012}. A systematic method in deriving these vortices is presented by modifying the original Ginzburg--Landau Lagrangian, which ultimately amounts to allowing the surface's metric to depend on the modulus of the Higgs field. This modified model gives rise to accordingly modified vortex equations whose integrability can be studied by Painlev\'e analysis. The modified model considered here is based on the introduction of a power of the modulus of the Higgs field as a factor of the pure gauge part of the Lagrangian and a suitable modification of the $|\phi|^4$ potential so that the Bogomolny argument \cite{Bogomolny1976} still holds. Under such class of models, all Painlev\'e-integrable cases are derived and it is shown that there are three integrable cases in total, two of them being those originally found in \cite{Dunajski2012}. All of them can be interpreted as ordinary Taubes vortices on surfaces with conical singularity. 

The second half of the thesis concerns projective geometry. Mainly two problems were studied in $2$ dimensions, the metrisability of projective structures and the existence of first integrals of the geodesic equations of a general affine connection. The latter problem is related to the existence of Hamiltonian formulations of hydrodynamic-type systems of two components.

It is well known that, given a metric on a manifold, there exist curves, called geodesics, such that their length locally minimises the distance between $2$ given points. The metrisability problem consists in studying the converse. Given a set of curves on a manifold, when is there a metric whose geodesics are precisely these curves? Here, only the local problem is considered. The metrisability problem is solved by finding a non-degenerate solution to the so called \textit{metrisability equations} \cite{East2008}. These equations are linear and tensorial for $2$-dimensional symmetric contravariant tensors. If a solution is non-degenerate (interpreting it as a symmetric matrix), then a metric that solves the problem can be constructed and such solutions are comprehensive.% and it has been solved in $2$ dimensions \cite{Bryant2009}.

In $2$ dimensions, the set of curves that can be interpreted as geodesics of affine connections is given by solutions of second order ODEs polynomial in first derivatives and involving only cubic powers or less, as explained in Section \ref{secMetProjSt}. Classic examples of such ODEs are those admitting the Painlev\'e property, among which there are the six Painlev\'e equations. In Chapter \ref{chapPainMet}, we study the metrisability problem of the projective structure defined by all Painlev\'e equations. It turns out that, the only Painlev\'e equations that define a metrisable projective structure are those that either admit a first integral that is quadradic in first derivatives or define a projectively flat structure. As explained in Section \ref{secmetP}, the only Painlev\'e equations for which there is a choice of constants such that this happens are III, V and VI. The corresponding metrics all admit a Killing vector, which gives rise to a linear first integral to the parametrised geodesic equations, while the geodesic Hamiltonian gives rise to a quadratic constant of the geodesic flow. By eliminating the affine parameter between them or, equivalently, using Theorem \ref{thm1stintegral}, we recover the above-mentioned quadratic first integrals of the Painlev\'e equations. In Section \ref{SecCoalescence}, we show that the metrics associated to the metrisable cases of equations III and V are related to each other in the same way as III can be recovered from V by coalescence.

It is known that the maximal dimension of the solution space of the metrisability equations in $n$ dimensions is $\frac{(n+1)(n+2)}{2}$ \cite{East2008}. It is conjectured that the maximal dimension of the solution space for degenerate solutions (imposing the constraint that the determinant of the $2$-tensor is zero) is $\frac{n(n-1)}{2}$. For $n=2$ dimensions this was proved in Lemma 4.3 of \cite{Bryant2009}. In Section \ref{secDegSol} we provide some evidence for this conjecture and derive some consequences thereof. In Theorem \ref{conjecture}, this conjecture is proved under the condition that the kernel of the solution to the metrisability equations (which is non-trivial by hypothesis) is fixed. It is also shown, in the same Section, that this maximal dimension is attained by a Newtonian projective structure. An intriguing feature of the kernel of a degenerate solution is presented in Theorem \ref{thmintdist}, where it is proven that any such kernel defines an integrable distribution and that the corresponding integral manifolds are totally geodesic, naturally defining an induced projective structure on them. Equivalently, the image of any degenerate solution (interpreted as a symmetric matrix) defines an integrable distribution and spans the integral manifold.

Chapter \ref{chapterKilling} concerns the problem of local existence of Killing $1$-forms for a given $2$-dimensional affine connection and its relation with hydrodynamic-type systems of two components. In Theorem \ref{theo_1} it is shown, using Frobenius Theorem, that the necessary and sufficient conditions for an affine connection to admit a Killing $1$-form are the vanishing of two scalar quantities. The existence of precisely two Killing forms involves the non-vanishing of the anti-symmetric part of the Ricci tensor and the vanishing of a $2$-tensor and, finally, the existence of three Killing forms (which is the maximal number) is equivalent to the projective flatness of the projective structure defined by the connection along with the symmetry of its Ricci tensor. For special connections, that is to say with symmetric Ricci tensor, the scalar conditions become the Liouville invariants $\nu_5$ and $w_1$. We are then able to find a geometric proof of a result of Liouville \cite{Liouville1886} that says that if $\nu_5=w_1=0$, then there exists a change of coordinates such that the ODE associated to the projective structure does not involve first derivatives (c.f. Theorem \ref{thmLiouInv}). The proof is essentially the determination of coordinates such that the affine connection defined by Thomas symbols is exact. As explained in Corollary \ref{corMet}, the derivation of the scalar conditions and the fact that degenerate solutions to the metrisability problem give rise to Killing forms of special connections, along with the invariants in \cite{Bryant2009}, gives a complete list of invariants that determine whether a projective structure is metrisable, without needing to calculate its solutions. Proposition \ref{prop_rank1} gives the form of all possible connections admitting exactly two first integrals up to change of coordinates.

Theorem \ref{theo_2} relates the existence of Killing forms of particular types of affine connections to the existence of Hamiltonian structures of hydro\-dynamic-type systems. This relation, along with the scalars and $2$-tensor mentioned in the previous paragraph gives a solution to the problem of finding necessary and sufficient conditions for a hydrodynamic-type system of two components to admit $1$, $2$ or $3$ Hamiltonian structures. 

\chapter{Vortex equations from self-duality}\label{chapSDYM}

%\section{Introduction}\label{Intro1}

The Abelian Higgs model at critical coupling admits static solutions  modeling vortices which neither attract nor repel each other \cite{ManSutbook}. The mathematical  content of the model consists of a Hermitian complex line bundle $L$ over a Riemannian surface $(\Sigma, g_\Sigma)$, together with a $U(1)$ connection $a$ and a complex Higgs field $\phi$ satisfying the Bogomolny
equations
\[
\overline{\mathcal D}\phi=0, \quad f=\omega_\Sigma(1-|\phi|^2).
\]
Here $\omega_\Sigma$ is the K\"ahler form on $\Sigma$, 
$\overline{\mathcal D}=\overline{\p}-ia^{(0,1)}$ is the covariant $\overline{\p}$-operator (the anti-holomorphic part of the covariant derivative $D$ defined by $a$), $f=da$ is the Abelian Maxwell field and $a^{(0,1)}$ is the anti-holomorphic part of $a$. Setting $|\phi|^2=\exp{(h)}$, and solving the first Bogomolny equation for the one-form $a$ and substituting it into the second one -- as in the end of the proof of Proposition \ref{prop2} -- reduces the second equation to the Taubes equation \cite{Taubes1980}
\be
\label{taubes}
\triangle h+2(1-e^h)=0,
\ee
which is valid outside small discs enclosing the logarithmic singularities of $h$ -- the locations
of vortices on $\Sigma$. Here $\triangle$ is the Laplace operator on $(\Sigma, g_\Sigma)$.

In \cite{Manton2017} Manton has considered a two-parameter generalisation of the Taubes equation
\be
\label{manton_eq}
\triangle h-2(C_0-Ce^h)=0.
\ee
The constants $C_0$ and $C$ can be rescaled to $0, 1$ or $-1$, and Manton has argued
that only five combinations lead to non-singular vortex solutions:
%which can be stable
%only for $C\leq 0$
\begin{itemize}
\item $C=C_0=-1$ corresponds to the Taubes equation \cite{Taubes1980}. The magnetic field $f$ decays to zero away from vortex  center.
\item $C=C_0=1$ is the Popov equation \cite{Popov2009, Popov2012}.
\item $C=0, C_0=-1$ corresponds to the magnetic field with constant strength equal to $1$. In \cite{Manton2017} this was called the Bradlow vortex.
\item $C=1, C_0=-1$ is the Ambj\o rn--Olesen vortex. The magnetic field is enhanced away from the position of the vortex \cite{AmbOl1988, Manton2017}.
\item $C=1, C_0=0$ is the Jackiw--Pi vortex equation, which arises in a first order Chern--Simons theory \cite{JackPi1990, HorvZhang2009}. In this case $|\phi|^2$ tends to zero at the position of the vortex and (on non-compact surfaces) at $\infty$.
\end{itemize}

The aim of this chapter is to show (Theorem \ref{main_theo} in Section \ref{sec_ec}) that Manton's equation (\ref{manton_eq}) for all values of $C_0, C$ arises as the symmetry reduction of the anti-self-dual Yang--Mills equations (ASDYM) on a four-manifold $M=\Sigma\times N$, where $N$ is a surface of constant curvature, and the symmetry group is the group of local isometries of $N$. The value of $C_0$ in (\ref{manton_eq}) is determined by the curvature of $N$ (in fact, the Gauss curvature of $N$ is $-C_0$) and $C$ depends on the choice of the gauge group $G_C$. We shall demonstrate that $N=S^2$ if $C_0=-1$, $N=\HH^2$ if $C_0=1$, and $N=\R^2$ if $C_0=0$. The gauge group is $SU(2)$ if $C=-1$, $SU(1, 1)$ if $C=1$ and the Euclidean group $E(2)$ if $C=0$. In the integrable cases the Gaussian curvatures of $\Sigma$ and $N$ add up to zero. It is worth pointing out that similar symmetry reductions were performed in \cite{MasonWoodBook}, although not in the context of vortices.

In Section \ref{sec_sup} we shall show how the five vortex equations are related by a superposition principle which leads to a construction of vortices with higher vortex numbers.

The four-dimensional perspective allows us to derive a canonical expression for the resulting energy of vortices. By considering the kinetic energy of the dynamical Yang--Mills theory on $\R\times M$ in Section \ref{sec_mec}, we shall derive integral expressions for moduli space metrics corresponding to various choices of constants in (\ref{manton_eq}). If the gauge group is non-compact, then the kinetic energy and the resulting moduli space metric are not positive definite and the moduli space may contain surfaces where they identically vanish. In the integrable cases of equation (\ref{manton_eq}), the moduli space metric takes a simple form that generalises the one for integrable Taubes vortices described in \cite{Strachan1992}. These integrable cases correspond to $M$ being Ricci-flat -- that is to say, the Gauss curvature of $\Sigma$ being $C_0$ -- in which case the ASDYM equations are integrable under the twistorial point of view. %This chapter is based on my paper with collaboration with my supervisor \cite{ConDun2017}.
%We shall show that there exist surfaces in the moduli space of Jackiw--Pi and Popov vortices which are totally null.
% -- the corresponding vortices
%do not move, at least in the moduli space approximation.

\section{Equivariant anti-self-dual connections and symmetry reduction}
In this Section we shall formulate the main theorem. Let us first introduce
some notation.
\subsection{The group $G_C$} \label{secgroup}
The key role will be played by a Lie group $G_C\subset SL(2, \C)$ which consists of matrices $K$ such that
\[
K\left(\begin{array}{c c}
  1 & 0 \\ 
  0 & -C
 \end{array}\right)
K^\dagger=
\left(\begin{array}{c c}
  1 & 0 \\ 
  0 & -C
\end{array}\right), \quad\mbox{where}\quad C\in\R,
\]
or, equivalently,
\be
\label{group_c}
G_C=\left\{K=\left(\begin{array}{c c}
  k_1 & k_2 \\ 
  C\overline{k_2} & \overline{k_1}
 \end{array}\right);k_1, k_2\in\mathbb C,\quad
 \text{and}\quad |k_1|^2-C|k_2|^2=1\right\}.
\ee
Therefore $G_{-1}=SU(2), G_{1}=SU(1, 1)$ and $G_0=E(2)$ -- the Euclidean group. The generators of the corresponding Lie algebra $\mathfrak{g}_C$,  
\be
\label{lie_alg_rep}
J_1=\frac{1}{2}\left(\begin{array}{c c}
  0 & i \\ 
  -Ci & 0
 \end{array}\right),\quad 	
 J_2=\frac{1}{2}\left(\begin{array}{c c}
  0 & -1 \\ 
  -C & 0
 \end{array}\right),\quad J_3=\frac{1}{2}\left(\begin{array}{c c}
  i & 0 \\ 
  0 & -i
 \end{array}\right),
\ee
satisfy commutation relations
\[
\left[J_1 ,J_2 \right]=-C J_3\; ,\quad \left[J_2, J_3\right]=J_1\; ,\quad \left[J_3 ,J_1 \right]=
J_2.
\]
The explicit parametrisation of $G_C$ as well as the left-invariant one-forms are constructed in Appendix \ref{appeGroup}.	

In what follows, $G_{C_0}$ will denote the Lie group defined in the same way, but changing $C$ into $C_0$.

\subsection{The four-manifold} 
Let $M$ be the Riemannian four-manifold  given by the Cartesian product $\Sigma\times N$ with 
a product metric
\be
\label{metric_on_m}
g=g_\Sigma+g_N,
\ee
where $(\Sigma, g_\Sigma)$ is the Riemann surface introduced in the beginning of the chapter, and $(N, g_N)$ is a surface of constant Gaussian curvature $-C_0$. Let $w$ be a local holomorphic coordinate on $\Sigma$, and $z$ be a local holomorphic coordinate on $N$ so that
\[
g_\Sigma=\Omega dw d\ov{w}, \quad \mbox{and}\quad
g_N=\frac{4 dzd\ov{z}}{(1-C_0|z|^2)^2},
\]
where $\Omega=\Omega(w, \ov{w})$ is the conformal factor on $\Sigma$.
The  K\"ahler forms $\omega_\Sigma$ on $\Sigma$  and  $\omega_N$ on $N$ are given by
$$
\omega_\Sigma= \frac{i}{2}\Omega dw\wedge d\ov{w}  , \quad
\omega_N=\frac{2i dz\wedge d\ov{z}}{(1-C_0|z|^2)^2}=2id\beta,
$$
where
\be\label{eqbeta}
\beta=\frac{zd\ov{z}-\ov{z}dz}{2(1-C_0|z|^2)}.
\ee
%Let the  one-form $\beta$ be a $U(1)$ connection with curvature $\omega_N$. Thus, locally we can set
%\[
%\beta=\frac{zd\ov{z}-\ov{z}dz}{2(1-C_0|z|^2)}, \quad \omega_N=2id\beta.
%\] 
We shall choose an orientation on $M$ by fixing the volume form $\mbox{vol}_M=\omega_\Sigma\wedge\omega_N$.
%%%%%%%%

\subsection{Equivariance}
\label{sec_ec}
 Let $G_C$ and $G_{C_0}$ be Lie groups corresponding, via (\ref{group_c}), to two real constants  $C$ and $C_0$. In the theorem below, $G_C$ will play a role of a gauge group -- $G$ in Appendix \ref{secConstSym} -- and $G_{C_0}$ will be the symmetry group -- $S$ in the same Appendix.

Let $\mathcal E\rightarrow M$ be a vector bundle with a connection which, in a local trivialisation, is represented by a Lie-algebra valued one-form $A\in\mathfrak{g}_C\otimes\Lambda^{1}(M)$. The Lie group $G_{C_0}$ is a subgroup of the conformal group on $(M, g)$, and acts on $M$ isometrically by
\begin{equation}\label{actionC}
(w, z)\mapsto \Big(w, \frac{k_1 z+k_2}{C_0\overline{k_2} z+\overline{k_1}}\Big),
\end{equation}
where $\left(\begin{array}{c c}
  k_1 & k_2 \\ 
  C_0\overline{k_2} & \overline{k_1}
 \end{array}\right)$ is supposed to belong to $G_{C_0}$.

This relation will allow us to relate the holomorphic coordinates $z,\overline z$ of $N$ to the coordinates of $G_{C_0}/U(1)$ as described in Appendix \ref{appeGroup} (c.f. equation (\ref{zk_trans})). 

We shall impose the symmetry equivariance condition (\ref{equivdef}) on $A$. Applied more specifically to our case,
\begin{equation}\label{equivariance_onA}
{\mathcal L}_X A= DW,
\end{equation}
where $X$ is any vector field on $M$ generating the action (\ref{actionC}) and $DW\equiv dW-[A, W]$ is the covariant derivative of a $\mathfrak{g}_C$-valued function on $M$.	

In the coordinates $(w, \ov{w}, z, \ov{z})$ introduced above we have
$$
A=A_\Sigma+A_N, \quad \mbox{where}\quad A_{\Sigma}=A_wdw+A_{\ov{w}}d\ov{w}
\quad\mbox{and}\quad A_{N}=A_zdz+A_{\ov{z}}d\ov{z}.
$$

\begin{theo}
\label{main_theo}
Let $A\in \Lambda^{1}(M)\otimes \mathfrak{g}_C$ be $G_{C_0}$-equivariant. Then
%and satisfy the ASDYM equations on $(M, g)$ 
\begin{enumerate}
\item There exists a gauge and a choice of complex structure in $M$ such that
\be
\label{Aansatz}
A=\left(\begin{array}{c c}
-C_0\beta+\frac{i}{2} a & -\frac{i}{1-C_0 z \overline z}\phi d\overline z\\ 
\frac{iC}{1-C_0 z \overline z}\overline\phi dz & C_0\beta-\frac{i}{2} a
\end{array}\right),
\ee
where $\beta$ is defined in (\ref{eqbeta}), $a$ is a $\mathfrak{u}(1)$-valued one-form, and $\phi$ is a complex Higgs field
on $\Sigma$.
\item The ASDYM equations on $(M, g)$ are
\be
\label{gen_bogomolny}
\overline{\mathcal D}\phi=0, \quad f+\omega_{\Sigma}(C_0-C|\phi|^2)=0
\ee
where $f=da$. Equivalently, setting $|\phi|^2=e^h$, 
\be
\label{manton_eq_thm}
\Delta_0 h-2\Omega (C_0-Ce^h)=0, \quad \mbox{where} \quad \Delta_0=4\p_w\p_{\ov{w}}.
\ee
\end{enumerate}
\end{theo}
We shall split the proof of this theorem into two Propositions
\begin{prop}
\label{prop1}
Let $G_{C_0}$ be the maximal group of isometries of $(N, g_N)$, where $N=\R^2, S^2$ or $\HH^2$. The most general $G_{C_0}$-equivariant $G_C$-connection is gauge equivalent to (\ref{Aansatz}), up to a choice of complex structure in $M$.
\end{prop}
\noindent
and
\begin{prop} 
\label{prop2}
The ASDYM equations on (\ref{Aansatz}) are equivalent to (\ref{gen_bogomolny}) or (\ref{manton_eq_thm}).
\end{prop}
\noindent
{\bf Proof of Proposition \ref{prop1}.}

In this proof, we use the coordinates $(\kappa_1,\kappa_2,\kappa_3)$ of $G_{C_0}$ as described in Appendix \ref{appeGroup} and the technique described in Appendix \ref{secConstSym}.

First, let us introduce $\mathfrak{g}_C$-valued Higgs fields $\Phi_1$, $\Phi_2$ and $\Phi_3$, that can be interpreted as the  should satisfy the constraint equations (\ref{AppeConstEq}),
\begin{align}
&[A_w,\Phi_3]=0,\label{consteq1} \\
&[A_{\overline w},\Phi_3]=0 \label{consteq2}\\
&f_{m3p}^{C_0}\Phi_p^s+f_{rqs}\Phi_m\Phi_3^q=0, \label{consteq3}
\end{align}
where $\Phi_p^s$ are the components of the Higgs fields defined by $\Phi_m=\Phi_m^s J_s$ and only depend on $(y,\overline y)$ and $A_w$ and $A_{\overline w}$ are the $\Sigma$-components of the gauge field and only depend on $(w,\overline w)$. The $\Phi_3$ components are assumed to be constants.

To solve these equations, choose a gauge such that $\Phi_3^1=\Phi_3^2=0$, we will justify why this is possible below. Then these equations imply $\left(\Phi_3^3\right)^2=1$, that is to say $\Phi_3^3=\epsilon$, where $\epsilon=\pm 1$. Then, if we define $\phi_1\equiv\Phi_1^1$ and $\phi_2\equiv\Phi_1^2$, we have the following solution
\begin{align*}
\Phi_1&=\phi_1 J_1+\phi_2 J_2 \\
\Phi_2&=\epsilon\phi_2 J_1-\epsilon\phi_1 J_2 \\
\Phi_3&=\epsilon J_3.
\end{align*}

Let us show that we can indeed always choose $\Phi_3$ to be diagonal. First, if $C=-1$ then this is always possible because hermitian matrices are diagonalisable by unitary change of bases. If $C\geq 0$, begin by performing a gauge transformation to set $\Phi_3^1=0$ and $\Phi_3^2\geq 0$. The solution to the constraint equations will give $\Phi_3^3=\pm \sqrt{1+C \left(\Phi_3^2\right)^2}$. Then perform a gauge transformation by $k_1=-\frac{i}{\sqrt 2} \sqrt{1+\sqrt{1+C \left(\Phi_3^2\right)^2}}$, $k_2=\pm\frac{\Phi_3^2}{\sqrt 2\sqrt{1+\sqrt{1+C \left(\Phi_3^2\right)^2}}}$, where we are using the description of $G_C$ as in (\ref{group_c}). This transformation will make $\Phi_3$ diagonal.
%which will make $\Phi_3$ diagonal.

The constraint equations (\ref{consteq1}) and (\ref{consteq2}) mean that $A_y$ and $A_{\overline y}$ are gauge potentials with gauge group being the little group of $\Phi_3=J_3$, which is the $U(1)$ generated by $J_3$ itself. Now, the only type of gauge transformations allowed are those in this $U(1)$ subgroup. The solution to those equations gives the $w$-components of the gauge potential $A_w=a_w J_3$, $A_{\overline w}=a_{\overline w} J_3$, for some functions $a_w$ and $a_{\overline w}$ which will transform as the components of a $U(1)$-gauge field on $\Sigma$.

With the Higgs fields above, we define a gauge potential on $G_{C_0}$ by
\begin{align*}
A_{\kappa_3}&=\Phi_m\xi_{m\,\kappa_3}=-\epsilon J_3\\
A_{\kappa_1}&=\Phi_m\xi_{m\,\kappa_1}=(\phi_1\sin\kappa_3-\epsilon\phi_2\cos\kappa_3)J_1+(\phi_2\sin\kappa_3+\epsilon\phi_1\cos\kappa_3)J_2\\
A_{\kappa_2}&=\Phi_m\xi_{m\,\kappa_2}=\left[1/\sqrt{-C_0}\sin(\sqrt{-C_0} \kappa_1)\cos\kappa_3\phi_1+\right.\\
&\left.+\epsilon/\sqrt{-C_0}\sin(\sqrt{-C_0} \kappa_1)\sin\kappa_3\phi_2\right]J_1+\nonumber \\
&+\left[1/\sqrt{-C_0}\sin(\sqrt{-C_0} \kappa_1)\cos\kappa_3\phi_2-\right.\\
&\left.-\epsilon/\sqrt{-C_0}\sin(\sqrt{-C_0} \kappa_1)\sin\kappa_3\phi_1\right]J_2+\epsilon\cos(\sqrt{-C_0}\kappa_1)J_3.
\end{align*}

In these expressions, we can set $A_{\kappa_3}=0$ and $\kappa_3=0$ by performing a gauge transformation $g A_\mu g^{-1}+\partial_\mu g g^{-1}$ with $g=diag(e^{\frac{i\epsilon}{2}\kappa_3},e^{-\frac{i\epsilon}{2}\kappa_3})$. We thus recover a gauge field defined now on the quotient $G_{C_0}/U(1)$,
\begin{align*}
A_{\kappa_1}&=-\epsilon\phi_2 J_1+\epsilon\phi_1 J_2\\
A_{\kappa_2}&=\frac{1}{\sqrt{-C_0}}\sin(\sqrt{-C_0} \kappa_1)\phi_1 J_1+\frac{1}{\sqrt{-C_0}}\sin(\sqrt{-C_0} \kappa_1)\phi_2 J_2+\\
&+\epsilon\cos(\sqrt{-C_0}\kappa_1)J_3.
\end{align*}

A direct calculation shows that this gauge potential is indeed $G_{C_0}$-equivariant. Indeed, it satisfies $\mathcal L_{\eta_i}A_\mu=D_\mu W_i$, $i=1,2,3$ where the $\eta_i$'s are left-invariant vector fields on $G_{C_0}$ defined in Appendix \ref{appeGroup} and
\begin{align*}
W_1&=-\epsilon\frac{\sqrt{-C_0}}{\sin(\sqrt{-C_0}\kappa_1)}\cos\kappa_2 J_3^{C},\\
W_2&=-\epsilon\frac{\sqrt{-C_0}}{\sin(\sqrt{-C_0}\kappa_1)}\sin\kappa_2 J_3^{C},\\
W_3&	=0.
\end{align*}

In order to obtain a neater final expression, perform another gauge transformation by $diag(e^{-\frac{i}{2}\epsilon\kappa_2},e^{\frac{i}{2}\epsilon\kappa_2})$ and change into $z$-coordinates to find the $G_{C_0}$-equivariant $G_C$-gauge potential (\ref{Aansatz}) on $M$,
%\begin{equation}\label{Aansatz}
%\mathcal A=\left(\begin{array}{c c}
%  \frac{\epsilon C_0}{2(1-C_0 z \bar z)}(\bar z dz-z d\bar z)+\frac{i}{2}(A_y dy+A_{\bar y} d\bar y) & -\frac{i}{2(1-C_0 z \bar z)}\phi ((1-\epsilon)dz+(1+\epsilon)d\bar z)\\ 
%  \frac{iC}{2(1-C_0 z \bar z)}\bar\phi((1+\epsilon)dz+(1-\epsilon)d\bar z) & -\frac{\epsilon C_0}{2(1-C_0 z \bar z)}(\bar z dz-z d\bar z)-\frac{i}{2}(A_y dy+A_{\bar y} d\bar y)
% \end{array}\right).
%\end{equation}
%where $\phi=\epsilon\phi_2-i\phi_1$ is a $U(1)$-Higgs field.
%\begin{equation}\label{Aansatz}
%A=\left(\begin{array}{c c}
%  \frac{C_0}{2(1-C_0 z \bar z)}(\bar z dz-z d\bar z)+\frac{i}{2}(A_y dy+A_{\bar y} d\bar y) & -\frac{i}{1-C_0 z \bar %z}\phi d\bar z\\ 
%  \frac{iC}{1-C_0 z \bar z}\bar\phi dz & -\frac{C_0}{2(1-C_0 z \bar z)}(\bar z dz-z d\bar z)-\frac{i}{2}(A_y %dy+A_{\bar y} d\bar y)
% \end{array}\right).
%\end{equation}
where $\phi=\phi_2-i\phi_1$ is a $U(1)$-Higgs field and we have chosen $\epsilon=1$.

If we had chosen $\epsilon=-1$, then the resulting form of $A$ would be
$$
A=\left(\begin{array}{c c}
  C_0\beta+\frac{i}{2}a & -\frac{i}{1-C_0 z \overline z}\phi dz\\ 
  \frac{iC}{1-C_0 z \overline z}\overline\phi d\overline z & -C_0\beta-\frac{i}{2}a
 \end{array}\right),
$$
which is equal to (\ref{Aansatz}) up to a change $z\leftrightarrow \overline z, w\leftrightarrow \overline w$, which corresponds to a change in the complex structure. Here, as before, $\beta$ is defined in (\ref{eqbeta}).
\koniec
\noindent
{\bf Proof of Proposition \ref{prop2}.}
Let $(w, z)$ be holomorphic coordinates on $M$. The basis of self-dual two forms
is spanned by \[
\Re(dw\wedge dz),\quad \Im(dw\wedge dz), \quad\omega_N+\omega_\Sigma.
\]
The ASDYM equations $(\omega_N+\omega_\Sigma)\wedge F=0$, and
$dw\wedge dz\wedge F=0$ take the form
\be
\label{asdym}
F_{wz}=0, \quad F_{\ov{wz}}=0, \quad \Omega^{-1} F_{w\ov{w}}+\frac{(1-C_0|z|^2)^2}{4}F_{z\ov{z}}=0.
\ee
%where $\phi=\phi_2-i\phi_1$ is a $U(1)$-Higgs field and we have chosen $\epsilon=1$, which we will %assume henceforth. The results for $\epsilon=-1$ are analogous and we will comment on them below.
The components of the gauge field $F_{\mu\nu}=\partial_\mu A_\nu-\partial_\nu A_\mu - [A_\mu, A_\nu]$ are given by
\begin{align*}
F_{z\overline z}&=\frac{-C_0+C\phi \overline\phi}{(1-C_0 z \overline z)^2}\,\sigma_3\,, \;\;\;&\;\;\; 
F_{w\overline w}&=\frac{i}{2}f_{w\overline w}\,\sigma_3 \,, \nonumber\\
F_{\overline z\overline w}&=\frac{i}{1-C_0 z \overline z}D_{\overline w}\phi\,\sigma_+\,,  \;\;\;&\;\;\
F_{z\overline w}&=-\frac{i}{1-C_0 z\overline z}D_{\overline w}\overline \phi\,\sigma_-\,, \nonumber\\
F_{zw}&=-\frac{i}{1-C_0 z \overline z}D_{y}\overline\phi\,\sigma_- \,,\;\;\;&\;\;\ 
F_{\overline z w}&=\frac{i}{1-C_0 z \overline z}D_{w}\phi\,\sigma_+\,,
\end{align*}
where $\phi=\phi_2-i\phi_1$, $f_{w \overline w}=\partial_w a_{\overline w}-\partial_{\overline w}a_w$, $D$ is the covariant derivative with respect to the $U(1)$-connection $a$ and
$$
\sigma_3=\left(\begin{array}{cc} 1 & 0 \\ 0  & -1 \end{array} \right)\,,\;\;\;\; \sigma_+=\left(\begin{array}{cc} 0 & 1 \\ 0  & 0 \end{array} \right)\,,\;\;\;\; \sigma_-=\left(\begin{array}{cc} 0 & 0 \\ C  & 0 \end{array} \right).
$$
Set
\[
D\phi =d\phi -i a\phi, \quad D\ov{\phi}=d\ov{\phi}+ia\ov{\phi}
%D_\mu\phi=(\partial_w-i a_\mu)\phi, D_\mu\overline\phi=(\partial_w+i a_\mu)\overline\phi, \quad
%\mu=w,\overline w
\]
and
\[
\mathcal D=dw\otimes(\p_w-ia_w), \quad \ov{\mathcal D}=d\ov{w}\otimes(\p_{\ov{w}}-ia_{\ov{w}}), \quad
\mbox{so that} \quad D=\mathcal D+\ov{\mathcal D}.
\]
The ASDYM equations (\ref{asdym})
%$$
%\mathcal F_{zy}=\mathcal F_{\bar z\bar y}=0 \quad \text{ and } \quad g^{y\bar y}\mathcal F_{y\bar %y}+g^{z\bar z}\mathcal F_{z\bar z}=0
%$$
yield vortex-type equations
\begin{align}
D_{\overline w}\phi=\partial_{\overline w}\phi-ia_{\overline w}\phi=0,\label{vortex1}\\
\frac{i}{2} f_{w\overline w}+\frac{\Omega}{4}(-C_0+C \phi \overline\phi)=0.\label{vortex2}
\end{align}
%This is a linear system of PDE's that is usually reduced to a single second order PDE. To do so, solve %the first equations for $A_{\bar y}$ in terms of $\phi$, calculate $F_{y,\bar y}$ using the fact that %$A_y=\overline{A_{\bar y}}$ and substitute in the second equation. The final result is 
%$$
%\Delta_0 h + 2\Omega(-C_0+C e^h)=0,
%$$
%where $h=\log|\phi|^2$. This equation is known as the ``Taubes equation" when $C_0,C<0$.
%[STATE THE RESULTS FOR $\epsilon=-1$]
This system of non-linear PDEs can be reduced to a single second order equation
for one scalar function. In fact, solve the first equation (\ref{vortex1}) for $a_{\overline w}$ so that $a_{\overline w}=-i\partial_{\overline w}\ln(\phi)$ and, using the reality of $a$, $a_w=i\p_{w}\ln(\ov{\phi})$. Using these expressions for the components of $a$, calculate the Abelian Maxwell field $f_{w\overline w}$ and the second equation (\ref{vortex2}) yields (\ref{manton_eq_thm}).
\koniec

\subsection{Integrable cases}\label{subsecInt}

Following the integrability dogma \cite{MasonWoodBook,DunajskiBook,Cald2014}, a symmetry reduction of ASDYM is integrable if the ASDYM equations are defined on a background $(M, g)$ with anti-self-dual Weyl curvature. Computing the Weyl tensor of (\ref{metric_on_m}) shows that conformal anti-self-duality is equivalent to the vanishing of  the scalar curvature of $g$. Thus in the integrable cases the Riemann surface $(\Sigma, g_\Sigma)$ on which the vortex equations are defined must have constant Gaussian curvature equal to minus the Gaussian curvature of $(N, g_N)$, i.e. locally,
\be\label{metricsM}
g_\Sigma=\frac{4 dwd\ov{w}}{(1+C_0|w|^2)^2}, \quad
g=\frac{4 dwd\ov{w}}{(1+C_0|w|^2)^2} + \frac{4 dzd\ov{z}}{(1-C_0|z|^2)^2}.
\ee
The local solutions of integrable vortex equations are given explicitly, in a suitable gauge, by \cite{Manton2017}
\be
\label{int_vortex}
\phi=\frac{1+C_0|w|^2}{1+C|s(w)|^2}\frac{ds}{dw},
\ee
where $s=s(w)$ is a holomorphic map from $\Sigma$ to a surface of curvature $C$. The vortices are located at zeros of $\phi$, which are the zeros of $ds/dw$ and the poles\footnote{Notice that $|\phi|^2$ is invariant under $s\mapsto 1/s$.} of $s$ of order at least $2$.

The integrable cases on simply-connected Riemann surfaces under the anti-self-duality framework are the following:
\begin{itemize}
\item The Taubes vortex ($C=C_0=-1$) is integrable on $\HH^2$, in which case it is a symmetry reduction of ASDYM on $\HH^2\times S^2$. In this case, $s$ is a Blaschke function \[
s(w)=\frac{(w-c_0)\dots(w-c_{\bf N})}{(w-\overline c_0)\dots(w-\overline c_{\bf N})},\] 
where $|c_k|<0$. This is the original integrable reduction of Witten \cite{Witten1977}.
\item The Popov vortex ($C=C_0=1$) is integrable on $S^2$, in which case it is a symmetry reduction from $S^2\times \HH^2$. In this case, $s:\CP^1\rightarrow\CP^1$ is a rational function $p(w)/q(w)$, where $p$ and $q$ are polynomials of the same degree with no common root.
\item The Bradlow vortex ($C=0$, $C_0=-1$) is integrable on $\HH^2$, in which case it is a symmetry reduction from $\HH^2\times S^2$.% In this case, $s$ is a holomorphic function from $\HH^2$ to $\mathbb R^2$ preserving the boundaries.
\item The Ambj\o rn--Olesen vortex ($C=1$, $C_0=-1$) is integrable on $\HH^2$, in which case it is a symmetry reduction from $\HH^2\times S^2$.% In this case, $s$ is a holomorphic function from $\HH^2$ to $S^2$.
\item The Jackiw--Pi vortex ($C=1$, $C_0=0$) is integrable on $\mathbb R^2$, in which case it is a symmetry reduction from $\mathbb R^2\times \mathbb R^2$. %In this case, $s$ is a holomorphic function from $\mathbb R^2$ to $S^2$.
\end{itemize}
In each case the symmetry group is $G_{C_0}$ and the gauge group is $G_C$.
%There are three cases to consider. If $C_0=0$ then $M=\R^p\times T^{4-p}$ with the flat metric and
%$p=0, \dots, 4$, and
%the ASDYM equations with the gauge group $SU(1, 1)$ yield the Jackiw--Pi vortex equation.
%If $C_0=1$ then $M=S^2\times \HH^2$ with the gauge group $SU(1, 1)$ gives the Popov vortex equation. Finally
%$C_0=-1$ corresponds to  $M=\HH^2\times S^2$ and the gauge group $SU(2)$ gives
%the Taubes equation. This is the original integrable reduction of Witten 
%\cite{Witten1977}.

%\begin{itemize}
%\item If $C_0=C=1$ then $s:\HH^2\rightarrow \HH^2$  is the Blaschke function
%\[
%s(w)=\frac{(w-c_1)\dots(w-c_{N+1})}{(1-\ov{c}_1w)\dots (1-\ov{c}_{N+1} w)}, \quad\mbox{where}
%\quad |c_k|<1.
%\]
%\item If $C_0=C=1$ then $s:\CP^1\rightarrow \CP^1$ is a rational function $p(w)/q(w)$, 
%where the polynomials $p, q$ have no common root.
%\item If $C_0=0$ and $C=1$, then $s$ can be taken to be a meromorphic function on a torus.
%\end{itemize}
These integrable cases of (\ref{manton_eq}) do not exhaust the list of all integrable vortices: there are other integrable cases related to the sinh-Gordon and the Tzitzeica equations \cite{Dunajski2012, ConDor2015, Contatto2017}.
%%%%%%%%%%%%%%%%%%%%%%%%%%%%%%%%%%%%%%%%%%%%%%%%%%%%%%%%%%%%%%%%%%%%%%%%%%%%%%%%%%%%%%
\subsection{Superposition of vortices}
\label{sec_sup}
Given a solution to the vortex equation (\ref{manton_eq}) define the vortex number to be
\be
\label{vortex_number}
{\bf N}=\frac{1}{2\pi}\int_{\Sigma} f.
\ee
This is an integer equal to the first Chern number of the vortex line bundle $L\rightarrow \Sigma$,
and we shall assume that this integer is non-negative. 

Let us  now explain why there exist only five vortex equations among the nine possible combinations of values of $C$ and $C_0$. Equation (\ref{vortex1}) implies that 
the vortex number ${\bf N}$ coincides with the number of zeros of $\phi$ counted with multiplicities \cite{ManSutbook}. Since $\phi$ is holomorphic, ${\bf N}$ is necessarily non-negative. Since 
${\bf N}$ is proportional to the magnetic flux, the magnetic field $B\equiv-2if_{w\overline w}/\Omega$ cannot be negative everywhere on $\Sigma$, but  (\ref{vortex2}) implies that this is only possible for the choice of constants 
\[(C,C_0)=(-1,-1), (0,-1), (1,-1), (1,0), (1,1), (0, 0),
\]
where the last possibility means that the magnetic field is null everywhere away from the zero of the Higgs field and (\ref{manton_eq}) is the Laplace equation. We shall call this the 
Laplace vortex.

The resulting six equations are not disconnected from one another. We shall  show that it is possible to construct higher order vortex solutions of one type by superposing two other types of vortices.  
Let $h$ be a vortex solution on $\Sigma$ with vortex number ${\bf N}$ satisfying 
\be
\label{bap_eq1}
\Delta_0 h +2\Omega (-C_0+C_1 e^h)=0
\ee
away from the vortex locations, so that $|\phi|^2=e^h$ has ${\bf N}$ isolated zeros, counting multiplicities. 
We say that this vortex is of type $(C_1, C_0)$.

Consider a metric
on $\Sigma$
\[
\tilde{g}_{\Sigma}=e^h g_{\Sigma}
\]
which has conical singularities at zeros of $|\phi|^2$.
Let  $\tilde h$ be a vortex solution with vortex number ${\bf \widetilde N}$ on 
$(\Sigma, \tilde{g}_\Sigma)$
satisfying 
\be
\label{bap_eq2}
\Delta_0 \tilde h +2 e^h \Omega (-C_1+C e^{\tilde h})=0,
\ee
so that  $|\phi|^2=e^{\tilde h}$ has ${\bf \widetilde N}$ zeros. We shall 
say that this vortex is of type $(C, C_1)$ with a rescaled metric. Adding both 
PDEs (\ref{bap_eq1}) and (\ref{bap_eq2}), we find
that 
\[
\Delta_0 (h+\tilde h) +2\Omega (-C_0+C e^{h+\tilde h})=0
\]
and that $e^{h+\tilde h}$ has ${\bf N}+\bf {\widetilde N}$ zeros. 
Therefore the superposition of a vortex of type $(C, C_1)$ with a rescaled metric on a vortex of type $(C_1, C_0)$ gives rise to a vortex of type $(C, C_0)$ with vortex number being the sum of the first two vortex numbers,
${\bf N}+{\bf \widetilde{N}}$. 
In the case $C=C_0=-1$ this is the Baptista superposition rule \cite{Baptista2014}.

Taking into account the six possible vortex equations, we make a list of all possible superpositions
\begin{align*}
\text{Taubes}+\text{Taubes}&=\text{Taubes}\\
\text{Bradlow}+\text{Taubes}&=\text{Bradlow}\\
\text{Laplace}+\text{Bradlow}&=\text{Bradlow}\\
\text{Ambj\o rn--Olesen}+\text{Taubes}&=\text{Ambj\o rn--Olesen}\\
\text{Jackiw--Pi}+\text{Bradlow}&=\text{Ambj\o rn--Olesen}\\
\text{Popov}+\text{Ambj\o rn--Olesen}&=\text{Ambj\o rn--Olesen}\\
\text{Popov}+\text{Jackiw-Pi}&=\text{Jackiw--Pi}\\
\text{Jackiw--Pi}+\text{Laplace}&=\text{Jackiw--Pi}\\
\text{Popov}+\text{Popov}&=\text{Popov},
\end{align*}
where the non-commutative summation $+$ means that the first vortex (of type $(C, C_1)$ with a rescaled metric) is superposed on the second one (of type $(C_1, C_0)$) to result in the vortex on the right hand side of the equality (of type $(C, C_0)$) with higher vortex number.
%%%%%%%%%%%%%%%%%%%%%%%%%%%%%%%%%%%%%%%%%%%%%%%%%%%%%%%%%%%%%%%%%%%%%%%%%%%%
\section{Energy and moduli space metric}\label{secEnMod}
The energy functional $E$ of pure Yang--Mills theory on $M$ can be reduced to the energy function of an Abelian-Higgs model on $\Sigma$ using the ansatz (\ref{Aansatz}). This can be done by direct calculation,
\begin{align*}
E&=-\frac{1}{8\pi^2}\int_{M}\tr(F\wedge \star_g F)\\
&=\frac{1}{4\pi^2}\int_{N} \omega_N\int_\Sigma
\left[ \frac{1}{4}\left(B^2+(C_0-C\phi\overline\phi)^2\right)-\frac{C}{\Omega}
\left(\left|\mathcal D\phi\right|^2+\left|\overline{\mathcal D}\phi\right|^2\right)\right]\omega_\Sigma
%-\frac{1}{\Omega^2}f_{w\overline w}^2-
%\frac{C}{\Omega}\left(D_{\overline w}\overline\phi D_w\phi+D_{\overline w}\phi D_w\overline\phi\right)+
%\frac{1}{4}\left(-C_0+C\phi\overline\phi\right)^2\right]\omega_\Sigma.
%\frac{1}{16\pi^2}\int_{N} \omega_N\int_\Sigma\left[-\frac{1}{\Omega^2}f_{w\overline w}^2-
%\frac{C}{\Omega}\left(D_{\overline w}\overline\phi D_w\phi+D_{\overline w}\phi %D_w\overline\phi\right)+\frac{1}{4}\left(-C_0+C\phi\overline\phi\right)^2\right]\omega_\Sigma.
\end{align*}
where \[
B=-2if_{w\ov{w}}/\Omega
\] 
is the magnetic field on $\Sigma$. This expression for the energy is proportional to the one given in \cite{Manton2017}.

If we assume that $N$ is compact\footnote{Notice that this is already the case if $C_0=-1$, as $N=S^2$. Otherwise we assume that the corresponding surfaces are quotiented out by a discrete subgroup of $G_{C_0}$. If $C_0=0$, $N$ is a two-torus $\mathbb T^2$ ($\mathbb R^2$ quotiented out by a lattice) and if $C_0=1$ $N$ is a compact Riemann surface of genus ${\tt g}>1$ ($\mathbb H^2$ quotiented out by a Fuchsian group). The ansatz (\ref{Aansatz}) must then admit
this further discrete symmetry.} then the  first integral is the area of $N$ given by
%If $SU(1,C_0)/U(1)$ is compact, which is the case when  $C_0<0$ (sphere), or $C_0=0$ (quotient by a lattice to get a flat torus) or $C_0>0$ (quotient by a fuchsian group to get a compact Riemann surface of genus $>1$), then the first integral in the last equality is finite and gives the area of the surface
$$
\mbox{Area}_N=\int_{N} \omega_N=\begin{cases}
      \frac{4\pi}{-C_0}, & \text{if}\ C_0<0 \\
      4\pi, & \text{if}\ C_0=0 \\
      \frac{4\pi}{C_0}({\tt g}-1) & \text{if}\ C_0>0,
    \end{cases}
$$ 
where ${\tt g}$ is the genus of $N$ and we normalised the area of the torus ($C_0=0$) to $4\pi$.
%The energy functional thus becomes
%$$
%E=\frac{\mbox{Area}_N}{16\pi^2}\int_\Sigma\left[-\frac{1}{\Omega^2}F_{y\bar %y}^2-\frac{C}{\Omega}\left(D_{\bar y}\bar\phi D_y\phi+D_{\bar y}\phi% D_y\bar\phi\right)+\frac{1}{4}\left(-C_0+C\phi\bar\phi\right)^2\right] dVol_\Sigma,
%$$
%where $dVol_\Sigma=\frac{i}{2} \Omega dy\wedge d\bar y$ is the volume form of $\Sigma$. 
Thus the energy can be written, using the Bogomolny argument along with $[D_w,D_{\overline w}]\phi=-if_{w\overline w}\phi	$ and the integration by parts
(with an additional boundary condition $D\phi=0$ if $\Sigma$ is not compact) as
%in the form
%$$
%E=\frac{\mbox{Area}_N}{64\pi^2}\int_\Sigma\left[\left(-\frac{2 i}{\Omega}f_{w\bar %w}+C_0-C|\phi|^2\right)^2-\frac{8C}{\Omega}|D_{\bar w}\phi|^2+\frac{4i C_0}{\Omega} f_{w\bar %w}+\frac{2Ci}{\Omega}\left(\partial_1(\bar\phi D_2\phi)-\partial_2(\bar\phi %D_1\phi)\right)\right]\omega_\Sigma.
%$$
%The last two terms forms a total derivative which integrates to zero on a compact surface or if we %impose $D\phi=0$ on the boundary of $\Sigma$. Then 
$$
E=\frac{\mbox{Area}_N}{16\pi^2}\int_\Sigma\left[\left(B+C_0-C|\phi|^2\right)^2-\frac{8C}{\Omega}\left|\overline{\mathcal D}\phi\right|^2\right]\omega_\Sigma-C_0\frac{\mbox{Area}_N}{4\pi}{\bf N},
$$
where ${\bf N}$
is the vortex number defined by (\ref{vortex_number}).
%where ${\bf N}=-\frac{i}{\pi}\int_\Sigma f_{y\bar y}\frac{i}{2}dy\wedge d\bar y$ is the vortex number, %which is identified with the first Chern number. 
If the vortex equations (\ref{vortex1}) and (\ref{vortex2}) are satisfied, then the energy is proportional to the vortex number, characterising a non-interacting theory, and this value is the global minimum of $E$ if $C\leq 0$. However, equation (\ref{vortex1}) cannot be naturally derived from this argument if $C=0$. In fact the theory corresponding to this energy functional does not involve any Higgs field in this case even though the symmetry reduction of ASDYM necessarily gives rise to a holomorphic Higgs field satisfying (\ref{vortex1}). This is a counter-example to the principle of symmetric criticality, proved under certain conditions in \cite{Palais1979}.

If $N$ is not compact we can still make sense of {\it energy density} (or energy per unit of area of $N$).

%If the ASDYM equations hold then the
%Yang--Mills energy is proportional to the instanton number $k$.
Originally, $E$ is the energy functional of pure Yang--Mills theory in four dimensions, which under the ASDYM condition is %admits a Bogomolny argument, and if the SDYM equations are satisfied it is given by
$$
E=\frac{1}{8\pi^2}\int_{M}\tr\left(F\wedge F\right)\equiv k,
$$
where $k$ is the instanton number. %which is identified with the second Chern number. 
Comparing both expressions for the energy, we derive a relation between the vortex and instanton numbers,
$$
k=-C_0\frac{\mbox{Area}_N}{4\pi}{\bf N}=(1-{\tt{g}}){\bf N}, \quad {\tt g}=0,1,2,\dots\,.
$$
%in agreement with \cite{Popov2009, Popov2012}.
%%%%%%%%%%%%%%%%%%%%%%%%%%%%%%%%%%%%%%%%%%%%%%%%%%%%%%%%%%%%%%%%%%%%%%%
\subsection{Dynamical theory}
\label{sec_mec}
Yang--Mills instantons on $M$ can be regarded as static solitons on $\R\times M$ with a product
metric $-dt^2+g$. Implementing the symmetry reduction of Theorem \ref{main_theo}, but from five dimensions, leads to vortices on $\Sigma$ interpreted as stationary solitons in a dynamical theory on 
$\R\times\Sigma$. We shall use this approach to find the kinetic term in the total energy 
functional on $\R\times\Sigma$, and use it to read-off the metric on the moduli space
of static vortices. Let $\mathcal F$ be a $\mathfrak{g}_C$-valued Yang--Mills field
on $\R\times M$.
The action functional of pure YM theory with $t$-dependence is
$$
S=-\frac{1}{8\pi^2}\int_{\mathbb R\times M}\tr(\mathcal F\wedge\star_5\mathcal F)=\int_{\mathbb R}\mathcal L dt,
$$
where $\mathcal L$ is defined by the second equality and involves the integral on $M$ alone. Under the symmetry reduction of Theorem \ref{main_theo}, $\mathcal L$ becomes a Lagrangian on $\R\times\Sigma$,
\begin{align*}
\mathcal L= -\frac{\mbox{Area}_N}{4\pi^2}&	\int_\Sigma\Big(\frac{1}{4}(B^2+(C_0-C\phi\overline\phi)^2)-\frac{C}{\Omega}(D_{\overline w}\overline\phi D_w\phi+D_{\overline w}\phi D_w\overline\phi) -\\
&-\frac{1}{\Omega} f_{0w}f_{0\ov{w}}+\frac{C}{2}|D_0\phi|^2\Big)\omega_\Sigma.
%\mathcal L=-\frac{\mbox{Area}_N}{16\pi^2}\int_\Sigma [-\frac{1}{\Omega^2}f_{w\overline w}^2-
%\frac{1}{\Omega}f_{0w}f_{0\overline w}
\end{align*}
%$$
%-\frac{C}{\Omega}\left(D_{\bar w}\bar\phi D_w\phi+D_{\bar w}\phi D_w\bar\phi-\frac{\Omega}{2}D_0\phi %D_0\bar\phi\right)+\frac{1}{4}\left(-C_0+C\phi\bar\phi\right)^2 ] 
%\omega_\Sigma.
The Euler--Lagrange equations, resulting from calculating the variation with respect to $\phi, a_w$ 
and $a_0$, are, respectively,
\begin{align}
D_wD_{\overline w}\overline\phi+D_{\overline w}D_{w}\overline\phi-\frac{\Omega}{2}D_0D_0\overline\phi+\frac{\Omega}{2}(-C_0+C \phi\overline\phi)\overline\phi=0,\label{EL1}\\
-2\partial_{\overline w}\left(\frac{1}{\Omega}f_{w\overline w}\right)+\partial_0f_{0\overline w}-iC(D_{\overline w}\overline\phi \phi-D_{\overline w}\phi \overline\phi)=0,\label{EL2}\\
\partial_wf_{0\overline w}+\partial_{\overline w}f_{0w}+\frac{iC\Omega}{2}\left(\phi D_0\overline\phi-\overline\phi D_0\phi\right)=0.\label{EL3}
\end{align}
The equations resulting from varying $\overline\phi$ and $a_{\overline w}$ are the complex conjugate of equations (\ref{EL1}) and (\ref{EL2}), respectively. The third equation is usually referred to as Gauss' law.

This system of  second order dynamical equations is satisfied by static solutions to the first order vortex equations (\ref{vortex1}--\ref{vortex2}) in the temporal gauge $a_0=0$. In fact, Gauss' law (\ref{EL3}) is automatically satisfied. To see that (\ref{EL1}) is satisfied, use (the complex conjugate of) (\ref{vortex1}) to write $D_wD_{\overline w}\overline\phi=[D_w, D_{\overline w}]\overline\phi=if_{w\overline w}\overline\phi$ and eliminate $f_{w\overline w}$ with (\ref{vortex2}). Finally, equation (\ref{EL2}) is satisfied upon eliminating $f_{w\overline w}$ in the same way and using $\partial_{\overline w}\left(\phi\overline\phi\right)=
(D_{\overline w}\phi)\overline\phi+\phi D_{\overline w}\overline\phi=\phi D_{\overline w}\overline\phi$.
%which is actually independent of $\epsilon$.

The kinetic energy $T$ can be read off from $\mathcal L$. In the  temporal gauge
when $a_0=0$ it takes the form
%$$
%K=\frac{Area}{16\pi^2}\int_\Sigma\left[F_{0y}F_{0\bar y}-\Omega\frac{C}{2}D_0\phi %D_0\bar\phi\right]dVol_\Sigma
%$$
%In the temporal gauge $A_0=0$, the kinetic energy becomes
\be
\label{kin_en}
T=\frac{\mbox{Area}_N}{8\pi^2}\int_\Sigma\left[\frac{2}{\Omega}\dot a_w \dot a_{\overline w}-C\dot\phi \dot{\overline\phi}\right]\omega_\Sigma,
\ee
where the dots denote $t$-derivatives. This generalises the known kinetic energy for the Taubes vortex ($C,C_0<0$) \cite{Strachan1992}.

In the usual Abelian Higgs model in the critical coupling (yielding Taubes vortices), there is a moduli space ${\mathcal M}_N$ 
of static ${\bf N}$ vortex solutions. All these 
solutions have the same potential energy, so there are no static forces. The moduli space
acquires a metric from the kinetic energy of the theory, and the geodesics of this metric 
model slow motion of ${\bf N}$-vortices. There are several ways to obtain the metric for both 
flat and curved backgrounds \cite{ManSutbook,DorDunMan2013}. One way to proceed 
for the integrable vortices is to assume that the vortex positions %of zeros and poles in the holomorphic function $s(w)$ defining the vortex 
depend on time, and substitute the explicit solution (\ref{int_vortex}) into the kinetic energy (\ref{kin_en}).
This, when quotiented out by the gauge equivalence (which is equivalent to imposing Gauss' law \cite{ManSutbook}), gives a quadratic form on ${\mathcal M}_N$. In case of Taubes vortices the resulting metric is positive definite, but we see that (\ref{kin_en}) is not positive definite if $C>0$, which is the case for Jackiw--Pi, Popov and Ambj\o rn--Olesen vortices. 

Samols derived a localised expression for the metric of the moduli space of Taubes vortices \cite{Samols1992} (see also \cite{Strachan1992} and \cite{KruSpeight2010} for the metric of moduli space of hyperbolic vortices), where the moduli are the vortex positions (or zeros of the Higgs field). The moduli space metric of Taubes vortices with simple zeros $\{W_1,\dots,W_{\bf N}\}$ is, from (\ref{kin_en}),
$$
\frac{Area_N}{8\pi}\sum_{i,j=1}^{\bf N}\left(\Omega(W_k)\delta_{ij}+2\partial_{W_i}b_j\right)dW_i d \overline W_j,
$$
where $b_j=\partial_{\overline w}(h-2\log|w-W_j|)|_{w=W_j,\overline w=\overline W_j}$ and $\Omega(W_i)$ means that the conformal factor $\Omega$ is being evaluated at the point $(w,\overline w)=(W_i,\ov W_i)$.

A similar calculation as the one performed by Samols (cf. also \cite{ManSutbook}) adapted to vortices defined by (\ref{manton_eq}) gives the following result for the metric on the moduli subspace associated to the simple zeros of the Higgs field
\be\label{metrticmoduli}
\frac{Area_N}{8\pi}\sum_{i,j=1}^{\bf N}\left(-C_0\Omega(W_k)\delta_{ij}+2\partial_{W_i}b_j\right)dW_i d\overline W_j,
\ee
where the constant $C_0$ appears because the calculation involves the use of equation (\ref{manton_eq}).

In the integrable case, the background metric is locally given by $g_\Sigma$ in (\ref{metricsM}) and the Higgs field is (\ref{int_vortex}) in a particular gauge. As above, we assume that $|\phi|$ admits only simple zeros and that each of them is as zero of $ds/dw$ and not a pole, which is always the case up to performing the transformation $s\mapsto 1/s$ (which leaves $|\phi|$ invariant). In this case, $b_j$ can be calculated directly and is given by
$$
b_j=C_0 \frac{W_j}{1+C_0\left|W_j\right|^2}+\frac{3}{2}\overline\beta_j,
$$
where $\beta_j=\frac{s_3^{(j)}}{s_2^{(j)}}$ and $s^{(j)}_k=k!\dfrac{d^ks}{dw^k}\big|_{w=W_j}$, $k=0,1,2,3,\dots$, that is to say, the  $s^{(j)}_k$'s are defined by
$$
s(w)=s_0^{(j)}+s_2^{(j)}(w-W_j)^2+s_3^{(j)}(w-W_j)^3+\cdots,
$$
where the linear term is absent because we assumed that $W_j$ is a simple zero of $\dfrac{ds}{dw}$.

Now, the moduli space metric (\ref{metrticmoduli}) of simple vortices can be calculated and is given by
\be\label{metricmoduliint}
\frac{3\, Area_N}{16\pi}\sum_{i,j=1}^{\bf N}\dfrac{\partial \overline\beta_j}{\partial W_i} dW_i d \overline W_j.
\ee
This is in agreement with the formula derived for integrable Taubes vortices in \cite{Strachan1992}. In particular, (\ref{metricmoduliint}) tells us that the metric is zero when $s$ depends only holomorphically on the vortex positions, which is the case for some Popov and Jackiw--Pi vortices. In fact, this is what happens for the ${\bf N}=2$ Popov vortex on $\Sigma=S^2$ corresponding to $s=(w-W_1)^2/(w-W_2)^2$, and to the ${\bf N}=1$ Jackiw--Pi vortex on $\Sigma=\mathbb R^2$ corresponding to $s=(w-W_1)^2$.

%\newpage
%\thispagestyle{empty}
%\vspace*{\fill}
%\newpage
\chapter{Integrable Abelian vortex-like solitons}\label{chapVorLike}

%\section{Introduction}\label{secIntroModVor}
%%%%%%%%

%The Ginzburg--Landau model is a relativistic $U(1)$-gauge theory with a $\phi^4$ potential for the Higgs field, which is relevant in the theory of superconductors near the critical coupling. This model admits topological solitons called vortices stabilised by the topological charge. 
The existence and analyticity properties of Abelian vortices on the plane were largely studied in particular in \cite{JafTaubes,Taubes1980}. Generalised vortices were proposed by Lohe \cite{Lohe1981}, whose model admits other types of potentials at the expense of minimal coupling between the Higgs and gauge fields. In fact, Lohe's model modifies the kinetic term of the Higgs field as well as the potential in such a way that the Bogomolny argument still holds. The existence of generalised vortices under an analytical point of view was established in \cite{LoheVDH1983}. 

Integrability of the Abelian-Higgs model is well known on a hyperbolic background where the general solution can be explicitly described in terms of holomorphic maps \cite{Witten1977,Strachan1992,ManRink2010,MalMan2015} (see also the previous chapter). Moreover, Painlev\'e analysis shows that these are the only cases in which vortices are integrable in the Painlev\'e sense \cite{Shiff1991}. However two more isolated integrable cases of Abelian vortices were found by allowing the background metric to depend on the Higgs field, allowing the Bogomolny equations to be written as sinh-Gordon and Tzitzeica equations \cite{Dunajski2012}, albeit they do not arise from a variational principle approach from the Ginzburg--Landau model. 

In this chapter we present another modified version of the Abelian-Higgs model by coupling the Yang--Mills term of the Lagrangian with the Higgs field through a continuous function denoted by $G(|\phi|)$. Under mild conditions on $G$ and upon a suitable modification of the potential energy, the model admits vortex-like topological solitons from the Bogomolny argument. A choice of coupling function of the form $G(|\phi|)=|\phi|^{q+1}$, where $q\in\mathbb R$, will give rise to a model that includes the usual Abelian-Higgs vortices as a particular case ($q=-1$), but admits further Painlev\'e-integrable cases, including those described in \cite{Dunajski2012}, providing a variational approach to them. We determine all the possible values of $q$ and all possible background metrics yielding integrable models using the Painlev\'e test. The main result is that the only integrable cases with relevance to vortex theory, both for the ODE and for the PDE, correspond to $q=-\frac{1}{3},0,\frac{1}{3}$ on a flat surface and $q=-1$ on a hyberbolic surface.	

An ODE has the \textit{Painlev\'e property} if its general solution does not possess movable critical points. Critical points are singularities around which a function displays multivaluedness. This definition is motivated by the possibility of defining meromorphic functions from ODE's on the complex plane (c.f. \cite{Conte2013} for a short survey on the subject). Non-linear second order ODE's of the form $y''=R(x,y,y')$, where $R(x,y,y')$ is a rational function of $(y,y')$ with coefficients analytic in $x$, that possess the Painlev\'e property were classified \cite{Painleve1900,Painleve1902,Gambier1910,Ince} and there are $50$ of them up to the change of variables
$$
y\mapsto \frac{a(x)+b(x) y}{c(x)+d(x) y},
$$
for analytic functions $a,b,c,d$. This transformation preserves the Painle\-v\'e property. Among these $50$ equations, $44$ of them admit general solutions given by classical functions while the other $6$ are the Painlev\'e equations. The Painlev\'e test is a technique to find necessary conditions for an ODE to have the Painlev\'e property and will be described, in our case, in the following sections.

In section \ref{secModifiedVortices} we present the modified Ginzburg--Landau Lagrangian and the corresponding Bogomolny equations, which can actually be reduced to a single PDE that will be referred to as the modified Taubes equation. On surfaces of revolution, this equation admits a symmetry reduction to an ODE by rotational symmetry around the origin. On a surface of revolution, if the PDE passes the Painlev\'e test then so does the reduced ODE, but nothing guarantees that the converse is true.
%It might be that applying Painlev\'e analysis to the ODE reveals more integrable cases than one would have obtained by doing the Painlev\'e analysis on the PDE before imposing symmetry to the solutions. 
Thus we perform a separate analysis of the reduced ODE in section \ref{secODE} and of the PDE in Section \ref{secPDE}. It turns out that the analysis of the ODE does not reveal any more integrable cases than that of the PDE.

Within the class of models considered, there can only be integrable cases on a hyperbolic background with Gauss curvature $-1/2$ and on a flat surface. On a hyperbolic surface, the only integrable case corresponds to $q=-1$, which is the usual Abelian-Higgs model. The converse was already known \cite{Shiff1991}. On a flat background, there are three integrable models, corresponding to $q=1/3$, $q=0$ and $q=-1/3$. The first two values are equivalent to the isolated integrable vortices of \cite{Dunajski2012} while $q=-1/3$ gives rise to a new solution and completes the list of integrable models under the class we consider. The Bogomolny equation for this case can be written as the Tzitzeica equation. In order to write explicit soliton solutions with finite energy, some boundary conditions apply. It is not obvious how to apply these boundary conditions to solutions of the Tzitzeica equation, but if we restrict to rotationally symmetric solutions and reduce the Tzitzeica PDE to a Painlev\'e III ODE, whose solutions are well known in the asymptotics \cite{Kit1989}, we can impose the right boundary conditions by fixing some parameters of the third Painlev\'e transcendents. Besides its integrability features, these solutions can be interpreted \cite{Dunajski2012} as usual Abelian-Higgs vortices on backgrounds with conical and curvature singularities at the origin. It is worth noticing that even though our model includes the usual Abelian-Higgs model, there are other types of similar integrable vortex equations on different backgrounds \cite{Popov2012,Manton2013,Manton2017} that it does not cover.

%%%%%%%%
\section{Modified Abelian-Higgs model}\label{secModifiedVortices}
%%%%%%%%

We start with the Ginzburg--Landau theory with a modified Lagrangian on a smooth manifold $\mathbb R\times \Sigma$ with Lorentzian  metric $ds^2=dt^2-\Omega(dx^2+dy^2)$,
\begin{equation}\label{modLag}
L=\int\left(-\frac{G(|\phi|)^2}{4}F_{\mu\nu}F^{\mu\nu}+\frac{1}{2}\overline{D_\mu\phi} D^\mu\phi-V(|\phi|)\right)\Omega\, d^2x,
\end{equation}
where $G$ is a continuous function of $|\phi|$ on its domain of definition, $D_\mu=\partial_\mu-i a_\mu$ is the covariant derivative, $F_{\mu\nu}=\partial_\mu a_\nu-\partial_\nu a_\mu$ is the curvature $2$-form of the $U(1)$-connection $a$ and $\Omega=\Omega(x,y)$ is the conformal factor of the metric on $\Sigma$. The space indices will be denoted by $i,j,k\dots$ and range from $1$ to $2$ as $(x^1,x^2)=(x,y)$. We will also use complex coordinates $z=x+iy$ and polar coordinates $z=r e^{i\theta}$, whenever it is convenient.

If the potential is $V(|\phi|)=\frac{1}{8 G(|\phi|)^2}\left(1-|\phi|^2\right)^2$, which differs from the Ginzburg--Landau $\phi^4$ theory by the factor $G(|\phi|)^2$ in the denominator and spontaneously breaks symmetry, then the usual Bogomolny argument can be applied.  In fact, the modified energy functional is
\begin{align}\label{modEn}
E&=\frac{1}{2}\int\left(\frac{G(|\phi|)^2}{\Omega^2}B^2+\overline{D_i\phi} D^i\phi+\frac{1}{4G(|\phi|)^2}\left(1-|\phi|^2\right)^2\right)\Omega d^2x \nonumber\\
&=\frac{1}{2}\int\left[\frac{G(|\phi|)^2}{\Omega}\left(B-\frac{\Omega}{2G(|\phi|)^2}(1-|\phi|^2)\right)^2+|D_{\bar z}\phi|^2+B-\right.\nonumber\\
&\left.-i\left(\partial_1(\bar \phi D_2\phi)-\partial_2(\bar \phi D_1\phi)\right)\right] d^2x\nonumber\\
&=\frac{1}{2}\int\left[ \frac{G(|\phi|)^2}{\Omega}\left(B-\frac{\Omega}{2G(|\phi|)^2}(1-|\phi|^2)\right)^2+|D_{\bar z}\phi|^2\right] d^2x+\pi N,
\end{align}
where $B$ denotes the component $F_{12}$ and $N\equiv\frac{1}{2\pi}\int_\Sigma B$ is supposed to be positive, as the case $N<0$ is analogous. We have used the boundary conditions $|\phi|\to 1$ and $D_i\phi\to 0$ as $z$ approaches the boundary of $\Sigma$ in the last equality. We assume that all terms in the energy functional are integrable so that it is well defined, which is true for the  integrable cases analysed here.

Thus the modified Bogomolny equations are
\begin{align}
&D_{\bar z}\phi\equiv\dzbar\phi-ia_{\bar z}\phi=0\label{modBog1} \\
&B=\frac{\Omega}{2G(|\phi|)^2}(1-|\phi|^2).\label{modBog2}
\end{align}
Eliminating $a_{\bar z}$ in the second equation using the first one, remembering that $a_{\bar z}=\overline{a_z}$, gives the modified Taubes equation,
\begin{equation}\label{modTaubes}
\Delta_0 h+\frac{\Omega}{G(e^{h/2})^2}\left(1-e^h\right)=0,
\end{equation}
where $h=\ln |\phi|^2$ and $\Delta_0=\partial_x^2+\partial_y^2$ is the Laplacian operator. Solving (\ref{modTaubes}) and imposing the boundary conditions above gives rise to vortex-like topological solitons on the surface $\Sigma$ defined by constant time slices with metric
$$
g=\Omega (dr^2+r^2d\theta^2),
$$
and Gauss curvature given by
$$
K_\Sigma=-\frac{1}{2\Omega}\Delta_0\ln\Omega.
$$

Notice that equation (\ref{modTaubes}) should be modified in case $\phi$ has zeros, as this implies the presence of logarithmic singularities for $h$ and the term $\Delta_0 h$ would generate delta functions. In fact, the usual Taubes equation ($G=1$) is often corrected with delta function sources added by hand to take these singularities into account \cite{ManSutbook}, which may occur at a general point and are the coordinates of the moduli space of vortices. This means that these log-singularities are movable in general and we will bypass them in our Painlev\'e analysis using an exponential change of variables $\chi=e^h$ in the following sections.

From now on, as explained in the beginning of the chapter, we will assume that $G(e^{h/2})^2=e^{(q+1)h/2}$, for a general $q\in\mathbb R$ and, in the next sections, study the integrability features of equation (\ref{modTaubes}). With this choice, we are going to impose another two conditions to the Higgs field. First, we require that the Higgs field is non-vanishing except on a finite number $l$ of distinct points $z_1,\dots,z_l$ and secondly that in a neighbourhood of each point $z_i$, there exists $n_i\in\mathbb N^*$ such that 
\begin{equation}\label{eqBCond}
\phi=(z-z_i)^{n_i}\psi_i(z,\bar z),
\end{equation}
where $\psi_i$ is a nowhere vanishing continuous function on the neighbourhood that is differentiable everywhere except possibly at $z_i$.

These conditions are the most natural ones to impose when seeking a generalisation of the Abelian Higgs model. In fact, they are immediately satisfied for the Abelian Higgs model, which can be proved from the existence of smooth solutions to the Bogomolny equations \cite{JafTaubes}; however, smoothness is an excessively strong condition to impose on the solutions of (\ref{modBog1}-\ref{modBog2}) in general. We will justify these conditions in Section \ref{secODE} in order to rule out solutions on smooth surfaces that do not have a similar behaviour as Abelian vortices (c.f. (\ref{solLiou1}-\ref{solLiou3})).

To begin with, we will suppose that $\Sigma$ is a surface of revolution so that the conformal factor is only a function of the radial coordinate, $\Omega=\Omega(r)$, as well as the modulus of the Higgs field, i.e., $h=h(r)$. This reduces (\ref{modTaubes}) to an ODE, that will be analysed in section \ref{secODE}. Then, in section \ref{secPDE}, we perform the analysis of the PDE (\ref{modTaubes}) in general. 
%We will see that the only integrable cases with relevance to vortex theory, both for the ODE and for the PDE, correspond to $q=-\frac{1}{3},0,\frac{1}{3}$ with $\Sigma$ endowed with a flat metric and $q=-1$ with $\Sigma$ being a hyberbolic surface.	

\section{Painlev\'e analysis of the ODE}\label{secODE}
%\cite{ARS1980}
We apply Painlev\'e analysis \cite{ARS1980} to seek choices of $\Omega$ such that equation (\ref{modTaubes}) for $G(e^{h/2})^2=e^{(q+1)h/2}$ is integrable, assuming cylindrical symmetry, that is to say $\Omega=\Omega(r)$ and $h=h(r)$. Because of the logarithmic divergence of $h$ where the Higgs field vanishes, we look instead at the equation for $\chi=e^{h}$, for which the ODE reduced from (\ref{modTaubes}) is
\begin{equation}\label{modTaubeschi}
\chi''-\frac{\chi'^2}{\chi}+\frac{1}{r}\chi'+\frac{\Omega(r)}{\chi^{(q-1)/2}}\left(1-\chi\right)=0.
\end{equation}

In practise, the aim of the analysis is to determine in which cases the general solutions of the ODE can be locally written in the form $\chi=(r-r_0)^p\sum_{j=0}^\infty \chi_j (r-r_0)^j$, where $\chi_j$ are constants, $\chi_0\not=0$, $r_0 > 0$ is arbitrary and $p$ is assumed to be an integer, a hypothesis that will be justified in the next section. The arbitrary constant $r_0$ represents the position of a movable singularity (either of $\chi$ or $1/\chi$), that we expect not to be critical (or multivalued) for the Painlev\'e property to hold. We suppose that $r_0\not=0$ in order to avoid the coordinate singularity at $r=0$ of (\ref{modTaubeschi}). We look for the dominant behaviour by substituting $\chi\sim \chi_0\left(r-r_0\right)^p$ in (\ref{modTaubeschi}). We have to study it differently according to $p$ is positive or negative. In what follows, $\Omega(r_0)$ will need to be expanded in powers of $r_0$ up to the relevant order.

We start by supposing that $p>0$. Balancing of the dominant terms (1st, 2nd and 4th) requires $p=\frac{4}{1+q} > 0$ and $\chi_0^{2/p}=\frac{\Omega(r_0)}{p}$. Since we are dealing with a second order ODE, its general solution should involve two constants of integration. One of them is the arbitrary constant $r_0$ itself, the other one will be a $\chi_{s}$, for some $s\geq 0$. The order $p+s$ in which this second constant appears is called the order of resonance (or the Fuchs index). Upon substituting the above series in (\ref{modTaubeschi}) and balancing all the powers of $r-r_0$, the constants $\chi_j$ should in principle be determined in terms of $r_0$ and $\chi_{s}$. Notice that the constant $\chi_0$ was fixed above in terms of $r_0$, so we can already tell that $s>0$. Moreover, $\chi_{s}$ will be a free parameter if and only if the leading order in which it appears in the expansion of (\ref{modTaubeschi}) involves $\chi_j$ algebraically for some $j>s$, so that $\chi_j$ can be determined in terms of $r_0$ and $\chi_s$, for any value of $\chi_s$. This leading order will necessarily come from the dominant terms (1st, 2nd and 4th) of (\ref{modTaubeschi}). Therefore, we look for the order of the resonance as follows. Keep just these dominant terms and substitute $\chi=\left(\frac{\Omega(r_0)}{p}\right)^{p/2}(r-r_0)^p+\chi_{s}(r-r_0)^{p+s}$. Expand the resulting expression in powers of $r-r_0$, keeping only the leading order of terms involving $\chi_s$, which will clearly be linear in $\chi_s$ (because $s>0$):
$$
(r-r_0)^{p+s-2}(s^2-s-2)\chi_s,
$$
whose vanishing implies $s=-1$ or $s=2$. The second, positive root indicates a resonance at order $p+2$ of the expansion in $r-r_0$, which means that $\chi_2$ can only appear in the coefficient of order $(r-r_0)^{p+1}$ or higher in the expansion of (\ref{modTaubeschi}) when $\chi$ is replaced by the power series. But at these orders the coefficients $\chi_{j\geq 3}$ are present and thus $\chi_2$ is not fixed. 

We are interested in analysing the order of resonance, as this will provide constraints on the geometry of $\Sigma$. Thus, we write $\chi=\left(\frac{\Omega(r_0)}{p}\right)^{p/2}(r-r_0)^p+\chi_1(r-r_0)^{p+1}+\chi_2(r-r_0)^{p+2}$, substitute it in (\ref{modTaubeschi}) and divide by $(r-r_0)^p$ to get
\begin{align}\label{PAeq}
&(r-r_0)^p \left(\frac{-2 (p-1) p^{-\frac{p}{2}}  \Omega (r_0)^{p/2}\chi_1-\Omega (r_0)^{p-1} \Omega '(r_0)}{r-r_0}-\frac{p^{1-p} \Omega (r_0)^p}{(r-r_0)^2}\right)+\nonumber\\
&+\frac{1}{r-r_0}\left(p^{1-\frac{p}{2}} \Omega (r_0)^{\frac{p}{2}-1} \Omega '(r_0)+\frac{p^{1-\frac{p}{2}} \Omega (r_0)^{p/2}}{r_0}-2 \chi_1\right)+\nonumber\\
&+ (p-2)  \Omega(r_0)^{-1} \Omega '(r_0)\chi_1+\frac{p+1}{r_0}  \chi_1-2 (p-1) p^{p/2-1} \Omega (r_0)^{-p/2}\chi_1^2+\nonumber \\ 
&+\frac{p^{-p/2+1}}{2} \Omega(r_0)^{p/2-1} \Omega''(r_0)-\frac{p^{-p/2+1}}{r_0^2}  \Omega(r_0)^{p/2}=0,
\end{align}
up to order of a positive power of $r-r_0$.

Equating the coefficient of each order to zero, we find conditions on $\Omega''(r_0)$ and $\chi_1$. As mentioned above, had we written $\chi$ as an infinite series $\chi=\sum_{n\geq 0}\chi_n (r-r_0)^{p+n}$, we would have been able to calculate recursively $\chi_n\;(n\geq 3)$ in terms of $r_0$ and $\chi_2$. 
 
If $p\geq 2$, we calculate $\chi_1$ from the term of order $(r-r_0)^{-1}$ and then the term of order $(r-r_0)^0$ gives the following equation for the conformal factor
$$
\Delta_0 \ln \Omega(r_0)=0,
$$
which has to be valid for any $r_0$, yielding a differential equation whose solution is $\Omega(r)=c_1 r^{c_2}$, for some constants $c_1>0$ and $c_2>-2$, thus the surface $\Sigma$ is locally flat. We require that $c_2>-2$ since the origin $r=0$ would be at infinite distance from any other point otherwise. In fact, by performing the change of radial variable $R=\frac{2\sqrt c_1}{c_2+2}r^{\frac{c_2+2}{2}}$, the metric becomes $dR^2+\left(\frac{c_2+2}{2}\right)^2R^2d\theta^2$ so that we can set $\Omega=1$ in (\ref{modTaubeschi}) at the expense of introducing a deficit in the angular variable $\theta$, characterising a conical singularity at the origin. We suppose that the background is a smooth manifold and therefore we do not take into account these singularities and suppose $c_2=0$.
%Thus, the equation for which we perform Painlev\'e analysis is 
%\begin{equation}\label{eqKGexp}
%\Delta_0 h+e^{-\frac{q+1}{2}h}\left(1-e^h\right)=0
%\end{equation}

If $p=2$, the term of order $(r-r_0)^{p-2}$ in (\ref{PAeq}) contributes. Thus, the vanishing of (\ref{PAeq}) at orders $(r-r_0)^{-1}$ and $(r-r_0)^0$ implies $\chi_1= \frac{1}{2}\left(\Omega'(r_0)+\frac{\Omega(r_0)}{r_0}\right)$ and 
\begin{equation}\label{eqLiouvEqOmega}
\Delta_0 \ln \Omega(r_0)=\Omega(r_0).
\end{equation}
This equation means that $\Sigma$ has constant Gauss curvature $-1/2$. This is not surprising as for $p=2$, $q=1$, and thus equation (\ref{modTaubes}) is the usual Taubes equation, up to replacing $h$ by $-h$, whose Painlev\'e integrability was studied in \cite{Shiff1991}. Solutions to the modified Taubes equation in this case would involve a Blaschke product but from condition (\ref{eqBCond}) the magnetic field $B$ would not be integrable due to divergences where the Higgs field vanishes and thus we would not be able to define a magnetic flux and the energy would be infinite. Let us however point out that, choosing the solution to (\ref{eqLiouvEqOmega}) to be $\Omega(r)=\frac{4}{(1-r^2)^2}$, this case admits the following solutions for the squared modulus of the Higgs field
\begin{align}
\chi&=\frac{4 r^2 (\ln r)^2}{(1-r^2)^2}, \label{solLiou1}\\
\chi&=\frac{(r^{c+1}-r^{-c+1})^2}{c^2(1-r^2)^2}, \quad 0<c<1,\label{solLiou2}\\
\chi&=\frac{4 r^2 \sin^2(c \ln r)}{c^2(1-r^2)^2}, \quad c>0,\label{solLiou3}
\end{align}
which are not analogous to Abelian vortices on smooth surfaces and can be ruled out by the conditions imposed in the end of Section \ref{secModifiedVortices}. These solutions were obtained from results of \cite{Popov1993} (c.f. also section $5$ of \cite{ConDor2015}).

If $p=1$ then all the terms in (\ref{PAeq}) contribute and we find that the conformal factor should satisfy the following differential equation
\begin{equation}\label{conffactorpeq1}
\Omega''(r_0)-\frac{\Omega'(r_0)^2}{\Omega (r_0)}+\frac{\Omega '(r_0)}{r_0}-\Omega '(r_0) \sqrt{\Omega (r_0)}-\frac{2 \Omega (r_0)^{3/2}}{r_0}=0,
\end{equation}
which can be rewritten in terms of $F(r)=\sqrt{\Omega(r)}$ as
\begin{equation}\label{diffeq}
F''(r)-\frac{F'(r)^2}{F(r)}+\frac{F'(r)}{r}-\frac{F(r)^2}{r}-F(r) F'(r)=0.
\end{equation}
Its general solution is 
\begin{equation}\label{eq1v}
F(r)=\frac{C_1}{r} \frac{(C_2 r)^{C_1}}{1-(C_2 r)^{C_1}},
\end{equation}
where $C_1$ and $C_2$ are arbitrary positive constants, so that the origin $r=0$ is at finite distance from any other point. Under the change of variables $R=(C_2 r)^{C_1/2}$, equation (\ref{modTaubeschi}) becomes
\begin{equation}\label{modTaupeq1}
\dfrac{d^2\chi}{dR^2}-\frac{1}{\chi}\left(\dfrac{d\chi}{dR}\right)+\frac{1}{R}\dfrac{d\chi}{dR}+\frac{4 R^2}{(1-R^2)^2}\frac{1}{\chi}\left(1-\chi\right)=0,
\end{equation}
where we assume that $0\leq R<1$. A one parameter family of solutions to this equation was given in \cite{Shiff1991} (c.f. equations (2.15--2.16) of this reference).
It satisfies the necessary conditions for Painlev\'e property established so far, but its analysis is not finished yet. To complete the Painlev\'e test we follow the usual procedure. Expand $\chi=\chi_0 (R-R_0)+\chi_1 (R-R_0)^2+\chi_2 (R-R_0)^3+\cdots$, substitute it in (\ref{modTaupeq1}) and expand the left hand side in powers of $R-R_0$. The vanishing of the leading order implies
$$
\chi_0=\pm\frac{2 R_0}{1-R_0^2},
$$
where $0\leq R_0 <1$.
The case in which we choose the $+$ sign was already analysed above and led us to equation (\ref{conffactorpeq1}). Now, if we choose the $-$ sign, the vanishing of the new leading term implies 
$$
\chi_1=-3\frac{1+R_0^2}{(1-R_0^2)^2}.
$$
With these choices of $\chi_0$ and $\chi_1$, the left hand side of (\ref{modTaupeq1}) becomes
$$
-\frac{16 R_0}{(1-R_0^2)^3}(R-R_0)+O\left((R-R_0)^2\right),
$$
whose first term cannot be eliminated by any choice of $\chi_i$. This means that the expansion of $\chi$ should involve logarithmic terms of the form $\ln(R-R_0)$. Therefore, equation (\ref{modTaupeq1}) does not pass the Painlev\'e test. Another solution to (\ref{diffeq}) can be obtained by taking the limit $C_1\to 0$ in (\ref{eq1v}) and then equation (\ref{modTaubeschi}) becomes
$$
\dfrac{d^2\chi}{dR^2}-\frac{1}{\chi}\left(\dfrac{d\chi}{dR}\right)+\frac{1}{R}\dfrac{d\chi}{dR}+\frac{ R^2}{(\ln R)^2}\frac{1}{\chi}\left(1-\chi\right)=0,\quad R=C_2 r,
$$
which also fails the Painlev\'e test, as a similar calculation shows.

In the case $p<0$, a similar procedure will lead to the condition $p=-\frac{4}{1-q}$. The conditions for Painlev\'e integrability can be derived from the case $p>0$ above. In fact, under the change of variables $\chi\mapsto 1/\chi$ in equation (\ref{modTaubeschi}), $q$ is changed into $-q$ or, writing this equation in terms of $p$ using $q=\frac{4-|p|}{p}$, $p$ is changed into $-p$. Therefore, the conditions for Painlev\'e integrability in the cases $p=-1$, $p=-2$ and $p\leq -3$ are the same as in the cases $p=1$, $p=2$ and $p\geq 3$, respectively. Namely, for $p=-1$, there are no integrable soliton solutions, for $p=-2$, $\Sigma$ must be a hyperbolic space with constant curvature $-1/2$ and for $p<-2$, $\Sigma$ must be flat up to conical singularities. Even though from the integrability point of view cases $p=-2$ (or $q=-1$) and $p=2$ (or $q=1$) are the same, for $p=-2$ we obtain the ordinary Taubes equations of the Abelian Higgs model, which admits soliton solutions satisfying our conditions as opposed to the case $p=2$.

To complete the integrability analysis we need to have a closer look in the range $|p|\geq 3$, in which case $-1/3 \leq q \leq 1/3$. This is because for this range of $q$, we can find a $p_1>0$ and a $p_2<0$ such that $q=\frac{1}{p_i}\left(4-|p_i|\right),\; i=1,2$. But for Painlev\'e integrability to take place, the integrability conditions should hold for all possible choices of leading order $p$. Therefore, we have to solve 
$$
\frac{1}{p_1}\left(4-p_1\right)=q=\frac{1}{p_2}\left(4+p_2\right),
$$
for integers $p_1>0$ and $p_2<0$. There are exactly three solutions to this equation: $(p_1,p_2)=(6,-3),\;(4,-4)$ and $(3,-6)$, yielding $q=-\frac{1}{3},\;0$ and $\frac{1}{3}$, respectively. Since the cases $q=1/3$ and $q=0$ lead to the models studied in \cite{Dunajski2012}, we will present the explicit vortex solutions to the case $q=-1/3$ in Section \ref{solutions} (c.f. equation (\ref{solMik})), after the Painlev\'e analysis of the PDE.

\section{Painlevé analysis of the PDE}\label{secPDE}

%The Painlev\'e analysis as described above were used to determine the integrable cases for the Abelian Higgs and Chern-Simons vortex equations \cite{Shiff1991}. The only possible integrable case for the former is on hyperbolic surfaces of Gauss curvature $-1/2$. Vortices on hyperbolic geometry were first derived in \cite{Witten1977}, where they were interpreted as $SO(3)$-equivariant instantons in $\mathbb R^4$, and subsequently studied in \cite{Strachan1992,ManRink2010}. In this section we apply the same method to the modified Abelian Higgs model described above. The solutions presented here have interest not only because of their integrability features but also because they have an interpretation as Abelian vortices on surfaces with a conical singularity in the sense of \cite{Dunajski2012}.
We will find all possible choices of $G(e^{h/2})^2=e^{(q+1)h/2}$ and of background metric $\Omega$ such that equation (\ref{modTaubes}) admits the Painlev\'e property, now without imposing any symmetry to the PDE. We will do the analy\-sis using the method proposed by Weiss, Tabor and Carnevale \cite{WTC1983}, which is the analogue Painlev\'e test for PDEs. As in the previous section, in order to avoid the logarithmic singularities in the analysis we look instead at the equation for $\chi=e^{h}$,
\begin{equation}\label{PDEChi}
\Delta_0\chi-\frac{1}{\chi}|\nabla\chi|^2+\frac{\Omega(x,y)}{\chi^{(q-1)/2}}\left(1-\chi\right)=0,
\end{equation}
where $\nabla\chi=(\partial_x\chi,\partial_y\chi)$ is the gradient vector of $\chi$ and $|\nabla\chi|^2=\left(\partial_x\chi\right)^2+\left(\partial_y\chi\right)^2$ is its Euclidean norm.

We look for the dominant behaviour by setting 
$$
\chi\sim\chi_0(x,y)\varphi(x,y)^p,
$$
where $\chi_0$ is a non-zero function to be determined and $p$ is an integer, as justified below. Keeping the lowest order terms in $\varphi$, we find
\begin{equation}\label{eqloworders}
\chi_0^{(1-q)/2}\Omega\,(1-\chi_0\,\varphi^p)\, \varphi^{p(1-q)/2}-p\, \chi_0 \,|\nabla\varphi|^2\, \varphi^{p-2}=0.
\end{equation}
We then need to separate the analysis into two different cases, $p>0$ and $p<0$. 

If $p>0$ then the term in $\varphi^p$ in the parenthesis of (\ref{eqloworders}) is of higher order and can be neglected at this stage. Then we equate the powers of $\varphi$ for the remaining two terms, $p(1-q)/2=p-2$ which gives a relation between $q$ and $p$, which  will be convenient to be solved for $q$:
$$
q=\frac{4-p}{p}.
$$
We solve (\ref{eqloworders}) for $\chi_0$ to find 
$$
\chi_0=\left(\frac{\Omega}{p |\nabla\varphi|^2}\right)^{p/2}.
$$

Anticipating from the ODE analysis above that there will be a resonance at second order, we expand $\chi$ as $\chi=\varphi^p(\chi_0+\chi_1 \varphi+\chi_2 \varphi^2+\cdots)$, substitute it in (\ref{PDEChi}), divide by $\varphi^p$ and expand the whole expression in powers of $\varphi$ up to the first two lowest orders, which are $\varphi^{-1}$ and $\varphi^0$, keeping (if necessary) the terms of order $\varphi^{p-2}$ and $\varphi^{p-1}$. The terms of order $\varphi^{p-2}$ and $\varphi^{p-1}$ are 
\begin{equation}\label{termorderp}
-p^{1-p}\Omega^p |\nabla\varphi|^{2-2p}\varphi^{p-2}
\end{equation}
and
\begin{equation}\label{termorderp1}
-2(p-1)\left(\frac{\Omega}{p}\right)^{p/2} |\nabla\varphi|^{2-p}\chi_1\varphi^{p-1},
\end{equation}
respectively, which arise from the very last term in (\ref{PDEChi}).

These terms will not contribute to the analysis if $p>2$. Therefore, we will separate the analysis into the cases $p=1$, $p=2$ and $p\geq 3$.

If $p=1$, then the vanishing of the term of order $\varphi^{-1}$, which involves (\ref{termorderp}), gives rise to an algebraic equation for $\chi_1$ whose solution is
$$
\chi_1=\frac{\sqrt{\Omega(x,y)}}{2|\nabla\varphi|^3}\left(\Delta_0\varphi-\sqrt{\Omega(x,y)}|\nabla\varphi|\right).
$$
Then the term of order $\varphi^0$ will not depend on $\chi_2$, manifesting the resonance at this order predicted above. Instead, this term is a fairly big expression involving $\varphi$ and $\Omega$ (and their partial derivatives up to second order) that should vanish for any small function $\varphi$. Making the choices $\varphi=\pm\epsilon x$, $\varphi=\pm\epsilon y$ and $\varphi=\epsilon x y$, where $\epsilon$ is a small positive constant, we get differential equations for $\Omega$ that can only be solved by $\Omega=0$. This case is thus not interesting for our purposes.

If $p=2$ then $q=1$ and a change of variables of the form $\chi\mapsto\chi^{-1}$ will put (\ref{PDEChi}) in the form of the usual Taubes equation, whose Painlev\'e analysis requires $\Sigma$ to be a hyperbolic space of curvature $-1/2$ \cite{Shiff1991}. As for the ODE in the previous section, condition (\ref{eqBCond}) implies that the divergence of the magnetic field (\ref{modBog2}) at each zero of the Higgs field would make the magnetic flux infinite, and thus no solution would fit our requirements.

If $p\geq 3$ then the lowest order term is 
$$
\left(\frac{\Omega}{p|\nabla\varphi|^2}\right)^{p/2}\left[p\Delta_0\varphi-2\left(\frac{p}{\Omega}\right)^{p/2}|\nabla\varphi|^{p+2}\chi_1\right]\frac{1}{\varphi}.
$$
The term in brackets should vanish, resulting in an equation for $\chi_1$ which can be solved by
$$
\chi_1=\frac{p^2}{2\Omega}\left(\frac{\Omega}{p|\nabla\varphi|^2}\right)^{(p+2)/2}\Delta_0\varphi.
$$
This choice of $\chi_1$ annihilates the term of order $\varphi^{-1}$ and we are left with the term of order $\varphi^0$  which is
$$
-\frac{p}{2}\Delta_0\ln\Omega \left(\frac{\Omega}{p|\nabla\varphi|^2}\right)^{p/2}.
$$
We notice that it does not involve $\chi_2$, indicating the resonance anticipated earlier. The conformal factor $\Omega$ should then satisfy $\Delta_0\ln\Omega=0$. In other terms, the metric should be flat, up to possible conical singularities. Thus, we can choose local coordinates to set $\Omega=1$ under smoothness assumptions. 

As for the ODE, the conditions for Painlevé integrability in the cases $p=-1$, $p=-2$ and $p\leq -3$ are the same as for $p=1$, $p=2$ and $p\geq 3$, respectively, as we can go from $p$ to $-p$ by changing $\chi$ into $\chi^{-1}$. Therefore, the integrable cases for the PDE correspond to the same as for the ODE, that is to say either $\Sigma$ is a hyperbolic surface of curvature $-1/2$ and $|q|=1$ or $\Sigma$ is flat and $|q|=1/3$. Notice however that for $p=1$ (or $q=1$)  we did not have a soliton solution but for $p=-1$ (or $q=-1$) we find exactly the usual Abelian Higgs model on hyperbolic surfaces, whose solutions are well understood.

Here it is worth pausing to explain why we require $p$ to be an integer. If $p$ is not an integer then the PDE does not admit the Painlev\'e property, however it may be transformed into one having this property under a change of variables replacing $\chi$ by some power of $\chi$, which might reveal further integrability properties. However, once we substitute the series expansion $\chi=\varphi^p\sum_{k\geq 0}\chi_k\varphi^k$ in (\ref{PDEChi}) and divide the left hand side by $\varphi^p$, the resulting expression takes the form
$$
\left(\text{power series in }\varphi\right) -\Omega \varphi^{p-2}\left(\sum_{k\geq 0}\chi_k\varphi^k\right)^{2\frac{p-1}{p}}=0,
$$
and for the second term to vanish for $p$ non-integer while $\chi_0\not=0$ we would require that $\Omega=0$, which is not our interest.

We have done the Painlev\'e analysis by expanding the $\chi$ in power series of $\varphi$. We could have also used the ``reduced ansatz" proposed by M. Kruskal and explained in \cite{WTC1983} which consists in supposing that $\partial_x\varphi\not=0$ and expanding $\chi$ in power series of $x-\psi(y)$, where $\psi$ is a function such that $\varphi(\psi(y),y)=0$ that exists by the implicit function theorem. Even though this ansatz is clearly analogous to the Painlev\'e analysis for ODEs and can simplify calculations considerably, in our case we would have needed to expand $\Omega(x,y)$ in power series of $x-\psi(y)$ with respect to the first variable and thus we decided not to use it. Anyway, similar calculations with this ansatz yield the same results.

\subsection{Explicit solutions}\label{solutions}

For $\Omega=1$ and $q=0$, (\ref{PDEChi}) becomes the sinh-Gordon equation $\Delta_0 \frac{h}{2}=\sinh\frac{h}{2}$ while for $q=\pm \frac{1}{3}$, it becomes the Tzitzeica equation \cite{Mik1979,Mik1981,FordyGib1980}
\begin{equation}\label{Tzieq}
\Delta_0 u+\frac{1}{3}\left(e^{-2u}-e^u\right)=0,
\end{equation}
where $u=q h$. These equations were studied in the context of Abelian vortices in \cite{Dunajski2012}, where the cases considered correspond to $q=\frac{1}{3}$ and $q=0$ in our language. However, the analysis presented here points to a new solution in the case $q=-\frac{1}{3}$ and completes the list of integrable cases under the class of models considered. We will focus on the details of this new solution, bearing in mind that they are analogous for the other two cases.

We still need to apply the boundary conditions so that we can calculate physical quantities such as the energy, magnetic flux and vortex strength. We thus have to know the behaviour of the asymptotics of the solutions to (\ref{Tzieq}). If we apply the cylindrical symmetry reduction $u=u(r)$, supposing that $u$ is only a function of the radial coordinate, (\ref{Tzieq}) reduces to a Painlev\'e III equation with choice of parameters $(1,0,0,-1)$ under the change of variables $u(r)=\ln w(r)-\frac{1}{2}\ln r+\frac{1}{4}\ln \frac{27}{4},\; r=\frac{3\sqrt{3}}{2}\rho^{2/3}$:
$$
\dfrac{d^2 w}{d\rho^2}=\frac{1}{w}\left(\dfrac{dw}{d\rho}\right)^2-\frac{1}{\rho}\dfrac{dw}{d\rho}+\frac{w^2}{\rho}-\frac{1}{w}.
$$
The behaviour of its solutions in the asymptotics were studied in \cite{Kit1989}. We thus apply this reduction and equation (18) in \cite{Kit1989} with $g_1=g_2=0$, $g_3=1$, $\tau=\frac{r^2}{12}$ and $s=1+2\cos\left[\frac{\pi}{9}(6-2N)\right]$ to find
\begin{equation}\label{solMik}
h=-3u \sim_{r\to 0} -3\ln\left[\frac{2\alpha}{9}(N-3)^2 12^{\frac{N}{3}}\frac{r^{-\frac{2N}{3}}}{\left(1-\alpha 12^{\frac{1}{3}(N-3)}r^{-\frac{2}{3}(N-3)}\right)^2}\right],
\end{equation}
where
$$
\alpha=3^{\frac{2}{3}(N-3)}\frac{\Gamma\left(\frac{1}{3}\left(2+\frac{N}{3}\right)\right)\Gamma\left(\frac{1}{3}\left(1+\frac{2N}{3}\right)\right)}{\Gamma\left(\frac{1}{3}\left(4-\frac{N}{3}\right)\right)\Gamma\left(\frac{1}{3}\left(5-\frac{2N}{3}\right)\right)}
$$
and $N$ is the topological charge (or vortex number), which is allowed to take values $N=1$ and $N=2$.

The results below fig. 1 in the same reference gives the behaviour at $r\to\infty$,
\begin{equation}\label{solMikinfty}
h=-3u \sim_{r\to\infty}-\frac{3\sqrt{3}}{\pi}\left\{1+2 \cos\left[\frac{\pi}{9}\left(6-2N\right)\right] \right\}K_0(r),
\end{equation}
where $K_0(r)\sim_{r\to\infty}\sqrt{\frac{\pi}{2r}}e^{-r}$ is the modified Bessel function of second kind. The strength of the vortex can be read off from the coefficient before the Bessel function $K_0$ and takes approximate values $2.23$ and $4.19$ for $N=1$ and $2$, respectively. For comparison, these values are approximately $1.80$ and $1.45$ for the models with $q=0$ and $q=1/3$, respectively, for which only $N=1$ vortex solutions are allowed \cite{Dunajski2012}.

\vspace{-3.5cm}	
\begin{center}
%\vspace{-3.5cm}
  \begin{figure}[ht]
  \begin{tabular}{ll}
\hspace{-2.8cm}
\includegraphics[scale=0.5]{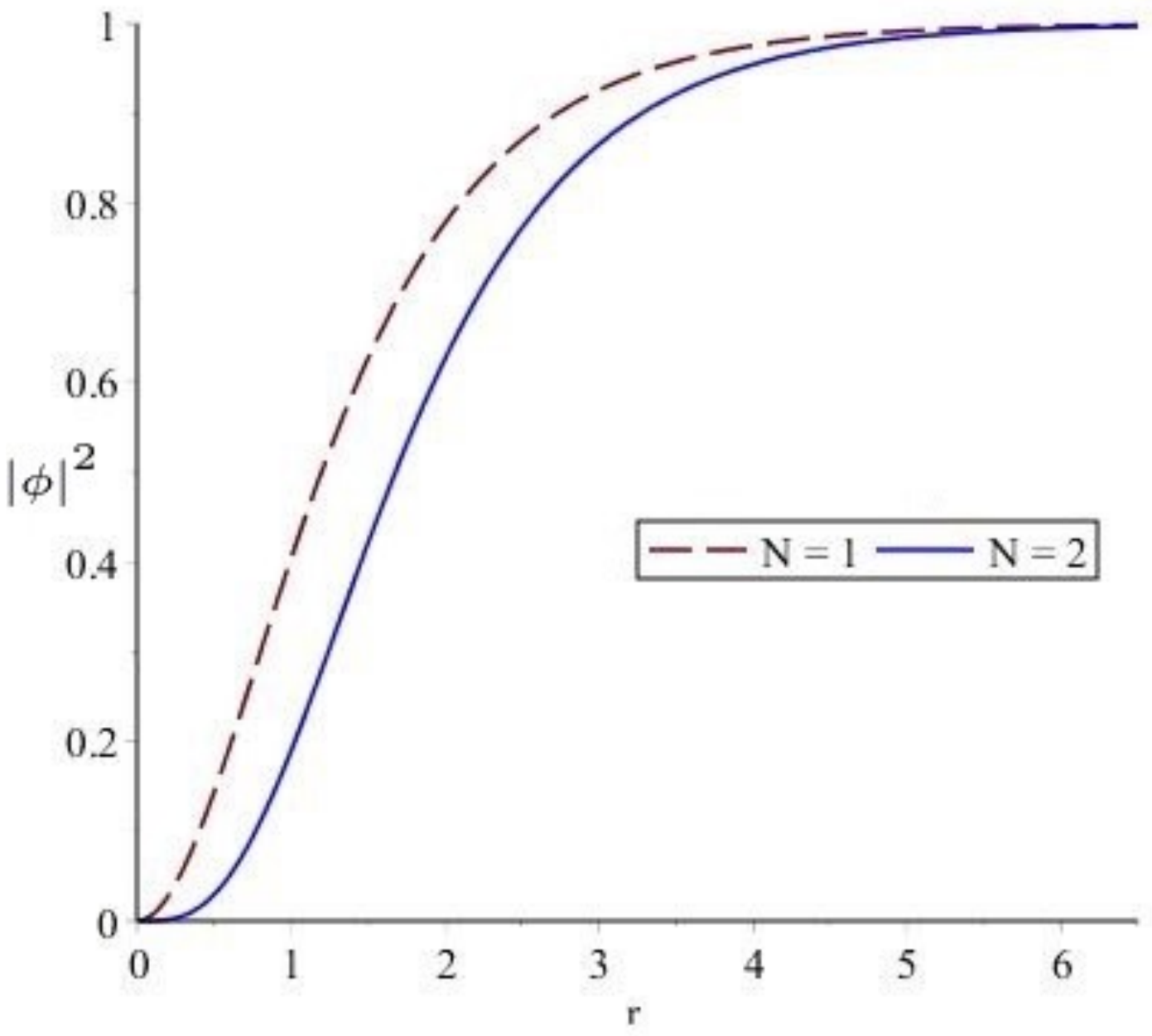} &\hspace{-2cm} \includegraphics[scale=0.5]{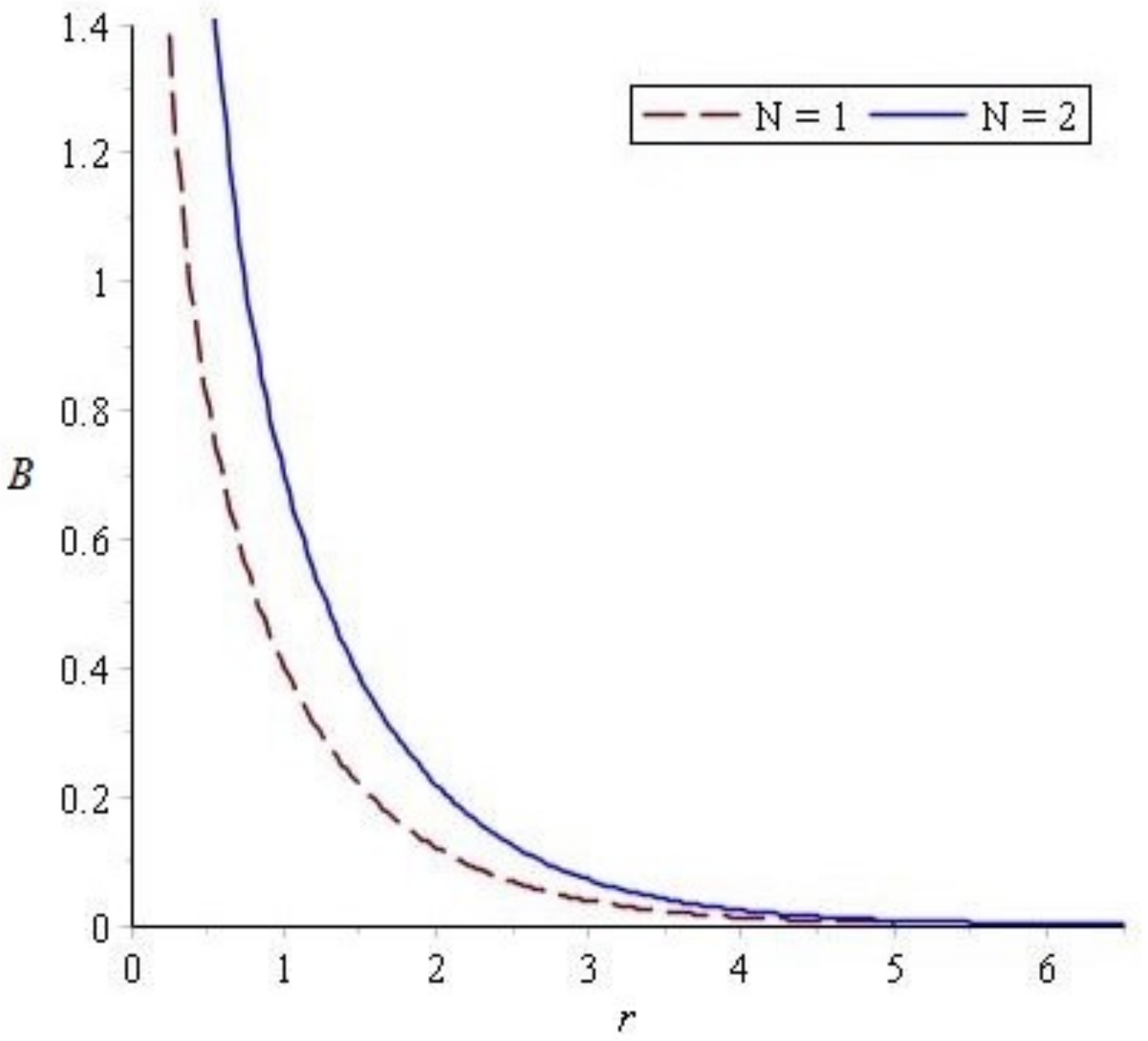}
  \end{tabular}
  \vspace{-4cm}
\caption{Square of the modulus of the Higgs field $|\phi|^2$ (left plot) and the magnetic field $B$ (right plot) as functions of $r$ for vortex number $N=1$ (dashed lines) and $N=2$ (full lines) for solution (\ref{solMik} --\ref{solMikinfty}).}
  \label{plots}
  \end{figure}
\end{center}
  \vspace{-1cm}
In figure \ref{plots} we plot the magnitude of the Higgs field square $|\phi|^2$ and the magnetic field $B$ as functions of $r$ associated to this solution for both vortex numbers. We notice, using equation (\ref{solMik}), that the magnetic field blows up at the origin as $B\sim r^{-2N/3}$ and would not be integrable for $N\geq 3$. This restricts $N$ to be $1$ or $2$, as mentioned above, in order to obtain a finite magnetic flux. In fact, a direct calculation shows that $\int_{\Sigma} B =2\pi N$ (c.f. also equation (\ref{modEn})). It can be done by using equations (\ref{modBog2}) and (\ref{modTaubes}) along with rotational symmetry to write $B=-\dfrac{1}{2r}\dfrac{d}{dr}\left(r\dfrac{dh}{dr}\right)$, then
$$
\frac{1}{2\pi}\int_{\Sigma}B=\frac{1}{2\pi}\int_0^{2\pi}\int_0^\infty-\frac{1}{2r}\dfrac{d}{dr}\left(r\dfrac{dh}{dr}\right) r dr d\theta=-\frac{1}{2}\left[r\dfrac{dh}{dr}\right]_{r=0}^{\infty}=N,
$$
where we have used the asymptotic expressions (\ref{solMik}) and (\ref{solMikinfty}) in the last equality.

The magnetic field for the models corresponding to $q=0$ and $q=1/3$ present a similar behaviour. At the origin they diverge as $B\sim r^{-1}$ and $B\sim r^{-4/3}$, respectively, while they monotonically tend to zero at infinity. Both give the same magnetic flux $2\pi$ corresponding to $N=1$ vortex solutions.

\newpage
\thispagestyle{empty}
\vspace*{\fill}
\newpage
\chapter{Metrisability of Painlev\'e equations and degenerate solutions}\label{chapPainMet}
%%%%%%%%%%%%%%%%

%\section{Introduction}
% An ODE has the \textit{Painlev\'e property} if its general solution does not possess movable critical points. Critical points are singularities around which a function displays multivaluedness. This definition is motivated by the possibility of defining functions from ODE's on the complex plane (c.f. \cite{Conte2013} for a short survey on the subject).  Non-linear second order ODE's of the form $y''=R(x,y,y')$, where $R(x,y,y')$ is a rational function of $(y,y')$ with coefficients analytic in $x$, that possess the Painlev\'e property were classified \cite{Painleve1900,Painleve1902,Gambier1910,Ince} and there are $50$ of them up to the change of variables
%$$
%y\mapsto \frac{a(x)+b(x) y}{c(x)+d(x) y},
%$$
%for analytic functions $a,b,c,d$. This transformation preserves the Painle\-v\'e property. Among these $50$ equations, $44$ of them admit general solutions given by classical functions while the other $6$ are the Painlev\'e equations.
In chapter \ref{chapVorLike}, the Painlev\'e test was presented as a means to find necessary conditions for a differential equation to admit the Painlev\'e property. We also mentioned that there exists a classification of second order ODE's admiting this property. These equations define $2$-dimensional projective structures whose metrisability can be analysed. The local metrisability problem in $2D$ was solved in \cite{Bryant2009}. It turns out that such projective structures are metisable if and only if their corresponding ODE's admit a first integral quadratic in $y'$. 
%This implication is somehow unexpected. In fact, no fundamental relation between PP and metrisability had been established as yet, especially because the latter is a coordinate independent property. We perform the general metrisability analysis of the Painlev\'e equations and determine which of them. 
In particular, we conclude that the real solutions $y(x)$ of the Painlev\'e transcendents cannot be geodesic curves of a metric, except for the special case for PVI$\left(0,0,0,\frac{1}{2}\right)$, which admits the so called ``Picard solutions" in terms of the Weierstrass $\wp$ function and for which the projective structure is flat. Moreover, the metrics associated to the metrisable projective structures all admit a Killing vector, allowing us to recover the known quadratic first integrals, providing a geometrical interpretation to them.

In the second part of this chapter, we consider a problem raised in \cite{DunEast2016}, where it is shown that the Egorov and Newtonian projective structures in $3$ dimensions are not metrisable and admit a degenerate solution space of the metrisability equations of dimension at most $3$. We prove that in $n$ dimensions the solutions of rank $k<n$ for a given kernel form a vector space of dimension at most $\frac{k(k+1)}{2}$ (c.f. Theorem \ref{conjecture}) and that this bound is attained.% This result includes Lemma 4.3 of \cite{Bryant2009}.
  
%(to be written later) Include \cite{Fed2014} and others...

\section{Metrisability of projective structures}\label{secMetProjSt}

Consider the set of affine torsion-free connections on an $n$-dimensional simply-connected smooth orientable manifold $M$. We define the following equivalence relation: two connections $\Gamma$ and $\hat\Gamma$ are \textit{projectively equivalent} if they share the same unparametrised geodesics. We will also write $\nabla$ and $\hat\nabla$, referring to their covariant derivatives.
\begin{defn}\label{DefProjSt}
A \textit{projective structure} $\left[\Gamma \right]$ is the class of torsion-free connections that are projectively equivalent to $\Gamma$. We could also denote $\left[\Gamma\right]$ by its corresponding covariant derivative $\left[\nabla\right]$.
\end{defn}	 
There is an algebraic formulation of this equivalence relation given by
\begin{prop}
Two torsion free connections $\Gamma$ and $\hat\Gamma$ are projectively equivalent if and only if there exists a one form $\Upsilon_a$ such that
\begin{equation}\label{projequiv}
\hat\Gamma^a_{bc}=\Gamma^a_{bc}+\Upsilon_b\delta^a_c+\Upsilon_c\delta^a_b.
\end{equation}
\end{prop}
\begin{proof}
Let $V^a$ be a vector field tangent to a geodesic of the connection $\Gamma$. It means that there exists a function $v$ such that
\begin{equation}\label{propgeod1}
V^a\nabla_aV^b=v V^b.
\end{equation}
We would like to characterise all possible $\hat\Gamma$ such that
\begin{equation}\label{propgeod2}
V^a\hat\nabla_aV^b=\hat v V^b,
\end{equation}
for a function $\hat v$.

Subtract (\ref{propgeod2}) and (\ref{propgeod1}) to find
\begin{equation}\label{condgamma}
V^a V^c\tilde\Gamma^b_{ac}\propto V^b,
\end{equation}
where $\tilde\Gamma^b_{ac}=\hat \Gamma^b_{ac}-\Gamma^b_{ac}$.
Equation (\ref{condgamma}) should be satisfied for any geodesic and thus for any vector field $V^a$.

On the other hand, 
\begin{equation}\label{condprop1}
V^a V^c\tilde\Gamma^{[b}_{ac}V^{d]}=0,\; \forall V^a
\end{equation} 
if and only if 
\begin{equation}\label{condgammatilde}
\tilde\Gamma^{b}_{ac}=\Upsilon_a\delta^b_c+\Upsilon_c\delta^b_a
\end{equation} 
for some $1$-form $\Upsilon_a$. In fact, this can be easily seen by applying (\ref{condprop1}) with particular choices of $V^a$ to find restrictions on $\Gamma^{b}_{ac}$.

Therefore, (\ref{condgammatilde}) is equivalent to (\ref{condgamma}), and the proof is complete.
\end{proof}
Even though we will work out projectively invariant results, there is a natural choice of representative in each projective class as we now describe.

Since the manifold is orientable by hypothesis, there exists a nowhere vanishing volume form $\epsilon_{ab\cdots d}$, which is obviously unique up to a scale. Given a projective class $\left[\Gamma\right]$, we can choose a representative  $\nabla$ or, in other words, a $1$-form $\Upsilon_a$, making the volume form parallel, namely
\begin{equation}\label{specialcond}
\nabla_a\epsilon_{bc\cdots d}=0.
\end{equation}
In fact, if $\hat\nabla$ is a representative of the projective class, then $\nabla$ is associated to it via (\ref{projequiv}) with $\Upsilon_a=-\frac{1}{(n+1)\epsilon_{12\cdots n}}\hat\nabla_a\epsilon_{12\cdots n}$, which is a $1$-form.
\begin{defn}
A torsion-free affine connection $\nabla_a$ admitting a volume form $\epsilon_{ab\cdots c}$ satisfying (\ref{specialcond}) is called \textit{special}.
\end{defn}
Special connections will play a crucial role in setting up the metrisability problem in terms of a linear system of overdetermined PDE's (c.f. Theorem \ref{thmmeteq}). However, the property of being special is not invariant under a general change of connection (\ref{projequiv}) but only under those corresponding to exact $1$-forms $\Upsilon_a=\nabla_a f$, for a smooth function $f$. To understand why, we have to look at second order derivatives. Let $R$ be the Riemann tensor of $\nabla$, then
\begin{equation}\label{condspecial}
\left(\nabla_a\nabla_b-\nabla_b\nabla_a\right)\epsilon_{c\cdots d}=-R_{ab}{}^e{}_e\epsilon_{c\cdots d}
\end{equation}
which should vanish for special connections. Let us change the connection according to (\ref{projequiv}) and require that the transformed expression also vanish. Under (\ref{projequiv}), the contracted Riemann tensor transforms as $\hat R_{ab}{}^e{}_e=R_{ab}{}^e{}_e+2(n+1)\nabla_{[a}\Upsilon_{b]}$ and the transformed volume form $\hat\epsilon_{c\cdots d}$ must be proportional to $\epsilon_{c\cdots d}$ (as any volume form). Therefore, a necessary condition for the $\hat\nabla$ to be special is
$$
\left(\hat\nabla_a\hat\nabla_b-\hat\nabla_b\hat\nabla_a\right)\hat\epsilon_{c\cdots d}=-2(n+1)\nabla_{[a}\Upsilon_{b]}\hat\epsilon_{c\cdots d}=0
$$
implying that $\Upsilon$ is closed and therefore exact on a simply connected manifold. 

Now, under a change of connection of the form $\Upsilon_a=\nabla_a f$,
$$
\hat\nabla_a\epsilon_{b\cdots d}=\nabla_a\epsilon_{b\cdots d}-(n+1)\Upsilon_a \epsilon_{b\cdots d}.
$$
Hence $\hat\nabla_a\hat\epsilon_{b\cdots c}=e^{(n+1)f}\nabla_a\epsilon_{b\cdots c}=0$, where $\hat\epsilon_{b\cdots c}=e^{(n+1)f}\epsilon_{b\cdots c}$. We call such transformations \textit{special} and a corresponding equivalence class \textit{special projective class}. So, the volume form is rescaled by $e^{(n+1)f}$ under a special transformation $\Upsilon_a=\nabla_a f$ as a requirement for it to remain parallel.

The Riemann tensor of an affine connection can be decomposed as 
\begin{equation}\label{decRiemann}
R_{ab\;\;c}^{\;\;\;\, d}=W_{ab\;\;c}^{\;\;\; d}+\delta^d_a P_{bc}-\delta^d_b P_{ac}+\beta_{ab}\delta^c_d,
\end{equation}
were $\beta_{ab}=P_{ba}-P_{ab}$, $P_{ab}$ is the projective Schouten tensor; and $W_{ab\;\;c}^{\;\;\; d}$ is totally trace-free, has the same symmetries as $R_{ab\;\;c}^{\;\;\;\, d}$ (including the algebraic Bianchi identity) and is called the projective Weyl tensor. A connection is special if and only if $\beta_{ab}=0$, which corresponds to the vanishing of (\ref{condspecial}). In this case, $P_{ab}=\frac{1}{n-1}R_{ab}$ is symmetric.

Now we are able to formulate the problem in terms of a linear system of PDE's.

\begin{theo}\label{thmmeteq}
Let $\nabla$ be a special torsion-free connection and $\sigma^{ab}$ be a symmetric tensor such that $\det(\sigma)\equiv \epsilon_{a\cdots b}\epsilon_{c\cdots d}\sigma^{ac}\cdots\sigma^{bd}\not=0$ and
\begin{equation}\label{meteq}
\nabla_a\sigma^{bc}=\delta^b_a\mu^c+\delta^c_a\mu^b,
\end{equation}
for some vector field $\mu^a$. Then, $\nabla$ is projectively equivalent to the Levi-Civita connection of the metric defined by 
\begin{equation}\label{metcont}
g^{ab}=\det(\sigma) \sigma^{ab}.
\end{equation}
\end{theo}
\begin{proof}
See \cite{East2008}.
\end{proof}

By contracting indices in (\ref{meteq}) we see that $\mu^a=\frac{1}{n+1}\nabla_b\sigma^{ab}$. Then, (\ref{meteq}) can be rewritten as
\begin{equation}\label{eqmettracefree}
\left(\nabla_a\sigma^{bc}\right)_\circ\equiv\nabla_a\sigma^{bc}-\frac{1}{n+1}\delta^b_a\nabla_d\sigma^{cd}-\frac{1}{n+1}\delta^c_a\nabla_d\sigma^{bd}=0,
\end{equation}
where $\left(\nabla_a\sigma^{bc}\right)_\circ$ denotes the trace-free part of $\nabla_a\sigma^{bc}$. Equation (\ref{meteq}) is projectively invariant since it implies
$$
\hat\nabla_a\hat\sigma^{bc}=\delta^b_a\hat\mu^c+\delta^c_a\hat\mu^b,
$$
where $\hat\sigma^{bc}=e^{-2f}\sigma^{bc}$ and $\hat\mu^b=e^{-2f}\left(\mu^b+\Upsilon_d\sigma^{db}\right)$.
\subsection{Metrisability in $2$ dimensions}

%This section is based on \cite{Bryant2009}.

In $2$ dimensions we can use the volume form to raise and lower indices according to $\epsilon^{ab} v_b=v^a$ and $ v^b\epsilon_{ba}=v_a$. Clearly, $\epsilon^{ab}\epsilon_{ac}=\delta^b_c$.

Since we are working with special connections, we can lower the indices $bc$ in (\ref{meteq}) and use the symmetry of $\sigma$ to conclude that, in two dimensions, (\ref{meteq}) is equivalent to a Killing equation
\begin{equation}\label{eqmetsym}
\nabla_{(a}\sigma_{bc)}=0.
\end{equation}
In fact, to show that (\ref{eqmetsym}) implies (\ref{meteq}), write (\ref{eqmetsym}) explicitly for the four choices of symmetrised indices $(abc)$ and raise indices by multiplying by $\epsilon^{12}$. Notice however that this equivalence is not true in higher dimensions.

Let us analyse (\ref{eqmetsym}) in more detail. Let us choose local coordinates $(x^1,x^2)\equiv(x,y)$ for $M$. Let $\psi_1=\sigma_{11}$, $\psi_2=\sigma_{12}$ and $\psi_3=\sigma_{22}$. Choose the following Thomas symbols \cite{Thomas1925} as representative of the projective structure: $\Pi^c_{ab}=\Gamma^c_{ab}-\frac{1}{3}\Gamma^d_{da}\delta^c_b-\frac{1}{3}\Gamma^d_{db}\delta^c_a$, which is independent of the initial choice $\Gamma^a_{bc}\in\left[\Gamma\right]$. Now, it is worth commenting that $\Gamma^d_{da}$ is not a 1-form, and thus $\Pi^c_{ab}$ does not transform as an affine connection in general, but only under coordinate transformations of constant Jacobian\footnote{The coordinate transformation of Thomas symbols are given in equation (\ref{transfThomas}).}. So once we choose this representative we can only apply this kind of coordinate transformations in (\ref{eqmetsym}). In particular, Thomas \cite{Thomas1925} introduced the terminologies ``equi-transformation" for coordinate changes preserving the volume (of Jacobian identically $1$), ``projective connection" for $\Pi^c_{ab}$ and ``equi-tensor" for entities such as $\Gamma^d_{da}$ transforming like tensors under equi-transformations. 

Equation (\ref{eqmetsym}), using the projective connection $\Pi^c_{ab}$, yields
\begin{align}
\dfrac{\partial \psi_1}{\partial x}&=\frac{2}{3} A_1\psi_1-2A_0\psi_2, \label{syseq1} \\
\dfrac{\partial \psi_3}{\partial y}&=2 A_3\psi_2-\frac{2}{3}A_2\psi_3, \label{syseq2}\\
\dfrac{\partial \psi_1}{\partial y}+2\dfrac{\partial \psi_2}{\partial x}&=\frac{4}{3} A_2\psi_1-\frac{2}{3}A_1\psi_2-2A_0 \psi_3,\label{syseq3} \\
\dfrac{\partial \psi_3}{\partial x}+2\dfrac{\partial \psi_2}{\partial y}&=2 A_3\psi_1-\frac{4}{3}A_1\psi_3+\frac{2}{3}A_2 \label{syseq4}\psi_2,
\end{align}
where
\begin{equation}\label{eqAs}
A_0=-\Gamma^2_{11}, \quad A_1=\Gamma^1_{11}-2\Gamma^2_{12}, \quad A_2=2\Gamma^1_{12}-\Gamma^2_{22},\quad A_3=\Gamma^1_{22}.
\end{equation}
The functions $A_i$ will reappear as coefficients of the unparametrised geodesic equations in (\ref{eqUmparGeo}) and are all independent of the representative in $\left[\Gamma\right]$.

According to Theorem \ref{thmmeteq}, if there exists a solution of (\ref{syseq1}--\ref{syseq4}) such that $\Delta\equiv\psi_1\psi_3-\psi_2^2\not=0$, then the corresponding projective structure admits a metric connection. The metric components (\ref{metcont}) with low indices read
\begin{equation}\label{metrisable}
E=\psi_1/\Delta^2,\quad F=\psi_2/\Delta^2,\quad G=\psi_3/\Delta^2, \quad \Delta=\psi_1\psi_3-\psi_2^2\not=0,
\end{equation}
where $E=g_{11}, F=g_{12}, G=g_{22}$.

\section{Hamiltonian description of geodesics}
Consider the metric
$$
g=E(x,y)dx^2+2F(x,y) dxdy+G(x,y) dy^2
$$
and the geodesic Hamiltonian
\begin{equation}\label{eqHamGeo}
H=\frac{1}{2}g_{ab}p^a p^b.
\end{equation}

After eliminating $p_a:=g_{ab}p^b$ in Hamilton's equations
\begin{align}
\dot x^a&=\dfrac{\partial H}{\partial p_a} \\
\dot p_a&=-\dfrac{\partial H}{\partial x^a},
\end{align}
where $\dot{ }$ indicates derivative with respect to $t$, we find the geodesic equations 
\begin{equation}\label{eqGeodesics}
\ddot x^a+\Gamma^a_{bc}\dot x^b \dot x^c=0.
\end{equation}
We have identified $(x,y)=(x^1,x^2)$.

Since we are working in two dimensions, there are two of those. The unparametrised geodesic equation can be found by eliminating $t$ in (\ref{eqGeodesics}). It is a second order ODE at most cubic in $y'(x)$,
\begin{equation}\label{eqUmparGeo}
y''=A_3(x,y) y'^3+A_2(x,y) y'^2+A_1(x,y) y'+A_0(x,y),
\end{equation}
where the $A_i$'s are defined in (\ref{eqAs}) and ${}^\prime$ means $\dfrac{d}{dx}=\dfrac{1}{\dot x}\dfrac{d}{dt}$, assuming $\dot x\not=0$. To derive this equation, it is useful to suppose that $\dot x$ and $\dot y$ are different from zero, divide the equations for $x$ and $y$ in (\ref{eqGeodesics}) by $\dot x$ and $\dot y$, respectively and subtract them. Equation (\ref{eqUmparGeo}) is obviously independent of the parameter of the geodesic equations and can be derived in the same way even if we do not start with an affine parameter. Here we assumed that we can invert the function $x=x(t)$, which is true in sufficiently small neighbourhoods of points $t$ where $\dot x(t)\not= 0$ as assumed above.

A particular interesting case for the integrability point of view is when the metric $g$ admits a Killing vector $K$. In this case the following quantity is conserved along geodesics
\begin{equation}\label{ConsKilling}
K_a \dot x^a=g_{ab}K^{b}\dot x^a.
\end{equation}

By construction, the Hamiltonian (\ref{eqHamGeo}) is also conserved along unparametrised geodesics. Thus, using (\ref{ConsKilling}) to eliminate $t$ in (\ref{eqHamGeo}) we find that
\begin{equation}\label{eq1stIntegral}
\frac{1}{(K_1+K_2 y')^2}\left(E(x,y)+2F(x,y)y'+G(x,y)y'^2\right)
\end{equation}
is also conserved along geodesics. In other words, (\ref{eq1stIntegral}) is a first integral of the unparametrised geodesic equation (\ref{eqUmparGeo}).

This result generalises to the following theorems.
\begin{theo}\label{thm1stintegral}
If a projective structure in $n=2$ dimensions admits at least two linearly independent solutions $\sigma^{(1)}$ and $\sigma^{(2)}$ to (\ref{eqmetsym}), then
$$
I(x,y,y'):=\frac{\sigma^{(1)}_{11}+2\sigma^{(1)}_{12}y'+\sigma^{(1)}_{22}y'^2}{\sigma^{(2)}_{11}+2\sigma^{(2)}_{12}y'+\sigma^{(2)}_{22}y'^2}
$$
is a first integral of the unparametrised geodesic equation (\ref{eqUmparGeo}).

If one of the solutions, say $\sigma^{(2)}$ is degenerate, then the projective structure is metrisable and the corresponding metric $g$ admits a Killing vector.
\end{theo}
\begin{proof}
For the first part, it suffices to notice that equation (\ref{eqmetsym}) means that $\sigma^{(i)}$ $(i=1,2)$ are Killing tensors. Therefore $I^{(i)}(t):=\sigma^{(i)}_{11}\dot x^2+2\sigma^{(i)}_{12}\dot x\dot y+\sigma^{(i)}_{22}\dot y^2$ are conserved along geodesics, that is to say $\dot I^{(i)}\equiv 0$ if $(x(t),y(t))$ parametrise a geodesic with affine parameter $t$. Then
$$
\dfrac{d}{dx}I(x,y(x),y'(x))=\frac{1}{\dot x}\dfrac{d}{dt}\frac{I^{(1)}(t)}{I^{(2)}(t)}=0,
$$
where we have used $\dot y/\dot x=y'$ to write $I=I^{(1)}/I^{(2)}$.

For the second part, the projective structure is metrisable as a consequence of Lemma 4.3 in \cite{Bryant2009}. Now, if $\sigma^{(2)}$ is degenerate, then there exists a non-vanishing $1$-form $\omega$ such that $\sigma^{(2)}_{ab}=\omega_a\omega_b$. Then the metrisability equations (\ref{eqmetsym}) yield
$$
\nabla_{(a}\omega_{b)}=0.
$$
It does not yet mean that $g^{ab}\omega_b$ is a Killing vector of $g$ because $\nabla$ is not the Levi-Civita connection, which is obtained by applying a transformation (\ref{projequiv}) with $\Upsilon_a=\nabla_a \left(-\frac{1}{2}\ln\left|\Delta\right|\right)$, so that
$$
\hat\nabla_{(a}\hat\omega_{b)}=0,
$$
where $\hat\omega=\frac{\omega}{\Delta}$ and $\hat\nabla$ is the Levi-Civita connection of the metric $g=\sigma/\Delta^2$ (c.f. equation (\ref{metrisable})). This means that $g^{ab}\hat\omega_b$ is a Killing vector of $g$.
\end{proof}
Again, we stress that $\Upsilon_a=\nabla_a \left(-\frac{1}{2}\ln\left|\Delta\right|\right)$ above is not a tensor, but an equi-tensor (c.f. paragraphs below equation (\ref{eqmetsym})).

The proof of the following theorem is included in the previous one and gives a relation between Killing forms and degenerate solutions to the metrisability equations in two dimensions.

\begin{theo}\label{thmKFormsDeg}
A two-dimensional projective structure $\left[\nabla\right]$ has a special representative $\nabla$ admitting a Killing form $\omega_a$, i.e., $\nabla_{(a}\omega_{b)}=0$, if and only if it admits a degenerate solution $\sigma_{ab}=\omega_{(a}\omega_{b)}$ to its metrisability equations (\ref{eqmetsym}). Moreover, if this projective structure is metrisable, then its degree of mobility is at least $2$.
\end{theo}

The converse of the last statement is not true. In fact, if the degree of mobility of a projective structure is greater than $1$, then the metrics do not necessarily admit a Killing vector. As a counterexample, consider the pair of metrics
$$
h_1=\left(X(x)-Y(y)\right)\left(dx^2+dy^2\right)
$$
and
$$
\quad h_2=\left(\frac{1}{Y(y)}-\frac{1}{X(x)}\right)\left(\frac{dx^2}{X(x)}+\frac{dy^2}{Y(y)}\right),
$$
which is projectively equivalent and whose unparametrised geodesic equation is
\begin{equation}\label{eqnoKilling}
y''+\frac{1}{2\left(X(x)-Y(y)\right)}\left(Y'+X' y'+Y' y'^2+X' y'^3\right)=0.
\end{equation}
These metrics do not admit a Killing vector in general, but it is clear that the projective structure has degree of mobility at least $2$.

From Theorem \ref{thm1stintegral}, a first integral of (\ref{eqnoKilling}) is 
$$
\left(\frac{\det h_1}{\det h_2}\right)^{2/3}\frac{(h_2)_{ab}\dot x^a\dot x^b}{(h_1)_{ab}\dot x^a\dot x^b}=\frac{Y(y)+X(x)y'^2}{1+y'^2}.
$$

In section \ref{secmetP}, we show that all Painlev\'e equations admit a special representative admitting a Killing form (following Theorem \ref{thmKFormsDeg}). The problem of existence of Killing forms for two-dimensional affine connections was solved in \cite{FCMD2015} (c.f. also Chapter \ref{chapterKilling}), where it was also shown that the semi-invariant $\nu_5$ defined in \cite{HietDry2002} necessarily  vanishes for projective structures having special representatives admitting Killing forms. Therefore, our results allow us to say that $\nu_5$ vanishes for all Painlev\'e equations \cite{HietDry2002}.

The simplest projective structure is the flat one, namely, whose unparametrised geodesic equation is $y''=0$ (up to point transformations). In this case, the general solution to equations (\ref{syseq1}--\ref{syseq4}) has $6$ constants of integration and the metric has constant curvature \cite{East2008}. Actually, these three properties are equivalent \cite{Bryant2009,East2008}.

Let us state a result known to R. Liouville giving necessary and sufficient conditions for the projective structure to be flat.
\begin{theo}\cite{Liouville1886}
A second order ODE $y''=\Lambda(x,y,y')$ is equivalent to $y''=0$ under a point transformation, i.e., the projective structure is flat, if and only if it is of the form (\ref{eqUmparGeo}) and the following quantities, called \textit{Liouville invariants}, vanish identically
\begin{align*}
L_1&=\frac{2}{3}\dfrac{\partial^2A_1}{\partial x\partial y}-\frac{1}{3}\dfrac{\partial^2A_2}{\partial x^2}-\dfrac{\partial^2A_0}{\partial y^2}+A_0\dfrac{\partial A_2}{\partial y}+A_2\dfrac{\partial A_0}{\partial y}-A_3\dfrac{\partial A_0}{\partial x}- \\
&-2 A_0 \dfrac{\partial A_3}{\partial x}-\frac{2}{3}A_1\dfrac{\partial A_1}{\partial y}+\frac{1}{3}A_1\dfrac{\partial A_2}{\partial x},
\end{align*}
\begin{align*}
L_2&=\frac{2}{3}\dfrac{\partial^2A_2}{\partial x\partial y}-\frac{1}{3}\dfrac{\partial^2A_1}{\partial x^2}-\dfrac{\partial^2A_3}{\partial x^2}-A_3\dfrac{\partial A_1}{\partial x}-A_1\dfrac{\partial A_3}{\partial x}+A_0\dfrac{\partial A_3}{\partial y}+\\
&+2 A_3 \dfrac{\partial A_0}{\partial y}+\frac{2}{3}A_2\dfrac{\partial A_2}{\partial x}-\frac{1}{3}A_2\dfrac{\partial A_1}{\partial y}.
\end{align*}
\end{theo}

The vanishing of both $L_1$ and $L_2$ is invariant under point transformation, but not the vanishing of each one individually. This is because $L_1$ and $L_2$ can be interpreted as the components of a $1$-form recovered in \cite{FCMD2015} (c.f. also Chapter \ref{chapterKilling}).

In what follows, we will study metrisability of projective structures defined from the Painlevé equations considered as unparametrised geo\-desic equations of the form (\ref{eqUmparGeo}). To calculate first integrals, when they exist, we will find a Killing vector and use (\ref{eq1stIntegral}), however we could also use Theorem \ref{thm1stintegral}.

%%%%%%%%%%%%%%%%%%%%%%%%%%%%%%%%%%%%%%%%%%%%%%%%%%%%%
\section{Metrisability of Painlevé equations}\label{secmetP}
%%%%%%%%%%%%%%%%%%%%%%%%%%%%%%%%%%%%%%%%%%%%%%%%%%%%%
In this section, we analyse the metrisability of projective structures originated from the Painlevé equations, which are of the form (\ref{eqUmparGeo}).  

The six Painlevé equations are:
\begin{align*}
y''&=6y^2+x &\text{(PI)}&\\
y''&=2y^3+xy+\alpha &\text{(PII)}&\\
y''&=\frac{1}{y}y'^2-\frac{1}{x}y'+\alpha \frac{y^2}{x}+\frac{\beta}{x}+\gamma y^3+\frac{\delta}{y} &\text{(PIII)}&\\
y''&=\frac{1}{2y}y'^2+\frac{3}{2}y^3+4xy^2+2(x^2-\alpha)y+\frac{\beta}{y} &\text{(PIV)}&\\
y''&=\left(\frac{1}{2y}+\frac{1}{y-1}\right)y'^2-\frac{1}{x}y'+\frac{(y-1)^2}{x^2}\left(\alpha y+\frac{\beta}{y}\right)+& & \\
&+\gamma\frac{y}{x}+\delta \frac{y(y+1)}{y-1} &\text{(PV)}& \\
y''&=\frac{1}{2}\left(\frac{1}{y}+\frac{1}{y-1}+\frac{1}{y-x} \right)y'^2-\left(\frac{1}{x}+\frac{1}{x-1}+\frac{1}{y-x}\right)y'+& & \\
&+\frac{y(y-1)(y-x)}{x^2(x-1)^2}\left[\alpha+\beta\frac{x}{y^2}+\gamma\frac{x-1}{(y-1)^2}+\delta\frac{ x(x-1)}{(y-x)^2} \right] &\text{(PVI)}&
\end{align*}
where $\alpha,\dots,\delta$ are constant parameters and the metrisability properties will strongly depend on their values. When convenient we will indicate them in parenthesis in front of the equation label, for instance: (PII)$(\alpha)$, (PIII)$(\alpha,\beta,\gamma,\delta)$ and so on.

These equations do not have a cubic term in $y'$ ($A_3=0$). A general (not unique!) approach to seek solutions to the metrisability problem of this kind of projective structure is the following: 
\begin{itemize}
\item[0.] (optional) calculate the invariants in \cite{East2008}. If they do not vanish identically, then there is no non-trivial solution to (\ref{syseq1}--\ref{syseq4}) and there is no point following the next steps;
\item[1.] solve equation (\ref{syseq2}) for $\psi_3$ (note that $A_3=0$);\\
\item[2.] substitute $\psi_3$ in (\ref{syseq4}) and solve it for $\psi_2$; \\
\item[3.] apply the integrability condition $\partial_x\partial_y\psi_1=\partial_y\partial_x\psi_1,\, \forall x,y,$ to the remaining equations (\ref{syseq1}) and (\ref{syseq3});\\
\item[4.] if step 3. is successful, solve equations (\ref{syseq1}) and (\ref{syseq3}).
\end{itemize}
Step 0. is optional because it is equivalent to step 3. After steps 1. and 2., in general, we end up with a solution for $\psi_2$ and $\psi_3$ depending on arbitrary functions of $x$ and $y$ originated from integration of PDEs. Step 3. is necessary to fix those functions up to constants of integration.

The above steps may be troublesome to be performed by hand, but they are easily implemented on the computer using softwares of symbolic calculus.

In \cite{Bryant2009} it is shown that (PI) is not metrisable. The same is true for (PII) and (PIV). The reason is that step 3. above implies $\psi_2=\psi_3=0$ and thus we cannot define a metric via (\ref{metrisable}). On the other hand, (PIII), (PV) and (PVI)-projective structures are metrisable for special values of parameters, as we discuss below. The values of the parameters can be found in step 3. from the condition $\partial_x\partial_y\psi_1=\partial_y\partial_x\psi_1$, which involves $\alpha,\dots,\delta$ and should hold for all values of $x$ and $y$ in their domain of definition.

For other choices of parameters, step 3. forces us to choose $\psi_2=\psi_3=0$ and get a degenerate solution. An obvious degenerate solution is the trivial one $\psi_i=0$. However, for the Painlev\'e equations, there always exist non-trivial ones spanning a $1$-dimensional space, which is the maximal dimension allowed to degenerate solutions (c.f. Corollary \ref{conjecture} below or Lemma 4.3 of \cite{Bryant2009}). To see this, set $\psi_2=\psi_3=0$. Then (\ref{syseq1}--\ref{syseq4}) reduce to a closed overdetermined system for $\psi_1$ which has a non-vanishing solution if and only if $\partial_y A_1=2\partial_x A_2$. It is straightforward that this condition is fulfilled by all equations (PI--PVI), that is why all invariants of \cite{Bryant2009} vanish for Painlev\'e equations. The degenerate solutions corresponding to each Painlev\'e equation are, up to a multiplicative constant, (PI,PII): $\psi_1=1$, (PIII): $\psi_1=\frac{y^{4/3}}{x^{2/3}}$, (PIV): $\psi_1= y^{2/3}$, (PV): $\psi_1=\frac{(1-y)^{4/3} y^{2/3}}{x^{2/3}}$ and (PVI): $\psi_1=(x-y)^{2/3}\left[\frac{(y-1) y}{(x-1) x}\right]^{2/3}$.

After analysing (PIII), (PV) and (PVI), we explain in Section \ref{SecCoalescence} how the metrisability results of (PIII) are related to those of (PV) by the method of coalescence of Painlev\'e equations.

%%%%%%%%%%%%%%%%%%%%%%%%%%%%%%%%%%%%%%%%%%
\subsection{Painlevé III}
%%%%%%%%%%%%%%%%%%%%%%%%%%%%%%%%%%%%%%%%%%
Applying steps 1. to 4. above lead us to conclude that (PIII) is metrisable only in the following cases: $\alpha=\gamma=0$ or $\beta=\delta=0$. If all parameters are zero, then the projective structure is flat. Actually, both cases are essentially the same since the change of coordinates $y\mapsto \frac{1}{y}$ induces (PIII)$(\alpha,\beta,\gamma,\delta)\to$ (PIII)$(-\beta,-\alpha,-\delta,-\gamma)$ and all results from one case can be recovered from the other through this map. Therefore, we only explicitly present the detailed results for $\beta=\delta=0$. We have two subcases to analyse, $(\alpha,\gamma)\not=(0,0)$ and $(\alpha,\gamma)=(0,0)$.

%\subsubsection{Case $\alpha=\gamma=0$ and $(\beta,\delta)\not=(0,0)$} \label{secPIIIcase1}
%
%If $\alpha=\gamma=0$ and $(\beta,\delta)\not=(0,0)$, we have a two-dimensional family of solutions giving rise to the metric
%$$
%g=\frac{y^2[B y^2+Ax(\delta x+2\beta y)]}{A^2x^2 [B y^2+A(\delta x^2+2\beta xy-y^2)]^2} dx^2-\frac{2 y^3}{A x [B y^2+A(\delta x^2+2\beta xy-y^2)]^2}dx dy+\frac{y^2}{A [B y^2+A(\delta x^2+2\beta xy-y^2)]^2}dy^2,
%$$
%where $A$ and $B$ are arbitrary constants. The metric admits a one-parameter family of isometries $(x,y)\mapsto (e^s x,e^s y)$, generated by the Killing vector
%$$
%K=x\dfrac{\partial}{\partial x}+y\dfrac{\partial}{\partial y}
%$$
%From the conservation of (\ref{ConsKilling}), we find that
%$$
%C_K=\frac{y^2\dot x}{x\left(\frac{B}{A} y^2+\delta x^2+2\beta xy-y^2\right)}
%$$
%is conserved along geodesics.
%
%From (\ref{eq1stIntegral}), we conclude that 
%$$
%I=x^2 \left(\frac{y'}{y}\right)^2-2x\frac{y'}{y}+\delta \frac{x^2}{y^2}+2\beta \frac{x}{y}
%$$
%is a first integral of (PIII), as known in the literature (c.f. \cite{Gromak1999}).
%
%Alternatively, we could have used Theorem \ref{thm1stintegral} to find the same results, indeed notice that for $A=0$ we have a degenerate solution.

\subsubsection{Case $\beta=\delta=0$ and $(\alpha,\gamma)\not=(0,0)$}

If $\beta=\delta=0$ and $(\alpha,\gamma)\not=(0,0)$, we have a two-dimensional family of solutions giving rise to the metric
\begin{align}\label{metPIII1}
g=&\frac{B-Axy(2\alpha+\gamma xy)}{A^2x^2(A-B+2A\alpha xy+A\gamma x^2y^2)^2} dx^2+ \nonumber \\
+&\frac{2}{Axy(A-B+2A\alpha xy+A\gamma x^2y^2)^2}dx dy+\nonumber \\
+&\frac{1}{A y^2(A-B+2A\alpha xy+A\gamma x^2y^2)^2}dy^2,
\end{align}
where $A$ and $B$ are arbitrary constants. The metric admits a one-parameter family of isometries $(x,y)\mapsto (e^s x,e^{-s} y)$, generated by the Killing vector
$$
K=x\dfrac{\partial}{\partial x}-y\dfrac{\partial}{\partial y}.
$$
Then, if we define $\tilde B=\frac{B}{A}-1$,
$$
C_K=\frac{\dot x}{x\left(\tilde B-2\alpha xy-\gamma x^2 y^2 \right)}
$$
is conserved along geodesics and
$$
I=x^2 \left(\frac{y'}{y}\right)^2+2x\frac{y'}{y}-2\alpha xy-\gamma x^2 y^2
$$
is a first integral of (PIII) \cite{Gromak1999}. Alternatively, we could have used Theorem \ref{thm1stintegral} to find the same results, in fact notice that for $A=0$ we have a degenerate solution.

%We stress that the corresponding results in the case $\alpha=\gamma=0$ and $(\beta,\delta)\not=(0,0)$ can be obtained through the change $y\mapsto \frac{1}{y}$ and $(\alpha,\beta,\gamma,\delta)\to (-\beta,-\alpha,-\delta,-\gamma)$. 

By defining new coordinates $r=xy$ and $s=\ln |x|$ the metric takes the form
$$
g=\frac{1}{A^3\left(-\tilde B+2\alpha r+\gamma r^2\right)^2r^2}dr^2-\frac{1}{A^3\left(-\tilde B+2\alpha r+\gamma r^2\right)}ds^2.
$$
By rescaling $r$ we can set either $\alpha$ to $1$ if $\alpha\not=0$ or $\gamma$ to $\gamma/{|\gamma|}$ if $\gamma\not=0$.

%In order to find a ``surface of revolution", define new coordinates 
%\begin{align*}
%t&=\ln|r| \\
%R(t)&=\sqrt{-\frac{1}{1 - B + 2 \alpha e^t + \gamma e^{2 t}}}\\
%Z(t)&=\bigintss^y\sqrt{-\frac{1 - \tilde B + 2 \alpha e^t + a^2 e^{2 t} + \gamma e^{2 t} + 2 \alpha \gamma e^{3 t} + 
% \gamma^2 e^{4 t}}{(-1 + \tilde B - 2 \alpha e^t - \gamma e^{2 t})^3}}dw
%\end{align*}
%the metric takes the form
%$$
%A^3 g=\left(R'(t)^2+Z'(t)^2\right)dt^2+R(t)^2 d\theta ^2=\left[1+\left(\dfrac{dZ}{dR}\right)^2\right]dR^2+R^2 d\theta ^2,
%$$
%representing the induced metric of a surface of revolution, around the axis $Z$, of the parametrized curve $t\mapsto (R(t),Z(t))$ or, equivalently, of the graph of $R\mapsto Z(R)$ if this function exists.
%
%For $\alpha=-1$, $\gamma=0$, $\tilde B=2$, and $\alpha=-1$, $\gamma=0$, $\tilde B=0$ this curve and its surface are depicted in Figures (\ref{fig:PIIIam1c0B2}) and (\ref{fig:PIIIam1c0B0}), respectively.
%\begin{figure}[H]
%\centering
%\includegraphics[width=0.35\textwidth]{CurvePIIIam1c0B2.pdf}
%\includegraphics[width=0.35\textwidth]{SurfacePIIIam1c0B2.pdf}
%\caption{\label{fig:PIIIam1c0B2}Curve $t\in[0,2]\mapsto (R(t),Z(t))$ and its surface of revolution for $\alpha=-1$, $\gamma=0$, $\tilde B=2$.}
%\end{figure}
%\begin{figure}[H]
%\centering
%\includegraphics[width=0.35\textwidth]{CurvePIIIam1c0B0.pdf}
%\includegraphics[width=0.35\textwidth]{SurfacePIIIam1c0B0.pdf}
%\caption{\label{fig:PIIIam1c0B0}Curve $t\in[0.1,5]\mapsto (R(t),Z(t))$ and its surface of revolution for $\alpha=-1$, $\gamma=0$, $\tilde B=0$.}
%\end{figure}

\subsubsection{Case $\alpha=\beta=\gamma=\delta=0$}

If $\alpha=\beta=\gamma=\delta=0$, we have a six-dimensional family of solutions, meaning that (PIII) gives rise to a projectively flat metric of constant curvature \cite{East2008}
\begin{multline*}
g=\left[-12 (C_1 + C_2 \ln x + C_3 (\ln x)^2) (3 C_6 - 2 C_5 \ln y  + 3 C_3 (\ln y)^2)+\right. \\ 
\left.+ (6 C_4 - 3 C_2 \ln y  + 2 \ln x (C_5 - 3 C_3 \ln y))^2\right]^{-2}\times\\
\times\left\{432 \left[3 C_6 - 2 C_5 \ln y + 3 C_3 (\ln y)^2\right]\left(\dfrac{dx}{x}\right)^2+ \right. \\
\left. +432 \left[6 C_4 - 3 C_2 \ln y + 2 \ln x (C_5 - 3 C_3 \ln y)\right]\dfrac{dxdy}{xy}+\right.\\
\left.+1296 \left[C_1 + C_2 \ln x + C_3 (\ln x)^2\right]\left(\dfrac{dy}{y}\right)^2 \right\},
\end{multline*}
where $C_i$, $i=1,\dots,6$ are arbitrary constants.

Since the projective structure is flat, (PIII)(0,0,0,0) can be put in the form $\dfrac{d^2Y}{dX^2}=0$ with $Y=e^y$ and $X=\ln x$. To see this, set $C_2=C_3=C_4=C_5=0$, then it is clear that $X,Y$ are flat coordinates for this metric.

\subsection{Painlevé V}

We have analogous results for (PV). The projective structure is metrisable if and only if $\gamma=\delta=0$, and it is projectively flat admitting a metric with constant curvature if and only if $\alpha=\beta=\gamma=\delta=0$.

\subsubsection{Case $\gamma=\delta=0$ and $(\alpha,\beta)\not=(0,0)$}

If $\gamma=\delta=0$ and $(\alpha,\beta)\not=(0,0)$, we have a two-dimensional family of solutions giving rise to the metric
\begin{equation}\label{metPV1}
g=\frac{y}{A^2x^2[By+2A(\beta -\alpha y^2)]} dx^2+\frac{y}{A(y-1)^2[By+2A(\beta -\alpha y^2)]^2}dy^2,
\end{equation}
which admits $(x,y)\mapsto (e^s x, y)$ as as one-parameter family of isometries, generated by the Killing vector
$$
K=x\dfrac{\partial}{\partial x}.
$$
Then, if we define $\tilde B=\frac{B}{A}$, the quantity
$$
C_K=\frac{y}{2\beta-\tilde B y-2\alpha y^2}\frac{\dot x}{x}
$$
is conserved along geodesics and
$$
I=\frac{1}{y}\left(\frac{xy'}{y-1}\right)^2+\frac{2\beta}{y}-2\alpha y
$$
is a first integral of PV (c.f. \cite{Gromak1999}).

By defining $\theta=\ln|x|$, the metric becomes
$$
g=\frac{y}{A^3[\tilde By+2(\beta -\alpha y^2)]} d\theta^2+\frac{y}{A^3(y-1)^2[\tilde By+2(\beta -\alpha y^2)]^2}dy^2.
$$
By redefining $A$, $B$ and $\theta$ we can set either $\beta$ to $\frac{\beta}{|\beta|}$ if $\beta\not=0$ or $\alpha$ to $\frac{\alpha}{|\alpha|}$ if $\alpha\not=0$.

%In order to find a ``surface of revolution", define new coordinates 
%\begin{align*}
%R(y)&=\frac{y}{\tilde By+2(\beta -\alpha y^2)}\\
%Z(y)&=\bigintss^y\sqrt{\frac{b^2 (-1 + w)^2 + 2 b (-1 + a (-1 + w)^2) w^2 + 
% w^3 (-\tilde B + a (2 + a (-1 + w)^2) y)}{(-1 + w)^2 w \left[-2 b + 
%   w \left(-\tilde B + 2 a w\right)\right]^3}}dw
%\end{align*}
%the metric takes the form
%$$
%A^3 g=\left(R'(y)^2+Z'(y)^2\right)dy^2+R(y)^2 d\theta ^2=\left[1+\left(\dfrac{dZ}{dR}\right)^2\right]dR^2+R^2 d\theta ^2,
%$$
%representing the induced metric of a surface of revolution, around the axis $Z$, of the parametrized curve $y\mapsto (R(y),Z(y))$ or, equivalently, of the graph of $R\mapsto Z(R)$ if this function exists.

%For $\alpha=1$, $\beta=0$, $\tilde B=1$, and $\alpha=1$, $\beta=2$, $\tilde B=0$ this curve and its surface are depicted in Figures (\ref{fig:PVa1b0B1}) and (\ref{fig:PVa1b2B0}), respectively.
%\begin{figure}[H]
%\centering
%\includegraphics[width=0.2\textwidth]{CurvePVa1b0B1.pdf}
%\includegraphics[width=0.4\textwidth]{SurfacePVa1b0B1.pdf}
%\caption{\label{fig:PVa1b0B1}Curve $y\in[0,0.5)\mapsto (R(y),Z(y))$ and its surface of revolution for $\alpha=1$, $\beta=0$, $\tilde B=1$.}
%\end{figure}
%\begin{figure}[H]
%\centering
%\includegraphics[width=0.3\textwidth]{CurvePVa1b2B0.pdf}
%\includegraphics[width=0.3\textwidth]{SurfacePVa1b2B0.pdf}
%\caption{\label{fig:PVa1b2B0}Curve $y\in[0.5,0.6]\mapsto (R(y),Z(y))$ and its surface of revolution for $\alpha=1$, $\beta=2$, $\tilde B=0$.}
%\end{figure}

\subsubsection{Case $\alpha=\beta=\gamma=\delta=0$}

If $\alpha=\beta=\gamma=\delta=0$, we have again a six-dimensional family of solutions, meaning that it gives rise to a projectively flat metric of constant curvature \cite{East2008}
\begin{multline*}
g=27\left[9 (C_1^2 - C_3 C_6) + 9 (C_4^2 - 4 C_3 C_5) (\arctanh \sqrt y)^2 +\right. \\
\left. + 6 \arctanh\sqrt y  (2 C_2 C_3 - 3 C_1 C_4 + (C_2 C_4 - 6 C_1 C_5) \ln x) +\right.\\ 
\left.+\ln x  (6 C_1 C_2 - 9 C_4 C_6 + (C_2^2 - 9 C_5 C_6) \ln x)\right]^{-2}\times\\
\times\left\{ \frac{ 3 C_6 - 4 C_2 \arctanh \sqrt y  + 12 C_5 \left(\arctanh\sqrt y\right)^2}{x^2}dx^2-\right.\\
\left.-2\frac{3 C_1 + C_2 \ln x - 3 \arctanh\sqrt y (C_4 + 2 C_5 \ln x)}{x (-1 + y) \sqrt y}dxdy+\right.\\
\left.+3\frac{C_3 + \ln x (C_4 + C_5\ln x)}{(-1 + y)^2 y}dy^2\right\},
\end{multline*}
where $C_i$, $i=1,\dots,6$ are arbitrary constants.

Equation (PV)$(0,0,0,0)$ can be put in the form $\dfrac{d^2Y}{dX^2}=0$ with $Y=\ln\left(\frac{1+\sqrt y}{\sqrt{1-y}}\right)$ and $X=\ln x$. These coordinates can be found by choosing $C_1=C_2=C_4=C_5=0$, then the metric is flat in the coordinates $X,Y$.

\subsection{Painlevé VI}

(PVI) is metrisable if and only if $\alpha=\beta=\gamma=0$, $\delta=\frac{1}{2}$. This choice of parameters is well-known. In fact, in this case (PVI) has a solution given in terms of the elliptic integral \cite{Gromak1999,casale07}
\begin{equation}\label{SolPVI}
\int^{y(x)}_0\frac{dw}{\sqrt{w(w-1)(w-x)}}=a\omega_1(x)+b\omega_2(x), 
\end{equation}
where the right hand side is the general solution of the Picard-Fuchs equation 
\begin{equation}\label{PFeq}
4x(x-1)\omega''(x)-4(2x-1)\omega'(x)-\omega(x)=0, 
\end{equation}
with $a$ and $b$ constants. Since the constants of integration appear linearly in (\ref{SolPVI}), the projective structure is flat\footnote{This is actually the definition of projective flatness used by Liouville \cite{Liouville1886}}. In fact, PVI$\left(0,0,0,\frac{1}{2}\right)$ is trivial in the variables $Y=\frac{1}{\omega_2(x)}\int^{y}_0\frac{dw}{\sqrt{w(w-1)(w-x)}},\; X=\frac{\omega_1(x)}{\omega_2(x)}$.

A ``first integral" is given by \cite{casale07}
\begin{equation}\label{integralPVI}
I=\frac{y'B(x)}{\sqrt{y(y-1)(y-x)}}+\int^y_0\left[A(x)+\frac{B(x)}{2(w-x)}\right]\frac{dw}{\sqrt{w(w-1)(w-x)}},
\end{equation}
where $A$ and $B$ are a solution to the Picard-Fuchs adjoint equations
\begin{align*}
\begin{cases}
A'(x)&=B(x)\frac{1}{4x(x-1)}\\
B'(x)&=-B(x)\frac{1-2x}{x(x-1)}-A(x)
\end{cases}
\end{align*}

%equations (\ref{syseq1}) and (\ref{syseq3}) are too complicated to be solved analytically for $\psi_1$ in step 4. In fact, after step 3. $\psi_2(x,y)$ and $\psi_3(x,y)$ involve $5$ constants of integration and combinations of elliptic functions. 

%%%%%%%%%%%%%%%%%%%%%%%%%%%%%%%%%%%%%%%%%%%%%%%%%
\subsection{Coalescence}\label{SecCoalescence}
%%%%%%%%%%%%%%%%%%%%%%%%%%%%%%%%%%%%%%%%%%%%%%%%%
The first five Painlev\'e equations (PI--PV) can be derived from (PVI) by the process of coalescence \cite{Ince}. It is described by the following change of variables and constants, upon taking the limit $\epsilon\to 0$,
\begin{align*}
(PVI) \to (PV)&: x\mapsto 1+\epsilon x, \quad \delta\mapsto \frac{\delta}{\epsilon^2} ,\quad \gamma\mapsto-\frac{\delta}{\epsilon^2}+\frac{\gamma}{\epsilon};\\
(PV) \to (PIV)&: x\mapsto 1+\sqrt{2} \epsilon x,\quad y\mapsto \frac{1}{\sqrt 2}\epsilon y,\quad \alpha\mapsto \frac{1}{2 \epsilon^4}, \quad \beta\mapsto\frac{\beta}{4}, \\
 &\gamma\mapsto -\frac{1}{\epsilon^4},\quad \delta\mapsto -\frac{1}{2 \epsilon^4}+\frac{\alpha}{\epsilon^2};\\
(PV) \to (PIII)&:  x\mapsto x^2, \quad y\mapsto 1+\epsilon xy,\quad \alpha\mapsto\frac{\gamma}{8\epsilon^2}+\frac{\alpha}{4\epsilon},  \\
&\beta\mapsto -\frac{\gamma}{8\epsilon^2},\quad \gamma\mapsto \frac{\epsilon \beta}{4},\quad \delta\mapsto \frac{\epsilon^2\delta}{8};\\
(PIV) \to (PII)&: x\mapsto -\frac{1}{\epsilon^3}+\frac{\epsilon x}{2^{2/3}},\quad y\mapsto \frac{1}{\epsilon^3}+\frac{2^{2/3}y}{\epsilon},\\
& \alpha\mapsto -\frac{1}{2\epsilon^6}-2\alpha,\quad \beta\mapsto -\frac{1}{2\epsilon^{12}};\\
(PIII) \to (PII)&: x\mapsto 1+\epsilon^2 x,\quad y\mapsto 1+2\epsilon y,\quad \alpha\mapsto -\frac{1}{2\epsilon^6},\\
& \beta\mapsto \frac{1}{2\epsilon^6}+\frac{2\alpha}{\epsilon^3},\quad \gamma\mapsto \frac{1}{4\epsilon^6},\quad \delta\mapsto-\frac{1}{4\epsilon^6};\\
(PII) \to (PI)&: x\mapsto -\frac{6}{\epsilon^{10}}+\epsilon^2 x,\quad y\mapsto \frac{1}{\epsilon^5}+\epsilon y,\quad \alpha\mapsto \frac{4}{\epsilon^{15}}.
\end{align*}
We can use this process to recover a metric of (PIII)$(\alpha,0,\gamma,0)$ from a metric of (PV)$(\alpha,\beta,0,0)$. To do so, it is necessary to start with (\ref{metPV1}) with the constants of integration
$$
A=\left(\frac{4 \gamma}{2 \alpha \epsilon + \gamma}\right)^{\frac{2}{3}},\quad B=\frac{(-\alpha\epsilon + \gamma)(4 \alpha \epsilon + 2 \gamma)^{\frac{1}{3}}}{
 \epsilon^2  \gamma^{\frac{1}{3}}}.
$$
Then, in the limit $\epsilon\to 0$, we find the metric (\ref{metPIII1}) with $A_{III}=1$ and $B_{III}=1-\frac{4\alpha^2}{\gamma}$, where we have attached the index $III$ to indicate that these constants $A_{III}$ and $B_{III}$ correspond to the metric of (PIII)$(\alpha,0,\gamma,0)$.  This is valid of course only if $\gamma\not=0$. In the case $\gamma=0$ we need $A= 4^{2/3}A_{III}$ and $B=\frac{2 \alpha A_{III} + (-A_{III} +B_{III}) \epsilon}{2^{2/3} \epsilon}$.%, so we still have freedom to choose two constants of integration $A_{III}$ and $B_{III}$.

Taking into account the transformation of the parameters $\alpha,\dots,\delta$ in the coalescence procedure, it is clear that we cannot recover metrisability of (PIII) and (PV) from (PVI).

\section{Degree of mobility of degenerate solutions}\label{secDegSol}

One of the conditions of Theorem \ref{thmmeteq} for the existence of a metric, is that $\det(\sigma)\not=0$. In this section we will be interested in solutions of (\ref{eqmettracefree}) such that $\det(\sigma)=0$. More precisely, we want to find bounds to the dimension of the solution space of (\ref{eqmettracefree}) under the restriction that $\sigma$ is a singular matrix. We show that this bound is $\frac{k(k+1)}{2}$ for rank-$k$ solutions once the kernel is fixed (c.f. \cite{Bryant2009} for the result in dimension $2$).

Let us write the metrisability equations (\ref{eqmettracefree}) under the condition $\det(\sigma)=0$. We suppose that there exists a non-vanishing $1$-form $\omega$ such that 
\begin{equation}\label{det0}
\sigma^{ab}\omega_b=0.
\end{equation}

Differentiate (\ref{det0}) to obtain
\begin{equation}\label{dereq}
\left(\nabla_c\sigma^{ab}\right)\omega_b+\sigma^{ab}\nabla_c\omega_b=0,
\end{equation}
use (\ref{eqmettracefree}),
\begin{equation}\label{eq1}
\sigma^{ab}\nabla_c\omega_b+\frac{\delta^a_c}{n+1}\omega_b\nabla_d\sigma^{bd}+\frac{\omega_c}{n+1}\nabla_d\sigma^{ad}=0%=\sigma^{ab}\nabla_c\omega_b+\frac{\omega_c}{n+1}\nabla_d\sigma^{ad}-\frac{\delta^a_c}{n+1}\nabla_d\omega_b\sigma^{bd},
\end{equation}
contract with $\omega_a$, use (\ref{det0}) and (\ref{dereq}) to find
\begin{equation}\label{eq3}
\sigma^{da}\nabla_d\omega_a=\omega_a\nabla_d\sigma^{da}=0.
\end{equation}
Substitute it in (\ref{eq1})
\begin{equation}\label{eq2}
\omega_c \nabla_d\sigma^{ad}=-(n+1)\sigma^{ab}\nabla_c\omega_b.
\end{equation}

This allows us to rewrite (\ref{meteq}) in a closed form
\begin{equation}\label{closetens}
\omega_d\nabla_c\sigma^{ab}=-\left(\delta_c^a\sigma^{be}+\delta^b_c\sigma^{ae}\right)\nabla_d\omega_e.
\end{equation}
The condition (\ref{det0}) allows us to determine all partial derivatives of $\sigma^{ab}$ in terms of $\omega$ and $\sigma$ and no prolongation is needed in this system once the $\omega$'s are fixed. It is not clear how to fix the kernel of degenerate solutions, but some differential and algebraic restrictions can be derived.

Equation (\ref{closetens}) should be consistent for all values of the index $d$. The consistency conditions are algebraic in $\sigma$ and can be obtained either by contracting (\ref{closetens}) with $\omega_a$ and using (\ref{dereq}) or by multiplying (\ref{closetens}) by $\omega_f$ and anti-symmetrising $[df]$:
\begin{equation}\label{algcond}
\sigma^{eb}\left(\omega_d\nabla_c\omega_e-\omega_c\nabla_d\omega_e\right)=0.
\end{equation}

Further conditions, this time algebraic in both $\sigma$ and $\omega$, can be obtained as follows. Differentiate (\ref{algcond}) 
\begin{align*}
\nabla_f\left[\sigma^{be}\left(\omega_d\nabla_c\omega_e-\omega_c\nabla_d\omega_e\right)\right]=0,
\end{align*}
multiply it by $\omega_g$ and use (\ref{closetens}) to get
$$
\sigma^{be}\left[\left(\omega_c\nabla_d\omega_f-\omega_d\nabla_c\omega_f\right)\nabla_g\omega_e+\omega_g\nabla_f\left(\omega_d\nabla_c\omega_e-\omega_c\nabla_d\omega_e\right)\right]=0.
$$
In particular, by using (\ref{algcond}) in the last two terms and anti-symmetrising in $c,d,f$ we get $\sigma^{ae}\omega_{[c}\nabla_f\nabla_{d]}\omega_e=0$, which is an algebraic condition involving the Riemann tensor
\begin{equation}\label{algcondR}
\sigma^{ae}\omega_{[c}{W_{fd]}}^b{}_{e}\,\omega_b=\sigma^{ae}\omega_{[c} {R_{fd]}}^b{}_{e}\omega_b=0,
\end{equation}
where we used (\ref{decRiemann}) in the first equality.

An inductive argument involving the second Bianchi identity generalises this result as the following theorem.
\begin{theo}
If $\sigma$ is a degenerate solution to the metrisability equations (\ref{meteq}) and $\omega$ is a closed element in its kernel, then 
$$
\sigma^{ae}\omega_{[c}\nabla_{a_1}\cdots\nabla_{a_m]}\omega_e=0
$$ 
for $m=1,\dots,n-1$. Equivalently,
\begin{itemize}
\item $\sigma^{ae}\omega_{[c}\nabla_{\left.a_1\right]}\omega_e=0$, if $m=1$,
\item $\sigma^{ae}\omega_{\left[c\right.} {R_{\left.a_1 a_2\right]}}^h{}_{e}\,\omega_h=0$, if $m=2$,
\item $\sigma^{ae}\omega_{\left[c\right.} {R_{a_1 a_2}}^h{}_{\vert f\vert}{R_{a_3 a_4}}^f{}_{\vert :\vert}\cdots {R_{\left.a_{m-1}a_{m}\right]}}^:{}_{e}\,\omega_h=0$, if $m>2$ is even and
\item $\sigma^{ae}\omega_{\left[c\right.} {R_{a_1 a_2}}^h{}_{\vert f\vert}{R_{a_3 a_4}}^f{}_{\vert :\vert}\cdots {R_{a_{m-2}a_{m-1}}}^:{}_{\vert e\vert}\,\nabla_{\left.a_m\right]}\omega_h=0$, if $m>1$ is odd,
\end{itemize}
where the $:$ in an index position means that the corresponding index is contracted with another index in $\cdots$ and the vertical bars $\vert\quad \vert$ around indices mean that these indices do not take part in the antisymmetrisation.
\end{theo}

\subsubsection{Geometric properties of the kernel}

Let us now study some geometric properties of the kernel of a degenerate solution $\sigma$. Since equation (\ref{eq3}) is true for all $\omega\in\ker\sigma$, it is a matter of linear algebra to deduce that $\nabla_d\sigma^{ad}\in\operatorname{Im}\sigma$. Thus there exists a $1$-form $S_a$ such that
\begin{equation}\label{eqdefS}
-\mu^a=-\frac{1}{n+1}\nabla_d\sigma^{ad}=\sigma^{ab}S_b.
\end{equation}

\begin{theo}\label{thmintdist}
Let $\sigma$ be a degenerate solution to (\ref{eqmettracefree}). Then, its kernel is spanned by exact $1$-forms. Moreover, the integral manifold defined by the kernel is totally geodesic.
\end{theo}
\begin{proof}
Let $\omega\in\ker\sigma$ and let $s=n-\text{rank }\sigma$ be the dimension of $\ker\sigma$. From (\ref{eq2}) and (\ref{eqdefS}),
$$ 
\sigma^{ab}\left(\nabla_c\omega_b+\omega_c S_b\right)=0.
$$
Therefore, there exist $1$-forms $T^{(i)}$ such that
$$
\nabla_c\omega_b+\omega_c S_b=T^{(1)}_c\omega^{(1)}_b+\cdots+T^{(s)}_c\omega^{(s)}_b.
$$
By anti-symmetrising $[cb]$ we conclude that $d\omega$ belongs to the ideal algebraically generated by $\ker\sigma$. By Frobenius theorem, $\ker\sigma$ is spanned by exact $1$-forms and there exists a coordinate system $y^1,\dots,y^n$ such that $\ker\sigma=span\{dy^{n-s+1},\dots,dy^n\}$. 

For the second part, let us work in these $y$-coordinates. Let $i,j,k,\dots$ be indices taking values from $1$ to $n-s$, $\alpha,\beta,\gamma,\dots$ take values from $n-s+1$ to $n$ and $a,b,c,\dots$ take values from $1$ to $n$. In such coordinates, saying that the submanifolds $\{y^\alpha=const\}$ are totally geodesic means that geodesics starting with with initial velocities tangent to such submanifolds remain in these submanifolds. Explicitly, let the affine parametrised geodesics be given as solutions to the equations 
$$
\ddot y^a+\Gamma^a_{bc}\dot y^b \dot y^c=0,
$$
then, choosing $a>n-s$ we see that if $\dot y^\alpha=0$, $\alpha=n-s+1,\dots,n$ at some value for the affine parameter, then $\dot y^\alpha$ remains $0$ for the whole geodesic motion. Looking at the geodesic equations above for $a>n-s$, we conclude that this is the case if and only if $\Gamma^\alpha_{ij}=0$. Let us now show that all such $\Gamma^\alpha_{ij}$ symbols vanish because of the metrisability equation.

Choosing $\omega=dy^{\alpha}$, we conclude that $\sigma^{a\alpha}=\sigma^{\alpha a}=0$ and that the submatrix $\left(\sigma^{ij}\right)_{i,j=1}^{n-s}$ is non-singular. Therefore, writing (\ref{eqmettracefree}) for $a,b,c>n-s$ yields
$$
\delta^\beta_\alpha \Gamma_{ij}^\gamma \sigma^{ij}+\delta^\gamma_\alpha \Gamma_{ij}^\beta \sigma^{ij}=0,
$$
which implies that 
\begin{equation}\label{metcond1}
\sigma^{ij}\Gamma^\beta_{ij}=0.
\end{equation}
Now, writing (\ref{eqmettracefree}) for $a,b\leq n-s$, $c>n-s$,
$$
\Gamma^\gamma_{ji}\sigma^{ki}=\frac{1}{n+1}\delta^k_j\Gamma^\gamma_{il}\sigma^{il}=0,
$$
where we used (\ref{metcond1}) in the last equality. This implies that $\Gamma^\gamma_{ji}=0$ because the submatrix $\left(\sigma^{ki}\right)$ is invertible. This concludes the proof that the integral submanifolds of $\ker\sigma$ are totally geodesic.
\end{proof}

From basic linear algebra, the submanifolds defined by $\ker\sigma$ are integrable submanifolds of the distribution generated by the image of $\sigma$. Thus we deduce the following corollary.

\begin{coll}
Under the conditions of Theorem \ref{thmintdist}, the distribution defined by the image of $\sigma$ is integrable.
\end{coll}
\begin{proof}
We write here an alternative proof to the one given by the previous paragraph.

Let $V^a_i=\sigma^{ab}\,\Omega^i_b,\, i=1,2,$ be two arbitrary vectors in the image of arbitrary $1$-forms $\Omega^i$ by $\sigma$. We want to show that $\left[V_1,V_2\right]\in\operatorname{Im}\sigma$. In fact, using (\ref{meteq}),
$$
\left[V_1,V_2\right]^c=2\,\sigma^{cb}\left(\Omega^{[1}_b\Omega^{2]}_d\mu^d+\sigma^{ae}\Omega^{[1}_e\nabla_a \Omega^{2]}_b\right)\in\operatorname{Im}\sigma.
$$
By Frobenius integrability, the image of $\sigma$ is an integrable distribution.

\end{proof}

Theorem \ref{thmintdist} allows us to make a suitable choice of coordinates to prove the final result of this section.
\begin{theo}\label{conjecture}
The maximal dimension of the solution space of (\ref{meteq}) under the condition $k\equiv\operatorname{rank}(\sigma)<n$ for a given $\ker\sigma$ is $\frac{k(k+1)}{2}$, where $k=0,\cdots n-1$.
\end{theo}
\begin{proof}
Once the kernel of $\sigma$ is fixed, using the results and the notation of the proof of Theorem \ref{thmintdist}, we can choose the kernel of $\sigma$ to be spanned by $\omega^{(i)}=dy^{n-s+i}$, $i=1,\dots,s$. This implies that $\sigma^{\alpha\beta}=0$. There are another $\frac{k(k+1)}{2}$ unknowns $\sigma^{ij}$. Equations (\ref{closetens}) give rise to a closed system of linear PDEs for $\sigma^{ij}$ and thus the maximal dimension of the solution space is $\frac{k(k+1)}{2}$.
\end{proof}

Let us give an example of a projective structure that is not metrisable and saturates the maximal bound $\frac{n(n-1)}{2}$. Consider, in local coordinates $x^1,\dots, x^n$, the connection given by $\Gamma^1_{nn}=-\frac{x^2}{2}$, $\Gamma^2_{nn}=-\frac{x^1}{2}$ and all the other symbols equal to zero. This is a particular case of the Newtonian projective structure \cite{DunEast2016, DunGun2016} describing a classical particle moving in the $(n-1)$-dimensional space $(x^1,\dots,x^{n-1})$, $x^n$ playing the role of time, under the potential $V=x^1x^2$. This Newtonian connection admits $dx^1\wedge\cdots\wedge dx^n$ as a parallel volume form, so it is special. A careful analysis of the metrisability equations (\ref{meteq}) gives the solution $\sigma^{an}=0$, for $a=1,\dots,n$ and $\sigma^{kl}=C_{kl}$, for $k,l=1,\cdots,n-1$, where $C_{kl}=C_{lk}$ are arbitrary constants. Thus, the projective structure of this connection saturates the bound $n(n-1)/2$ for the dimension of the solution space.

\newpage
\thispagestyle{empty}
\vspace*{\fill}
\newpage
\chapter{First integrals of affine connections and Hamiltonian systems of hydrodynamic type}\label{chapterKilling}

%\section{Introduction}
The existence of a first integral of a geodesic flow of an affine connection\footnote{Since the Killing equations are invariant under special change of connections in the sense of (\ref{projequiv}), the existence of linear integrals is a property of special projective classes (c.f. above equation (\ref{decRiemann})), although not of the whole projective structure.} puts restrictions on the form of the connection. A generic connection admits no first integrals. If the connection arises from a metric, and the first integral is linear in velocities, then the metric admits a one-parameter group of isometries generated by a Killing vector field. Characterising metrics which admit Killing vectors by local tensor obstructions is a classical problem which goes back at least to Darboux \cite{Darboux1887}, and can be solved completely in two dimensions. The analogous characterisation of non-metric affine connections has not been carried over in full\footnote{The remarkable exception is the paper of Levine \cite{Levine1948} and its extension \cite{Thompson1999} where the necessary condition for the existence of a first integral was found, albeit not in a form involving the Schouten and Cotton tensors. The sufficient conditions found in \cite{Levine1948} are not all independent. Levine gives seven tensor conditions, where in fact two scalar conditions suffice.}. It is given in Theorem \ref{theo_1}, where we construct two invariant scalar obstructions to the existence of a linear first integral. A non-metric connection can (unlike a Levi-Civita connection) admit precisely two independent linear local first integrals. This case will also be characterised
by a tensor obstruction.

%\vskip5pt
As an application of our results we shall, in Section \ref{sec_hydro}, characterise one-dimensional systems of hydrodynamic type which admit a Hamiltonian formulation of the Dubrovin--Novikov type \cite{DubNov1983}. The existence of such formulation leads to an overdetermined system of PDEs, and  we shall show (Theorem \ref{theo_2}) that this system is equivalent to a condition that  a certain non-metric affine connection admits a linear first integral. This, together with Theorem \ref{theo_1} will lead to a characterisation of Hamiltonian, bihamiltonian and trihamiltonian systems of hydrodynamic type.
In Section \ref{sec_ex} we shall give examples of connections resulting from hydrodynamic type systems.% In particular we shall show that systems arising from two-dimensional Frobenius manifolds are tri-Hamiltonian.

Let us state the main results of this chapter. Let $\nabla$ be a torsion-free  affine connection of differentiability class $C^4$ on a simply connected orientable surface $\Sigma$ (so we require the transition functions of $\Sigma$ to be of class $C^6$). A curve $\gamma:\R\rightarrow \Sigma$ is an affinely parametrised geodesic if $\nabla_{\dot{\gamma}}\dot{\gamma}=0$, or equivalently if
\be
\label{flow}
\ddot{X}^a+\Gamma^{a}_{bc}\dot{X}^b\dot{X}^c=0, \quad a, b, c =1, 2
\ee
where $X^{a}=X^{a}(\tau)$ is the curve $\gamma$ expressed in local coordinates $X^a$ on an open set $U\subset \Sigma$, $\tau$ is an affine parameter, $\Gamma_{ab}^c$ are the Christoffel symbols of $\nabla$, and we use the summation convention. A linear function on $T\Sigma$ given by $\kappa=K_a(X)\dot{X}^a$ is called a {\em first integral} if $d\kappa/d\tau=0$ when (\ref{flow}) holds, or 
equivalently if
\be
\label{killing}
\nabla_{(a} K_{b)}=0.
\ee
The following theorem gives local necessary and sufficient conditions for a connection to admit one, two or three linearly independent solutions to the Killing equation (\ref{killing}). The necessary conditions involve vanishing of obstructions $I_N$ and $T$ given by (\ref{i_n}) and (\ref{tensor_obs}) -- for these to make sense the connection needs to be at least three times differentiable. 

\begin{theo}\label{theo_1}
The necessary condition for a $C^4$ torsion-free affine connection $\nabla$ on  a surface $\Sigma$ to admit a linear first integral is the vanishing, on $\Sigma$, of scalar invariants $I_N$ and $I_S$ given by (\ref{i_n}) and (\ref{i_s}), respectively. For any point $p\in \Sigma$ there exists a neighbourhood $U\subset \Sigma$ of $p$ such that conditions $I_N=I_S=0$ on $U$ are sufficient for the existence of a first integral on $U$. There exist precisely two independent linear first integrals on $U$ 
if and only if the tensor  $T$ given by (\ref{tensor_obs}) vanishes and the skew part of the Ricci tensor of $\nabla$ is non-zero on $U$. There exist three independent first integrals on $U$  if and only if the connection is projectively flat and its Ricci tensor is symmetric.
\end{theo}

This theorem will be established by constructing (Proposition \ref{tractor_prop}) a prolongation connection $D$ on a rank-3 vector bundle $\Lambda^1(\Sigma)\oplus\Lambda^2(\Sigma)$ for the overdetermined system (\ref{killing}), and restricting the holonomy of its curvature when one, two or three parallel sections exist.  In Proposition \ref{prop_rank1} we shall find all local normal forms of connections from Theorem \ref{theo_1} which admit precisely two linear first integrals. 

Finally we shall consider one-dimensional systems of hydrodynamic type. Any such system with two dependent variables $(X^1, X^2)$ and two independent variables $(x, t)$ can be written in the so-called Riemann invariants as
\be
\label{hydro_system}
\frac{\p X^1}{\p t}=\lambda^{1}(X^1, X^2)\frac{\p X^1}{\p x}, \quad 
\frac{\p X^2}{\p t}=\lambda^{2}(X^1, X^2)\frac{\p X^2}{\p x},
\ee
where $\lambda^1\neq \lambda^2$ at a generic point. This system admits a Hamiltonian formulation of the Dubrovin--Novikov type, if it can be written as
\be
\label{hami}
\frac{\p X^a}{\p t}=\Omega^{ab}\frac{\delta H}{\delta X^b},
\ee
where $H[X^1, X^2]=\int_{\R} {\mathcal H}(X^1, X^2) dx$ is the Hamiltonian of hydrodynamic type, and the Poisson structure on the space of maps $\mbox{Map}(\R, U)$ is given by
\begin{equation}\label{eqSympSt}
\Omega^{ab}=g^{ab}\frac{\p}{\p x}+{b^{ab}_c}\frac{\p X^c}{\p x}.
\end{equation}
The Jacobi identity imposes severe constraints on $g(X^a)$ and $b(X^a)$ -- see Section \ref{sec_hydro} for details.
We shall prove
\begin{theo}\label{theo_2}
The hydrodynamic-type system (\ref{hydro_system}) admits one, two or three Hamiltonian formulations with hydrodynamic Hamiltonians if and only if the affine torsion-free connection $\nabla$ defined by its non-zero components
\begin{align}
\label{connection_{AB}}
\Gamma_{11}^1&=\partial_1 \ln{A}-2B, \quad &\Gamma_{22}^2&=\partial_2\ln{B}-2A,\nonumber\\
\Gamma_{12}^1&=-\Big(\frac{1}{2}\partial_2 \ln{A}+A\Big), \quad &\Gamma_{12}^2&=-\Big(\frac{1}{2}\partial_1 \ln{B}+B\Big),
\end{align}
where
\begin{equation}\label{functionsAB}
A=\frac{\p_2 \lambda^1}{\lambda^2-\lambda^1}, \quad B=\frac{\p_1 \lambda^2}{\lambda^1-\lambda^2}, \quad\mbox{and}\quad \partial_a=\dfrac{\partial}{\partial X^a}
\end{equation}
admits one, two or three  independent linear first integrals, respectively.
\end{theo}
This Theorem, together with Theorem \ref{theo_1} leads to explicit
obstructions for the existence of a Hamiltonian formulation (\ref{hami}).

\section{Killing operator for affine connection}
Given an affine connection $\nabla$ on a surface $\Sigma$, its curvature  is defined by
\[
[\nabla_a, \nabla_b]X^c={{R_{ab}}^c}{}_d X^d.
\]
In two dimensions the projective Weyl tensor vanishes, and the curvature  can be  uniquely decomposed as
\be
\label{formula_for_rho}
{{R_{ab}}^c}{}_d= 
\delta_a{}^c\Rho_{bd}-\delta_b{}^c\Rho_{ad} +B_{ab}\delta_d{}^c,
\ee
where $\Rho_{ab}$ is the Schouten tensor related to the Ricci tensor $R_{ab}={{R_{ca}}^c}_b$ of $\nabla$ by $\Rho_{ab}=(2/3)R_{ab}+(1/3)R_{ba}$, and $B_{ab}=\Rho_{ba}-\Rho_{ab}=-2\Rho_{[ab]}$. We shall assume that $\Sigma$ is orientable, and choose a volume form\footnote{All the results are independent of the choice of volume form. In fact, we could have avoided introducing $\epsilon_{ab}$ had we chosen to use tensor densities.} $\epsilon_{ab}$ on $\Sigma$. We shall also introduce $\epsilon^{ab}$ such that $\epsilon^{ab}\epsilon_{cb}=\delta_c^a$. These skew-symmetric tensors are used to raise and lower indices according to $V^a=\epsilon^{ab}V_b$ and $V_a=\epsilon_{ba}V^b$. Then
\[
\nabla_a\epsilon_{bc}=\theta_a\epsilon_{bc},
\]
where $\theta_a=(1/2)\epsilon^{bc}\nabla_a\epsilon_{bc}$. 
%If
%we chose $\epsilon=dX^1\wedge dX^2$, then $\theta_a=\Gamma_{ab}^b$. 
Set $\beta=B_{ab}\epsilon^{ab}$.
\begin{prop}
\label{tractor_prop}
There is a one-to-one correspondence between solutions to the Killing equations (\ref{killing}), and parallel sections of the prolongation connection $D$ on a rank-three vector bundle $E=\Lambda^1(\Sigma)\oplus\Lambda^2(\Sigma)\rightarrow \Sigma$ defined by
\be
\label{tractor_con}
{\quad{{D}}_a \left(\begin{array}{c}
K_b\\ 
\mu
\end{array} \right)= 
\left(\begin{array}{c} \nabla_a K_b-\epsilon_{ab}\mu \\ 
\nabla_a\mu -\Big({\Rho^b}_a+\frac{1}{2}\beta{\delta^b}_a\Big)K_b+\mu\theta_a
\end{array} \right).}
\ee
\end{prop}
{\bf Proof.}
Dropping the symmetrisation in (\ref{killing}) implies the existence of $\mu$
such that $\nabla_a K_b=\mu\epsilon_{ab}$. Differentiating this equation
covariantly, skew-symmetrising over all indices
and using  the curvature decomposition (\ref{formula_for_rho})
together with the Bianchi identity yields the statement of the Proposition.
\koniec
%The connection $D$ is related to the standard tractor connection in projective differential geometry (see e.g. \cite{East08}). 
In the proof of Theorem \ref{theo_1} we shall find the integrability conditions for the existence of parallel sections of this connection. This will lead to a set of invariants of an affine connection $\nabla$.

\newpage
{\bf Proof of Theorem \ref{theo_1}.}
The integrability conditions $(\nabla_a\nabla_b-\nabla_b\nabla_a)\mu=0$ give the algebraic condition
\begin{equation}\label{obs}
F^aK_a+\beta\mu=0, \quad\mbox{where}\quad
F^a=\frac{1}{3}\epsilon^{ab}(L_b-\epsilon^{cd}\nabla_b B_{cd})
%A^a=\frac{1}{3}\epsilon^{ab}(L_b-\nabla_b\beta-\beta\theta_b),
\end{equation}
and $L_b\equiv \epsilon^{cd}\nabla_c\Rho_{db}$ is the Cotton tensor of $\nabla$. Geometrically, the condition (\ref{obs}) means that the 
curvature of $D$, which is a matrix, has rank at most one, and annihilates a parallel section of $D$. Applying $\nabla_a$ to the condition (\ref{obs}), and using the vanishing of (\ref{tractor_con}) leads to two more algebraic conditions
\begin{equation}\label{obs22}
{\tilde M_a}{}^{b}K_b+\tilde N_a\mu=0,
\end{equation}
where
$$
\tilde{M}_a{}^b=\nabla_a F^b+\Big({\Rho^b}_a+\frac{1}{2}{\delta^b}_a\beta\Big)\beta, \quad \tilde N_a=-F_a+\nabla_a\beta-\beta\theta_a.
$$
Multiplying the equation (\ref{obs}) by $2\theta_a$, and adding the resulting expression to (\ref{obs22}) results in the changes ${\tilde M_a}{}^{b}\rightarrow {\tilde M_a}{}^{b}+2\theta_a F^b\equiv {M_a}{}^{b}$ and $\tilde N_a\rightarrow \tilde N_a+2\theta_a\beta\equiv N_a$. We can use this freedom to get rid of $\theta^a$ from in (\ref{obs22}), yielding
\begin{equation}
{M_a}^{b}K_b+N_a\mu=0,
\label{obs2}
\end{equation}
where
\begin{align}
\label{MandN}
{M_a}^b&=\frac{1}{3}\epsilon^{bc}\epsilon^{de}(\nabla_{a}Y_{dec}-\nabla_a\nabla_cB_{de})+\beta\;{\Rho^b}_a+\frac{1}{2}\beta^2{\delta^b}_a, \\
N_a&=-F_a+\epsilon^{bc}\nabla_a B_{bc},
\end{align}
and $Y_{cdb}=\nabla_{[c}\Rho_{d]b}$. Therefore a parallel section $\Psi\equiv(K_1, K_2, \mu)^T$ of $D$ must satisfy a system of three linear algebraic equations which we write in a matrix form as
\be
\label{matrixm}
{\mathcal M}\Psi\equiv
\left(\begin{array}{ccc}
F^1 & F^2 & \beta\\
{M_1}^{1} & {M_1}^{2} & N_1\\
{M_2}^{1} & {M_2}^{2} & N_2  
\end{array} \right)
\left(\begin{array}{c} K_1 \\ 
K_2\\
\mu
\end{array} \right)=0.
\ee
A necessary condition for the existence of a non-zero parallel section $\Psi$ is therefore the vanishing of the determinant of 
the matrix ${\mathcal M}$. This gives the first obstruction which we write
as a vanishing of the relative scalar invariant
\be
\label{i_n}
I_N=\epsilon_{cd}\epsilon^{be}{M_e}^c\Big(N_bF^d-\frac{1}{2}\beta {M_b}^d\Big).
%\Big(
%M^{bc}N^aA^d-\frac{1}{2}M^{ac}M^{bd}\beta\Big).
\ee
This invariant has weight $-5$: if we replace $\epsilon_{ab}$ by
$e^{f}\epsilon_{ab}$, where $f:\Sigma\rightarrow \R$, then
$I_N\rightarrow e^{-5f} I_N$.% Thus $I_N\otimes(\epsilon_{ab}dX^a\wedge dX^b)^{\otimes 5}$ is an invariant.
The vanishing of $I_N$ is not sufficient for the existence of a non-zero parallel section. To assure sufficiency assume 
that $I_N=0$. Rewrite (\ref{obs}) and (\ref{obs2}) as
\[
V^{\alpha}\Psi_\alpha=0,\quad (D_a V^{\alpha})\Psi_\alpha=0,\quad \alpha=1, \dots, 3
\]
where $V=(F^1, F^2, \beta)$ in the formula above is a section of the dual bundle $E^*$,
and $D_a$ is the dual connection inherited from (\ref{tractor_con}). We continue differentiating, and 
adding the linear equations on $\Psi$. The Frobenius theorem tells
us that the process terminates once a differentiation
does not add any additional independent equations, as then the rank of the  
matrix of equations on $\Psi$ stabilises and does not grow. The space of parallel sections of $D$ has dimension equal to $3$ (the rank of the bundle $E$) minus the number of independent equations on $\Psi$. 
Therefore the sufficient condition for the existence of a Killing form assuming that $I_N=0$ is
\begin{equation}
\label{6th}
{\mbox{rank}}{
\left(
\begin{array}{c}
{V}\\
{D_1 V}\\
{D_2 V}\\
{D_1D_1 V}\\
D_{(1}D_{2)}V\\
D_2D_2 V
\end{array}
\right )<3.}
\end{equation}
If $I_N=0$ and  $V\neq 0$, then 
\[
cV+c_1D_1V+c_2D_2V=0,
\]
where $(c, c_1, c_2)$ are some functions on $U$, and $(c_1, c_2)$ are  not both zero. This implies that the term $D_{(1}D_{2)}V$ in 
(\ref{6th}) is a linear combination of all other terms, and can be disregarded. Now, suppose
that $D_1V=0$. Then (\ref{6th}) 
becomes $\mbox{det}(V, D_2 V, D_2 D_2 V)=0$. Equivalently, if $\quad D_2V=0$ then
(\ref{6th})  becomes $\mbox{det}(V, D_1 V, D_1 D_1 V)=0$.
%Therefore either $D_a V^{\alpha}=0$, or
%$
%V\in\mbox{span}\;(D_1 V, D_2 V)
%$, so that
%\[
%V=c_1 D_1V+c_2 D_2 V
%\]
%for some functions $(c_1, c_2)$ on $U$ not both zero. 
%Generically neither $c_1$ nor $c_2$ vanish identically, and
%the sufficient condition is given by
%\be
%\label{is0}
%\mbox{det}(D_1 V, D_2 V, D_1 D_2 V)=0.
%\ee
%If $c_2=0$, then the sufficient condition is $\mbox{det}(V, D_2 V, D_2 D_2 V)=0$, and if $c_1=0$, then %it is
%$\mbox{det}(V, D_1 V, D_1 D_1 V)=0$. Conversely if these two determinants 
%vanish, then (\ref{is0}) also vanishes regardless on whether $c_1$ and 
%$c_2$ vanish or not. Thus, in practical implementation of the Theorem
%we should define the sufficient conditions as having two components
We conclude that (\ref{6th}) is equivalent to
\be
\label{i_ss}
(\mbox{det}(V, D_1 V, D_{(1} D_{1)} V),\mbox{det}(V, D_2 V, D_{(2} D_{2)} V))=(0,0),
\ee
as it is easy to show that the condition above implies (\ref{6th}). In fact, condition (\ref{i_ss}) above gives just one independent condition: if $c_2\neq 0$, then the sufficient condition is the vanishing of the first component and if $c_1\neq 0$, then it is the vanishing of the second one. The explicit tensor expression of the obstruction (\ref{i_ss}) is given by calculating $\det(V, D_a V, D_{(b} D_{c)} V)$, which is proportional to the tensor
\be
\label{w_tensor}
W_{abc}= F_e {M_a}^e V_{(bc)}-F_e U^e_{(bc)}N_a+\beta M_{ae}U^e_{(bc)}, \quad\mbox{where}
\ee
\begin{eqnarray*}
U^b_{ca}&=&\epsilon^{bd}\epsilon^{ef}[\frac{1}{3}(\nabla_c\nabla_a
Y_{efd}-\nabla_c\nabla_a\nabla_d B_{ef})+\nabla_c( B_{ef}P_{da})]\\
&+&\frac{1}{2}\epsilon^{ef}\epsilon^{gh}\nabla_c(B_{ef}
B_{gh})\delta^b_a+\epsilon^{bd}N_a(P_{dc}+\frac{1}{2}\beta \epsilon_{cd}),\quad\mbox{and}\\
V_{ca}&=&-M_{ac}-\frac{1}{3}\epsilon^{de}(\nabla_c\nabla_d
P_{ea}-\nabla_c\nabla_a B_{de}).
\end{eqnarray*}
Therefore, the sufficient condition for the existence of a Killing $1$-form, under $I_N=0$, is
\begin{equation}\label{i_s}
I_S\equiv(W_{111},W_{222})=(0,0).
\end{equation}
%\be
%\label{i_s}
%I_s\equiv[\mbox{det}(V, D_1 V, D_1 D_1 V), \mbox{det}(V, D_2 V, D_2 D_2 V)]=[0, 0].
%\ee
%\begin{theo}
%The necessary and sufficient conditions for a symmetric affine connection 
%$\nabla$ on  a surface $U$ to admit a first integral linear in the momenta
%is the vanishing of invariants $I_N$ and $I_S$ given by (\ref{i_n}) and
%(\ref{i_s}) respectively. 
%\end{theo}
%\vskip2pt

We shall now consider the case when there exist precisely two independent solutions to the Killing equation
(\ref{killing}) (note that this situation does not arise if  
$\nabla$ is a Levi-Civita connection of some metric, 
as then the number of Killing vectors can be
$0, 1$ or $3$ - the last case being projectively flat). Therefore the rank of the matrix ${\mathcal{M}}$ in (\ref{matrixm}) is equal to one. We find that this can happens  if and only if  $\beta\neq 0$
and
\be
\label{tensor_obs}
{T_a}^{b}=0, \quad\mbox{where}\quad
{T_a}^{b}\equiv N_aF^b-\beta {M_a}^{b}.
\ee
This condition
guarantees the vanishing of all two-by-two minors of ${\mathcal M}$.

Finally, there exist three independent parallel sections of $D$ iff
the curvature of $D$ vanishes, or equivalently if the matrix 
${\mathcal M}$ vanishes. This condition is equivalent to the projective
flatness of the connection $\nabla$ together with the condition $\beta=0$.

Let us now clarify why (\ref{i_s}) is actually a scalar condition. First, notice that interpreting the tensor $T$ in (\ref{tensor_obs}) as a $2\times 2$ matrix, we can calculate its determinant, which is given by
$$
\det T=\frac{1}{2}\epsilon^{ab}\epsilon_{cd}T_a{}^c T_b{}^d=\beta I_N.
$$
Thus, condition $I_N=0$ implies that $T$ is a degenerate matrix. In this case, let $P$ be a non-zero vector field such that $P^a T_a{}^b=0$ and choose coordinates such that $P^2=0$. In these coordinates, $T_1{}^b=0$, which means that $D_1 V\propto V$ (in fact, it is easy to show that the second line of $\mathcal M$ is proportional to the first one in this case). This implies that the only non-zero component of $W$ is $W_{222}$. Therefore, $W_{abc}$ is proportional to $P_a P_b P_c$ in these coordinates, hence in any other for these are tensorial quantities. This proves the following Lemma.
\begin{lemma}
If $I_N=0$, then there exists a non-zero $1$-form $P$ such that $W_{abc}=(\text{scalar}) P_a P_b P_c$ and $P^a T_{a}{}^b=0$.
\end{lemma}
The vanishing of $(\text{scalar})$ in the above Lemma is the actual scalar condition given by (\ref{i_s}), even though this scalar does not have an explicit formula from our construction.
\koniec
{\bf Remarks}
\begin{itemize}
\item
If the connection $\nabla$ is special (i.e. the Ricci tensor is symmetric, or equivalently $\beta=0$) then $I_N=-3^{-3}\nu_5$, where 
\[
\nu_5\equiv L^aL^b\nabla_aL_b
\]
is the Liouville projective invariant \cite{Liouville1889, Bryant2009}, and the indices are raised with a parallel
volume form.  Note that, unlike $\nu_5$, the obstruction
$I_N$ is not invariant under the projective changes of connection (see 
eq. (\ref{projective_change}) in \S\ref{sec_hydro}). 

Moreover, since the Killing equations are projectively invariant, the sufficient condition (\ref{i_s}) is equivalent 
%to the vanishing of a tensor
%\[
%W_{ab}\equiv L_c{\Rho^c}_{(a}L_{b)}-(\nabla_{(a}L_{|c|})\nabla_{b)}L^c-L_c\nabla_{(a}\nabla_{b)} L^c
%\]
to the vanishing of the invariant $w_1$ constructed by Liouville for second order ODEs in 
\cite{Liouville1889}. In fact, Liouville's $w_1$ can be obtained from $W_{111}$ (if $L_1\not=0$) or from $W_{222}$ (if $L_2\not=0$) by using the connection in the same projective class annihilating the volume form given by $\epsilon_{12}=1$.
\item
Theorem \ref{theo_1} generalises a well known characterisation of metrics 
which admit a Killing vector. %as those with functionally dependent scalar 
%invariants. 
See \cite{Kruglikov2008} or \cite{DunajskiBook} where a 3 by 3 matrix
analogous to ${\mathcal M}$ has been constructed. In this case
$N=-F= \frac{1}{3}*dR$, where $R$ is the scalar curvature, and $*$ is the Hodge operator of the metric $g$. The invariant (\ref{i_n}) reduces 
to\footnote{The prolongation procedure in \cite{DunajskiBook} has been carried over in the Riemannian case. The additional subtlety in the 
Lorentzian signature arises if $\nabla R$ is a non-zero null vector. We claim that, in this case, if the metric admits a Killing vector, then it has constant curvature (and thus admits three Killing vectors). 
To see it, assume that a Lorentzian metric admits a Killing vector $K$. If $K$ is null, then the metric is flat with $R=0$.
Otherwise it  can locally be put in 
the form $dY^2-f(Y)^2dX^2$ for some $f=f(Y)$, which implies that $R$ does not depend on $X$. Imposing the condition 
$|\nabla R|^2\equiv 0$ leads to $R=$ const.}
\[
I_N:=*\frac{1}{432}dR\wedge d(|\nabla R|^2).
\]
%\item Any affine connection $\nabla$ on $\Sigma$ corresponds to a family of neutral signature anti-self-dual Einstein metrics on a certain rank-$2$ affine bundle $M\rightarrow \Sigma$, given by \cite{DunMet2015}
%\[
%%\label{main_metric}
%G=\left(d\xi_a-\left(\Gamma_{ab}^c \xi_c-\Lambda \xi_a\xi_b- 
%\Lambda^{-1}\Rho_{ba}\right)d X^b\right)\odot d X^a,
%\]
%where $\xi_a$ are local coordinates on the fibres of $M$, and $-24\Lambda$ is the Ricci scalar. These metrics admit a linear first integral iff $\nabla$ admits a projective vector field.
\end{itemize}

From the first item in the above remarks, we deduce the following corollary.

\begin{coll}
A special connection that is not projectively flat admits a Killing $1$-form if and only if $\nu_5=w_1=0$.
\end{coll}

This allows us to determine a complete method to tell whether a projective structure is metrisable, that is to say, if possible solutions to its metrisability equations are degenerate.

\begin{coll}\label{corMet}
A projective structure in $2$ dimensions is metrisable if and only if one of the following happens:
\begin{itemize}
\item its metrisability equations admit a solution space of dimension at least $2$;
\item its metrisability equations admit a unique linearly independent solution and $\nu_5$ and $w_1$ are not both zero.
\end{itemize}
\end{coll}
\begin{proof}
The first item is a result of Lemma 4.3 of \cite{Bryant2009}. The second one comes from the fact that if the metrisability equations admit a unique linearly independent solution, then (from Theorem \ref{thmKFormsDeg} and the above corollary) it is degenerate if and only if $\nu_5=w_1=0$.
\end{proof}

It is natural to ask what can be said about the unparametrised geodesic equations when the special connections of the corresponding projective structure admits a Killing form. We will show that in this case (which corresponds to $\nu_5=w_1=0$) there exists a coordinate system in which the unparametrised geodesic equations take the form $y''=A_0(x,y)$ and these coordinates correspond to those in which the connection with Thomas symbols admits a closed Killing $1$-form. Actually, the fact that $\nu_5=w_1=0$ implied the existence of coordinates such that the unparametrised geodesic equations do not involve first derivatives was already known \cite{Liouville1889} (c.f. also \cite{BabBor99}), however we give here a different and independent proof using Theorem \ref{theo_1}, providing a geometric intuition to this choice of coordinates.

First, let us understand how Thomas symbols transform under change of coordinates. The following will apply in $n$ dimensions. Let $x^\alpha, x^\beta,$ $x^\sigma, \dots$ be local coordinates, where $\alpha,\beta,\sigma,\dots\in\{1,\dots,n\}$. Now, on an overlapping chart, let $y^i, y^j, y^k,\dots$ be another coordinate system, where $i,j,k,\dots \in\{1,\dots,n\}$. Then \cite{Thomas1925}
\begin{align}\label{transfThomas}
\Pi^\sigma_{\alpha\beta}&=\Pi^i_{jk} \dfrac{\partial y^j}{\partial x^\alpha}\dfrac{\partial y^k}{\partial x^\beta}\dfrac{\partial x^\sigma}{\partial y^i}+\dfrac{\partial x^\sigma}{\partial y^i}\dfrac{\partial^2 y^i}{\partial x^\alpha\partial x^\beta}-\frac{1}{n+1}\delta^\sigma_\beta \dfrac{\partial}{\partial x^\alpha}\ln \left|\dfrac{\partial y}{\partial x}\right|-\nonumber\\
&-\frac{1}{n+1}\delta^\sigma_\alpha \dfrac{\partial}{\partial x^\beta}\ln \left|\dfrac{\partial y}{\partial x}\right|,
\end{align}
where $\left|\dfrac{\partial y}{\partial x}\right|$ is the modulus of the Jacobian of the change of coordinates $y^i=y^i(x^\alpha)$.

Notice that, even though the property of Christoffel symbols being traceless is coordinate dependent (as $\Gamma^a_{ab}$ is not a $1$-form), the transformation (\ref{transfThomas}) assures that tracelessness is preserved for Thomas symbols. In fact, if $\Pi^i_{ji}=0$, then
$$
\Pi^\beta_{\alpha\beta}=\dfrac{\partial x^\beta}{\partial y^i}\dfrac{\partial^2 y^i}{\partial x^\alpha\partial x^\beta}-\dfrac{\partial}{\partial x^\alpha}\ln\left|\dfrac{\partial y}{\partial x}\right|=0.
$$

Now, let us see how the Killing equations $\nabla^\Pi_{(\alpha}K_{\beta)}=0$ transform. The symbol $\nabla^\Pi$ corresponds to the covariant derivative associated to traceless Christoffel symbols (i.e., the Thomas symbols). A direct calculation shows that
\begin{align*}
2\nabla^\Pi_{(\alpha}K_{\beta)}&\equiv\dfrac{\partial}{\partial x^\alpha}K_\beta+\dfrac{\partial}{\partial x^\beta}K_\alpha-2\Pi^\sigma_{\alpha\beta}K_\sigma=\\
&=\dfrac{\partial y^i}{\partial x^\alpha}\dfrac{\partial y^j}{\partial x^\beta}\left(\dfrac{\partial}{\partial y^i}K_j+\dfrac{\partial}{\partial y^j}K_i-2\Pi^k_{ij}K_k\right),
\end{align*}
if 
\begin{equation}\label{coordKillingThomas}
K_i=\left|\dfrac{\partial y}{\partial x}\right|^{\frac{2}{n+1}}\dfrac{\partial x^\alpha}{\partial y^i}K_\alpha.
\end{equation}

\begin{theo}\label{thmLiouInv}
The ODE $y''=A_0(x,y)+A_1(x,y) y'+A_2(x,y) (y')^2+A_3(x,y) (y')^3$ defining a non-flat projective structure admits coordinates $(X,Y)$ such that $Y_{XX}=f(X,Y)$ if and only if $\nu_5=w_1=0$. Moreover, this is also equivalent to the fact that the connection with Thomas symbols admits a Killing $1$-form given by $dX$.
\end{theo}
\begin{proof}
First, notice that if $Y_{XX}=f(X,Y)$, then the connection with Thomas symbols admits $dX$ as Killing form. Conversely, if this connection admits $dX$ as Killing form, then all the Thomas symbols are vanish except $\Pi^2_{11}$ and thus its projective structure is represented by an ODE in the required form.
The ``only if" part is immediate. For the ``if" part, choose coordinates such that $L_1\not=0$, which is possible since the projective structure is not flat. Then $\nu_5=w_1=0$ implies that special representatives of the projective structure admit a Killing $1$-form. In particular, the connection with Thomas symbols admits a Killing $1$-form $K_\alpha$. 

If the Killing form is closed, then it can be written as $dX$ under a volume-preserving transformation and this is the required coordinate system by the first part of the proof. 

If the Killing form is not closed, then it can be written in some coordinate system that we call again $(x,y)$ as $g(x,y)dx$, for some function $g$ such that $\partial_y g\not=0$. Then, the coordinate transformation $(X,Y)=(x,\int^y g(x,y')^{-3/2}dy')$ makes the Killing form closed, according to the rule (\ref{coordKillingThomas}). This, along with what we proved above, concludes the demonstration of the theorem.
\end{proof}

In the metric case, a Levi-Civita connection cannot admit precisely two local linear first integrals, as $\beta$ (which is proportional to the skew part of the Ricci tensor) vanishes. In the following proposition we shall explicitly find all local normal forms of non-metric affine connections which admit two first integrals.
\begin{prop}\label{prop_rank1}
Let $\nabla$ be an affine connection on a surface $\Sigma$ which admits exactly two non-proportional linear first integrals which are independent at some point $p\in \Sigma$. Local coordinates $(X, Y)$ can be chosen on an open set $U\subset \Sigma$ containing $p$ such that 
\begin{align}
\label{class_2}
\Gamma_{12}^1&=\Gamma_{21}^1= \frac{c}{2}, \quad \Gamma_{11}^2=\frac{P_X}{Q},\quad \Gamma_{12}^2=\Gamma_{21}^2 =\frac{P_Y+Q_X-cP}{2Q},\nonumber \\
\Gamma_{22}^2&=\frac{Q_Y}{Q},
\end{align}
and all other components vanish, where $c$ is a constant equal to $0$ or $1$,
and $(P, Q)$ are arbitrary functions of $(X, Y)$.
\end{prop}
{\bf Proof.}
Let the one-forms $K$ and $L$ be two solutions to the Killing equation.
If $K$ is closed, then there exist local coordinates $(X, Y)$ on $U$ such
that $K=dX$, and the corresponding first integral is $\dot{X}$. Therefore
$\ddot{X}=0$ and the connection components $\Gamma^1_{ab}$ vanish. Let the 
second solution of the Killing equation be of the form $L=PdX+QdY$ for some functions
$(P, Q)$. Imposing
\[
\frac{d}{d\tau}(P\dot{X}+Q\dot{Y})=0
\]
yields the non-zero components of the connection given by (\ref{class_2}) with $c=0$. If $dK\neq 0$, then coordinates $(X, Y)$ can be chosen so that $K=e^Y dX$.
The condition $d/d\tau(e^Y\dot{X})=0$ gives 
 $\Gamma^{1}_{12}=1/2$. Imposing the existence of the second integral
$(P\dot{X}+Q\dot{Y})$ yields the connection (\ref{class_2}) with $c=1$.
\koniec
Note that in both cases the ODEs for the unparametrised geodesics 
also admit a first integral,
given by $e^{-cY}({P+Y'Q})$, where $'=d/dX$. Conversely if a 2nd order ODE cubic in $Y'$ representing projective equivalence class\footnote{See Definition \ref{DefProjSt}.} $[\nabla]$ of affine connections admits a first integral linear in $Y'$, then $[\nabla]$ contains a connection of the form (\ref{class_2}) with $c=0$. To see it consider a second order  ODE of the form $(P+Y'Q)'=0$, where $(P, Q)$ are arbitrary functions of $(X, Y)$ and write it in the form
\be
\label{projective_ode}
Y''=\Gamma_{22}^1 (Y')^3+(2\Gamma_{12}^1-\Gamma_{22}^2) (Y')^2+(\Gamma_{11}^1-
2\Gamma_{12}^2)Y'-\Gamma_{11}^2.
\ee
Equation (\ref{projective_ode}) arises from eliminating the affine parameter 
$\tau$ between the two ODEs (\ref{flow}). Thus its integral curves are unparametrised geodesics of the affine connection $\nabla$.

%%%%%%%%%%%%%%%%%%%%%%%%%%%%%%%%%%%%%%%%%%%%%%%%%%%%%%%%%%%%%%%%%%%%%%%%%%%%
\section{Hamiltonian systems of hydrodynamic type}
\label{sec_hydro}
An $n$-component $(1+1)$ system of hydrodynamic type has the form $\p_t u^a={v^a}_b(u)\p_xu^b$, where $u^a=u^a(x, t)$ and $a, b=1, \dots, n$. From now on we shall assume that $n=2$ and that the matrix $v$ is diagonalisable at some point with distinct eigenvalues, in which case there always exists (in a neighbourhood of this point) two distinct functions (called the Riemann invariants) $X^1$ and $X^2$ of $(u^1, u^2)$ such that the system is diagonal, i.e. takes the form (\ref{hydro_system}) for some $\lambda^a(X^b)$. The existence of Riemann invariants is shown in Section \ref{AppeRiemInv} of the appendix.% and can be linearised by a hodograph transformation interchanging dependent $(X^1, X^2)$ and independent $(x, t)$ coordinates.

The hydrodynamic type system is said to admit a local Hamiltonian formulation with a Hamiltonian of hydrodynamic type \cite{DubNov1983, Ferapontov1991}, if there exists a functional $H[X^1, X^2]=\int_\R {\mathcal H}(X^1, X^2) dx$, where the density ${\mathcal H}$ does not depend on the derivatives of $X^a$ and  such that  (\ref{hami}) holds
%\[
%\frac{\p u^a}{\p t}=\Big(g^{ab}\frac{\p}{\p x}+{b^{ab}_c}\frac{\p u^c}{\p x}\Bi%g)\frac{\delta H}{\delta u^b}
%\]
for some functions $g^{ab}(X)$ and $b^{ab}_c(X)$.
If the matrix $g^{ab}$  is non-degenerate, then the Poisson bracket
\[
\{F, G\}=\int_\R \frac{\delta F}{\delta X^a} 
\Big(g^{ab}\frac{\p}{\p x}+{ b^{ab}_c}\frac{\p X^c}{\p x}\Big)
\frac{\delta G}{\delta X^b}dx 
\]
is skew-symmetric if $g^{ab}$ is symmetric and the metric
$g=g_{ab} dX^a dX^b$, where $g_{ab}g^{bc}={\delta_a}^c$, is parallel with respect to the connection with Christoffel symbols $\gamma_{ab}^c$ defined by $b^{ab}_c=-g^{ad}\gamma^b_{dc}$. The Jacobi identity then holds iff
the metric $g$ is flat, and the connection  defined by 
$\gamma_{ab}^c$ is torsion-free. The hydrodynamic type systems which 
admit a Hamiltonian of hydrodynamic type possess infinitely many
Poisson commuting first integrals, and are integrable in the Arnold--Liouville sense \cite{Tsarev1985}.

{\bf Proof of Theorem \ref{theo_2}.}
It was shown in \cite{DubNov1983} that a hydrodynamic 
type system in Riemann invariants is Hamiltonian in the sense defined above
if and only if there exists a flat diagonal metric
\be
\label{flat_met}
g=k^{-1} d(X^1)^2+f^{-1} d(X^2)^2
\ee
on a surface $U$ with local coordinates $(X^1, X^2)$
such that
\be
\label{eq1k}
\p_2 k+2Ak=0, \quad\p_1 f +2B f=0,
\ee
where $f, k$ are functions of $(X^1, X^2)$, and $(A, B)$ are given by
(\ref{functionsAB}). Flatness of the metric $g$ yields
\be
\label{eq2k}
(\p_2 A+A^2)f+(\p_1 B+B^2)k+\frac{1}{2}A\p_2 f+\frac{1}{2}B\p_1 k=0.
\ee
We verify that  equations (\ref{eq1k}) and (\ref{eq2k}) are equivalent
to  the Killing equations
(\ref{killing})
for an affine torsion-free connection $\nabla$ on $U$ defined by
(\ref{connection_{AB}}) where $K_1=Af, K_2=Bk$.
\koniec
Computing the relative invariants $I_N$ and $I_S$ gives explicit but 
complicated (albeit perfectly manageable by MAPLE) obstructions given in terms
of $(\lambda^1, \lambda^2)$ and their derivatives of order up to $6$.
These obstructions, together with the tensor (\ref{tensor_obs}) and the 
Cotton tensor of $\nabla$
characterise
Hamiltonian, bihamiltonian and trihamiltonian systems of hydrodynamic type. The trihamiltonian systems have been previously characterised by Ferapontov in \cite{Ferapontov1991} in terms of  two differential forms he called $\omega$ and $\Omega$. We shall now show how Ferapontov's formalism relates to our connection (\ref{connection_{AB}}). We shall find that $\Omega$ is proportional to the skew-symmetric part of the Ricci tensor of $\nabla$, and $\omega$ is the volume form of  the (generically) unique Lorentzian metric on $U$ which shares its unparametrised geodesics with $\nabla$.

We say that a symmetric affine connection $\nabla$ is {\em metric} if it is the Levi-Civita connection of some (pseudo)-Riemannian metric. Recall that an affine connection $\nabla$ is {\em metrisable} if it shares its unparametrised geodesic with some metric connection. Thus in the metrisable case there exists a one-form $\Upsilon$ and a metric $h$ such that the Levi-Civita connection of $h$ is given by
\be
\label{projective_change}
{\Gamma}^a_{bc}+\delta^a_b\Upsilon_c+\delta^a_c\Upsilon_b,
\ee
where ${\Gamma}^a_{bc}$ are the Christoffel  symbols of $\nabla$. Not all affine connections on a surface are metrisable. The necessary and sufficient conditions for metrisability have been found in \cite{Bryant2009} and in Corollary \ref{corMet}. 
\begin{prop}
The connection (\ref{connection_{AB}})  from Theorem \ref{theo_2} is generically not metric but is metrisable by the metric
\be
\label{metric_h}
h=AB (dX^1)\odot (dX^2).
\ee
\end{prop}
{\bf Proof.} 
The connection is generically not metric, as its Ricci tensor $R_{ab}$ is in general not symmetric. The skew part of $R_{ab}$  is given by
$$
(R_{21}-R_{12})dX^1\wedge dX^2=3d\Upsilon,%\mbox{where}\quad
$$
where
\begin{equation}\label{upsilon}
\Upsilon=\Big(\frac{1}{2}\partial_1 \ln{B}+B\Big)dX^1+\Big(\frac{1}{2}\partial_2 \ln{A}+A\Big)dX^2.
\end{equation}
The unparametrised geodesics of this connection are integral curves of a 2nd order ODE
\be
\label{ode}
Y''=(\p_X Z)Y'-(\p_Y Z) (Y')^2, \quad \mbox{where}\quad Z=\ln{(AB)},
\ee
and $(X^1, X^2)=(X, Y)$. The ODE (\ref{ode}) is also the equation for unparametrised geodesics of the pseudo-Riemannian metric (\ref{metric_h}) (it can be found directly by solving the metricity equations as in \cite{CD16}). The Levi-Civita connection of $h$ is given by (\ref{projective_change}), where $\Upsilon$ is given by (\ref{upsilon}). Therefore, connection (\ref{connection_{AB}}) is projectively equivalent to a metric connection.
\koniec
{\bf Remarks}
\begin{itemize}
\item The pseudo-Riemannian metric (\ref{metric_h}) 
depends only on the product $AB$, so the transformation
$(A\rightarrow \gamma A, B \rightarrow \gamma^{-1}B)$,
where $\gamma=\gamma(X^a)$ is a non-vanishing function,
does not change 
unparametrised geodesics. It corresponds to a projective change of connection 
(\ref{projective_change}) by a one-form
\[
\Upsilon=
\Big((1-\gamma^{-1})B+\frac{1}{2}\p_1\ln{\gamma}\Big)dX^1
+\Big((1-\gamma)A-\frac{1}{2}\p_2\ln{\gamma}\Big)dX^2.
\]
This transformation can be used to set $R_{[ab]}$ to zero, but it does not 
preserve  (\ref{hydro_system}).
\item As the Ricci tensor $R_{ab}$ is in general not symmetric,
the connection (\ref{connection_{AB}}) does not annihilate any volume form on $\Sigma$ which is parallel w.r.t $\nabla$. Therefore the Killing equations (\ref{killing}) do not imply the existence of a Killing vector for the metric $h$.
\item
The two-form $\Omega$ in Theorem $9$ of \cite{Ferapontov1991} equals $2d\Upsilon$, while $\omega$ in \cite{Ferapontov1991} is given by the volume form of $h$. In the trihamiltonian case the connection $\nabla$ is projectively flat. Equivalently  the metric (\ref{metric_h}) has constant Gaussian curvature, i.e. 
\be
\label{liouville}
(AB)^{-1}\p_1\p_2\ln{(AB)}=\mbox{const}.
\ee 
This is the Liouville equation from Section 5 in \cite{Ferapontov1991}.
\item If $n\geq 3$, there is always a discrepancy between the number of equations for a Killing tensor of any given rank and a number of conditions for a hydrodynamic-type system to admit a Hamiltonian formulation. Therefore Theorem \ref{theo_2} does not generalise to higher dimensions in any 
straightforward way.
\end{itemize}
\section{Examples}
\label{sec_ex}
In the examples below we set $X^1=X, X^2=Y$.
\subsection*{Example 1} Consider an affine connection (\ref{connection_{AB}}) corresponding to a system of hydrodynamic type with 
\[
A=cX+Y,  \quad B=X+cY,\quad \mbox{where}\quad c=\mbox{const.}
\]
This connection admits a parallel volume form iff $c=0$ or $c=1$.
If $c=0$ then the connection is projectively flat, and so the
system of hydrodynamic type is trihamiltonian.
Calculating the obstruction (\ref{tensor_obs}) yields
\begin{align*}
T&=\frac{8c^2(c^2-9)}{9(cX+Y)^3(X+cY)^3}\Big(dY\otimes\p_Y-dX\otimes\p_X+\\
&+\frac{X+cY}{cX+Y}dY\otimes \p_X-\frac{cX+Y}{X+cY} dX\otimes \p_Y\Big).
\end{align*}
Therefore, if $c=3$ or $c=-3$ then the connection admits precisely two linear first integrals, so the system is bihamiltonian. Finally for any $c$ not equal to $0, \pm 3$ the system admits a unique Hamiltonian.
\subsection*{Example 2} 
%Gas dynamics
%\[
%u_t+uu_x+f(v)v_x=0, \quad v_t+(uv)_x=0.
%\]
One dimensional non-linear elastic medium is governed by the system of PDEs \cite{OlNut1988,Sheftel1986}
\[
u_t=h^2(v)v_x, \quad v_t=u_x,
\]
where $h(v)$ is a function characterising the type of fluid. This system is Hamiltonian with ${\mathcal H}=u^2/2+F(v)$, where $F''=h^2$. We find the Riemann invariants $(X, Y)$ such that
\[
u=X+Y, \quad  v=G(X-Y),
\]
where $G'h(G)=1$ and $\lambda^1=-\lambda^2=\frac{1}{G'}$. Therefore
$
A=-B=-G''/(2G')
$
and we find $\beta=0$, so that the Ricci tensor of the associated connection (\ref{connection_{AB}}) is symmetric. In particular, Theorem \ref{theo_1} implies that the system can not admit precisely two Hamiltonian structures.

The projective flatness (\ref{liouville}) of the connection (\ref{connection_{AB}}) reduces to $(\ln{A^2})''=\mbox{const}.A^2$ which can be solved explicitly, and leads to a four-parameter family of trihamiltonian systems. The singular solution $A=1/(2z)$ corresponds to the Toda equation $v_{tt}=(\ln{v})_{xx}$.
%and thus $h(v)=\exp{(4\arctan{(\exp{v})})}$.
\subsection*{Example 3} 

We consider the system of hydrodynamic type (\ref{hydro_system}) with
\[
\lambda^1= -\lambda^2=(X-Y)^n (X+Y)^m.
\]
Examining the conditions of Theorem \ref{theo_1} for the resulting connection (\ref{connection_{AB}}) we find that this system is always bihamiltonian. It is trihamiltonian iff $nm(n^2-m^2)=0$.

\subsection*{Example 4. Frobenius manifolds}

A basic introduction to Frobenius manifolds and the corresponding hydrodynamic-type systems is given in Appendix \ref{appeFrobMan}, where it is shown that systems of the form (\ref{HTsysFrob}) are trihamiltonian corresponding to a $3$-parameter family of metrics for (\ref{eqSympSt}) given by (\ref{trimetric}).

\subsection*{Example 5. Zoll connections}
A Riemannian metric $h$ on a surface $\Sigma$ is {\em Zoll} if all
geodesics are simple closed curves of equal length. A two-dimensional sphere
admits a family of axisymmetric Zoll metrics given by
\be
\label{zoll_1}
h=(F(X)-1)^2dX^2+\sin^2{X}dY^2,
\ee
where $(X, Y)$ are spherical polar coordinates on $\Sigma=S^2$, and $F:[0, \pi]\rightarrow [0, 1]$ is any function such that $F(0)=F(\pi)=0$ and $F(\pi-X)=-F(X)$. A projective structure $[\nabla]$ on $\Sigma$ is {\em Zoll} if its unparametrised geodesics are  simple closed curves. The general projective structure admitting a projective vector field, and close to the flat structure of the round sphere is given by the second order ODE \cite{LeMason2002}
\begin{equation}
\label{zoll_2}
Y''=A_3 (Y')^3+A_2(Y')^2+A_1Y',
\end{equation}
where
\begin{align*}
A_1&=\frac{F'}{F-1}-2\cot{X},\\
A_2&=\frac{H'\sin{X}\cos{X} -2H}{\cos{X}(F-1)},\\
A_3&=-\frac{(H^2+1)\sin{X}\cos{X}}{(F-1)^2},
\end{align*}
where $F=F(X)$ is as before, and $H=H(X)$ satisfies $H(0)=H(\pi)=H(\pi/2)=0$,
and $H(\pi-X)=H(X)$. The metric case (\ref{zoll_1}) arises 
if $H=0$. A general 
connection $\nabla$ in this projective class with $\beta\neq 0$ will not admit
even a single first integral. We use Theorem \ref{theo_1} together with
(\ref{projective_ode}) to verify that the following choice of the 
representative connection 
\begin{equation}\label{zoll_3}
\Gamma_{11}^1=A_1, \quad \Gamma_{22}^1=A_3, \quad \Gamma_{12}^1=\Gamma_{21}^1=\frac{1}{2}A_2 
\end{equation}
admits a first integral for any $F$ and $H$. 
To find a (necessarily non-metric) Zoll connection 
with precisely two linear first integrals we 
use Proposition (\ref{prop_rank1}) and match the connection 
(\ref{zoll_3}) with the connection
(\ref{class_2}) (with the roles of $X$ and $Y$ reversed). 
%we find two cases
%with integrability conditions
%\[
%A_1+(\ln{A_3})'=0, \quad\mbox{or} \quad (A_3)'+A_1A_2-A_2-1=0.
%\]
%The first case can be solved to yield, for a given $H$, a one-parameter 
This, for any given $H$, leads to a one-parameter family of examples
\[
F=1+c(H^2+1)\cot{X}
\]
which does not satisfy the boundary conditions. The existence of a non-metric Zoll structure on $S^2$ with precisely two first integrals is an interesting open problem.

\chapter{Concluding remarks}

This thesis presented six main results. In chapter \ref{chapSDYM}, we show that the vortex equations proposed by Manton are symmetry reductions of ASDYM equations with various symmetry and gauge groups. Two immediate consequence are the characterisation of twistor-integrable cases and the calculation of the metric on a submanifold of the moduli space by the Samols' localisation method. This moduli space is actually comprehensive for Taubes vortices but not necessarily for the other types. A determination of the full moduli space would require analysis of the general solution of Manton's equation upon suitable boundary conditions. The main difficulty consists of the fact that Taubes' analysis on the ordinary vortex equations relied on the convexity of the energy function and thus is not immediately generalisable to the other cases.

The second main result is presented in chapter \ref{chapVorLike}, where a modified version of the Ginzburg--Landau theory is presented as a means to produce more integrable cases besides Taubes vortices on hyperbolic surfaces. The existence of such cases was evidenced by two vortex solutions presented by Dunajski, which are included in the model proposed. A detailed analysis to seek Painlev\'e-integrable cases was carried out in a particular class of equations and showed the existence of four integrable vortex-like solitons. One of them is the Taubes vortex on hyperbolic space itself, plus both solutions proposed by Dunajski. The fourth case is new and, as in Dunajski's cases, gives rise to solutions involving the third Painlev\'e transcendent. Existence of solutions to the model proposed in general is not established. In fact, the explicit solutions come from the assumption that the vortices are radially symmetric. Existence of non-symmetric solitons is an open problem and boils down to proving existence of solution of the sinh-Gordon and Tzitzeica PDEs under suitable boundary conditions -- basically, requiring finiteness of the energy functional, isolation of the zeroes of the Higgs field and integrality of its winding number around each zero, that is to say, existence of a holomorphic gauge.

In chapter \ref{chapPainMet}, two main results were presented. Firstly, it is shown that the projective structures defined by the Painlev\'e equations are metrisable if and only if the equation admits a first integral linear in first derivatives for PIII and PV or if the projective structure is flat, which includes a special case of PVI. It is also shown that the projective structures of all six Painlev\'e equations admit a degenerate solution to their metrisability equations. In the second half of the chapter, we show that the dimension of the space of degenerate solutions to the metrisability equations in $n$ dimensions, once the kernel is fixed, is $n(n-1)/2$. It is believed, though not yet proved, that the condition of fixed kernel is not necessary for the result to hold. It is shown that the kernel of a degenerate solution defines a Frobenius-integrable distribution and the induced submanifold is totally geodesic with respect to the original projective structure, thus defining an induced structure. The analysis of deformations of totally geodesic submanifolds of projective structures and their induced structure might provide geometrical tools to the generalisation of the result.

Chapter \ref{chapterKilling} introduced the last two main results of this work. We solve, locally, the problem of existence of Killing forms for general affine connections on the tangent space of a $2$-manifold and show that the existence depends only on the vanishing of two scalars, while further conditions control the number of Killing forms admitted. Moreover, it is shown that the problem of existence of a Hamiltonian description of hydrodynamic-type systems with two components boils down to the problem of existence of Killing forms of a particular type of affine connection constructed from the data of the hydrodynamic-type system. Consequently, all the conditions for the existence of Hamiltonian descriptions can be derived from the conditions for the existence of Killing forms for that connection. This chapter connects with chapter \ref{chapPainMet} in that a solution of the Killing equations for special connections implies the existence of a degenerate solution of the metrisability equations associated to the corresponding projective structure. This allowed us to add two simple conditions -- consisting of the above-mentioned  two scalars -- for the existence of non-degenerate solutions to the metrisability equations in $2$ dimensions.

\newpage
\thispagestyle{empty}
\vspace*{\fill}
\newpage

\appendix

\chapter{Construction of symmetric gauge fields}\label{secConstSym}
%%%%%%%%
In this appendix we summarise a general procedure to construct local gauge fields that are invariant under the action of a given symmetry group. This procedure is described in more details in \cite{ForgMan1980}. We start by setting some notation. Let $S$ and $G$ be Lie groups corresponding to symmetry and gauge groups, respectively. Let $M$ be a $d$-dimensional manifold admitting an action of the group $S$ with orbits of codimension $d'<d$. We assume that there exist local coordinates $x^\mu$, $\mu=1,\dots,d$ such that the submanifold $x^i=constant$, $i=1,\dots,d'$ are invariant under the action of $S$ and that this action is transitive, this means that these submanifolds are locally homeomorphic to a homogeneous space $S/R$, where $R$ is the little group of a chosen point $p_0$ in the submanifold. We use indices $i,j,k,\dots$ for the range of values $1,\dots,d'$ and $\mu,\nu,\sigma,\dots$ for the range $1,\dots,d$.

Let $A_\mu$ be a gauge field with gauge group $G$, locally a $G$-valued $1$-form. We say that $A$ is $S$-equivariant if the lifted action of $S$ on $A$ preserves $A$ up to gauge transformations. Infinitesimally, this is stated in the following form: given a vector field $\eta$ generator of the action $S$ on $M$, there exists a scalar $\mathfrak g$-valued function $W$ such that
\begin{equation}\label{equivdef}
\mathcal L_{\eta}A_\mu=D_\mu W,
\end{equation}
where $\mathcal L$ is the Lie derivative and $D$ is the covariant derivative defined by $A$.

In this appendix, we are going to describe, without details, the procedure used in \cite{ForgMan1980} to solve, for $A$, the system formed of (\ref{equivdef}) for all symmetry generators $\eta$. Basically, the procedure consists in extending the components of (\ref{equivdef}) corresponding to $S/R$ to the whole symmetry group $S$. Fixing the $x^i$'s and looking just at the coordinates of $S/R$, this allows the consistency conditions of the system to be solved by rather simple expressions for the $W$'s that can be even gauged out as fields on $S$, symplifying the extended equations (\ref{equivdef}) over $S$ for $A$. Once the extended equations for $A$ are solved, one has to make sure that $A$ -- which is now a field over $S$ -- can be interpreted as a gauge field on $S/R$. This means that the dependency of $A$ on the complementary coordinates of $S/R$ over $S$ (the $y^\omega$ below) can be gauged away. The conditions for this to be possible are the so called \textit{consistency equations} introduced later in a convenient form.    

As mentioned above, a submanifold $x^i=constant$, now denoted by $\mathcal H$, is locally homeomorphic to $S/R$ and this allows us to relate the coordinates of $S/R$ to the coordinates $x^{\alpha=d'+1,\dots,d}$ of $\mathcal H$. Let $y^{\omega}$ (late Greek indices) denote coordinates on $R$ and $y^{\alpha}$ (early Greek indices) denote coordinates on $S/R$, so that $(y^\alpha,y^\omega)$ forms a coordinate system of $S$. If $S$ has dimension $N$ and $R$ codimension $N'$, then $\alpha=1,\dots,N'$ and $\omega=N'+1,\dots,N$. In practise, this is realised as follows. For any $s\in S$, let $s_0(y^\alpha)\in Rs(y^\alpha)$ be a fixed element in the right coset $Rs$ varying smoothly with $y^\alpha$.  Then, $s$ can be uniquely written in the form $s=r(y^\omega)s_0(y^\alpha)$ for some $r\in R$. Let $p_0$ be a fixed point in $\mathcal H$ having $R$ as stabiliser. By transitivity, for any $p\in\mathcal H$, there exists $s\in S$ mapping $p$ to $p_0$. Clearly, this map lifts to the coset $Rs(y^\alpha)$ and thus we can associate to $p$ the coordinate $y^\alpha$ of the coset  $Rs$. This is how we endow $\mathcal H$ with coordinates $y^\alpha$. Notice that $d'+N'=d$.

Let $J_m$, $m=1,\dots,N$, be the generators of the Lie algebra of $S$ with structure constants $f_{mnp}$ so that $[J_m,J_n]=f_{mnp}J_p$. We assume that $J_m$ for $m>N'$ generate the Lie algebra of $R$. 

Construct the right-invariant $1$-forms denoted by $\xi_m$ and defined by the equation
$$
-ds s^{-1}=\xi_m J_m.
$$

Let $\Phi_n$, $n=1,\dots, N$ be local $\mathfrak g$-valued functions on $M$ so that $\Phi_{n>N'}$ are constants and $\Phi_{n\leq N'}$ depend only on $x^{i\leq d'}$. These Higgs fields correspond to the projection of the symmetric gauge potential on $S$ along the symmetry generators (c.f. equation (\ref{eqDefHiggs}) below). Solve the following constraint equations which, in a suitable gauge, take the form
\begin{align}\label{AppeConstEq}
&[A_i,\Phi_n]=0,\; i=1,\dots,\;D', n>N'\nonumber\\
&f_{mnp}\Phi_p+[\Phi_m,\Phi_n]=0, \; m=1,\dots,N,\; n>N'.
\end{align}
The second equation means that the $-\Phi_{n>N'}$ generate an $R$ Lie subalgebra in the gauge group $G$ while the first equation means that the $A_i$'s are gauge fields with gauge group being the little group of $R$ in $G$ and thus henceforth only gauge transformations in this little group are allowed.

Start by defining a gauge potential on $S$ as 
\begin{equation}\label{eqDefHiggs}
A_\omega=\Phi_m\left(\xi_m\right)_\omega \text{ and } A_\alpha=\Phi_m\left(\xi_m\right)_\alpha.
\end{equation}
By construction of the constraint equations, the components $A_\omega$ will be pure gauge of the form $A_\omega=\partial_\omega r r^{-1}$, for some $r=r(y^\omega)$ in the above mentioned little group. Moreover, under a gauge transformation by $r^{-1}$, $A_\omega$ will be set to zero while $A_\alpha$ will become $y^\omega$-independent. Now we can consistently define a gauge potential on $R/S$ and thus in the whole $M$ as
\begin{align*}
A_i=a_i(x^1,\dots,x^{d'}), \; i=1,\dots, d',\\
A_\alpha=r^{-1}\Phi_m r \left(\xi_m\right)_\alpha,\; \alpha=1,\dots,N'=d-d',
\end{align*}
where $a_i$ are arbitrary components of a gauge field that are constant along the $\mathcal H$'s. This is the most general ansatz for a gauge potential $A$ satisfying the $S$-equivariant condition (\ref{equivdef}).

\chapter{The group $G_{C_0}$}
\label{appeGroup}

%\setcounter{equation}{0}
%\label{Appendix_sec}
%\def\theequation{\thesection{A}\arabic{equation}}
%Everything that is said about the group $G_C$ in Section \ref{secgroup} is of %course true for the group $G_{C_0}$ upon replacing $C$ by $C_0$. 
In this appendix, we denote by $J_m^{C_0}$ ($m=1,2,3$) the generators of the Lie-algebra $\mathfrak g_{C_0}$ given by (\ref{lie_alg_rep}) with $C$ replaced by $C_0$. Obviously, all properties of $G_{C_0}$ are valid for $G_C$, and vice-versa, upon changing $C$ into $C_0$.

 A parametrisation of $G_{C_0}$ is given by 
$$
\hspace{-0.2in}K \!\!  = \!\! \left(\begin{array}{c c}
  \!\!\! e^{i(\kappa_3-\kappa_2)/2}\cos(\sqrt{-C_0}\,\kappa_1/2) & \! -\frac{1}{\sqrt{-C_0}}e^{i(\kappa_3+\kappa_2)/2}\sin(\sqrt{-C_0}\,\kappa_1/2) \\ 
  \! \sqrt{-C_0}e^{-i(\kappa_3+\kappa_2)/2}\sin(\sqrt{-C_0}\,\kappa_1/2) & \! e^{-i(\kappa_3-\kappa_2)/2}\cos(\sqrt{-C_0}\,\kappa_1/2)
 \end{array}\right),
$$
%$a_1=e^{i(\gamma-\beta)/2}\cos(\sqrt{-C_0}\alpha/2)$ and %$a_2=-\frac{1}{\sqrt{-C_0}}e^{i(\gamma+\beta)/2}\sin(\sqrt{-C_0}\alpha/2)$, 
where $0\leq\kappa_3\leq 4\pi$, $0\leq\kappa_2\leq 2\pi$, $0\leq\kappa_1<\pi/\sqrt{-C_0}$ if $C_0<0$ and $\kappa_1\geq 0$ if $C_0\geq 0$.

The coordinate $\kappa_3$ parametrises the $U(1)$ fibres of the fibration
$G_{C_0}\rightarrow N=G_{C_0}/U(1)$. In the proof of Proposition \ref{prop1} we  need expressions relating the
local coordinates $(z, \ov{z})$ on $N$ to $(\kappa_1, \kappa_2)$ on 
$G_{C_0}/U(1)$. Let $p\in N$ be a point corresponding
to the coordinate $z=0$. Consider the group action (\ref{actionC}) such that the RHS is $0$. This
gives a system of two equations for $(z, \ov{z})$ with a solution
\begin{equation}\label{zk_trans}
z=\frac{1}{\sqrt{-C_0}}\tan(\sqrt{-C_0}\, \kappa_1/2)e^{i\kappa_2}.
\end{equation}
Note that $\frac{1}{\sqrt{-C_0}}\tan(\sqrt{-C_0}\, \kappa_1/2)\geq 0$ regardless of the sign of $C_0$. 
The formula (\ref{zk_trans}) is  well defined for $C_0=0$ upon taking the limit $C_0\to 0$. The coordinate $\kappa_3$ of $G_{C_0}$ parametrises the stabiliser of $p\in N$, which is a $U(1)$ subgroup generated by $J_3^{C_0}$.

The right-invariant one-forms $\chi_1,\chi_2,\chi_3$  such that 
$$
(dK) K^{-1}+\sum_{m=1}^3 \chi_m \otimes  {J_m^{C_0}}=0,
$$ 
and the left-invariant vector fields $\eta^1, \eta^2, \eta^3$ on $G_{C_0}$ are given by
%\[
%-ds \; s^{-1}=\sum_{m=1}^3 \chi_m \otimes  {J_m}^{C_0}
%\]
%$$
%s=\left(\begin{array}{c c}
%  e^{i(\gamma-\beta)/2}\cos(\sqrt{-C_0}\alpha/2) & %-\frac{1}{\sqrt{-C_0}}e^{i(\gamma+\beta)/2}\sin(\sqrt{-C_0}\alpha/2) \\ 
%  \sqrt{-C_0}e^{-i(\gamma+\beta)/2}\sin(\sqrt{-C_0}\alpha/2) & %e^{-i(\gamma-\beta)/2}\cos(\sqrt{-C_0}\alpha/2)
% \end{array}\right),
%$$ 
\begin{align}
\label{l_forms}
\chi_1=&\left(\frac{1}{\sqrt{-C_0}}\sin(\sqrt{-C_0}\, \kappa_1)\cos\kappa_3\, d\kappa_2+\sin\kappa_3\, d\kappa_1\right),\nonumber\\
\chi_2=&\left(-\cos\kappa_3 d\kappa_1+\frac{1}{\sqrt{-C_0}}\sin\kappa_3 \sin(\sqrt{-C_0}\,\kappa_1) d\kappa_2\right),\nonumber\\
\chi_3=&\left(-d\kappa_3+\cos(\sqrt{-C_0}\,\kappa_1)d\kappa_2\right)
\end{align}
and
\begin{align}
\label{vect_appendix}
\eta_1&=-\sin{\kappa_2}\partial_{\kappa_1}-\frac{\sqrt{-C_0}}{\tan(\sqrt{-C_0}\,\kappa_1)}\cos\kappa_2 \partial_{\kappa_2}-\frac{\sqrt{-C_0}}{\sin(\sqrt{-C_0}\,\kappa_1)}\cos\kappa_2\partial_{\kappa_3}, \nonumber\\
\eta_2&=\cos\kappa_2\partial_{\kappa_1}-\frac{\sqrt{-C_0}}{\tan(\sqrt{-C_0}\,\kappa_1)}\sin\kappa_2 \partial_{\kappa_2}-\frac{\sqrt{-C_0}}{\sin(\sqrt{-C_0}\,\kappa_1)}\sin\kappa_2\partial_{\kappa_3}, \nonumber\\
\eta_3&=-\partial_{\kappa_2}.
\end{align}

\chapter{Frobenius manifolds and hydrodynamic-type systems}\label{appeFrobMan}

In this appendix we present the concept and basic definitions of Frobenius manifolds and how they give rise to hydrodynamic-type systems.

\begin{defn}[Frobenius algebra]
A commutative associative $\mathbb C$-algebra $(A,+,\cdot)$ with unity $e$ is a \textit{Frobenius algebra} if it is endowed with a bilinear symmetric non-degenerate inner product $<\, ,\,> : A \times A\to\mathbb C$ such that $<a\cdot b,c>=<a,b\cdot c>$.
\end{defn}

\begin{defn}[Frobenius manifold]
A smooth manifold $M$ is a Frobenius manifold if each fibre of the tangent bundle $TM$ is endowed with a Frobenius algebra structure such that
\begin{itemize}
\item[1.] The inner product $<\, ,\,>$ is a smooth flat metric $\eta$ on $M$. Its Levi-Civita connection will be denoted by $\nabla$.
\item[2.] The unity $e$ is a smooth section of $TM$ and is parallel: $\nabla e=0$.
\item[3.] The symmetric $3$-tensor $c(u,v,w)=<u\cdot v,w>$ is smooth and its  covariant derivative $\nabla c$ is a symmetric $4$-tensor.
\item[4.] There exists a conformal Killing vector field $E$ such that $\nabla\nabla E=0$, $\mathcal L_Ee=-e$ and $\mathcal L_E c=c$.
\end{itemize}
\end{defn}
Condition $4.$ means that the $1$-parameter group of diffeomorphisms generated by $E$ acts as \textit{rescalings} on the Frobenius algebras: $u\cdot v\mapsto k u\cdot v$, $e\mapsto k^{-1} e$, $k\in\mathbb C^*$.

Commutativity along with condition $3.$ implies that locally there exists a complex function $F$, called \textit{free energy} or \textit{prepotential}, such that $c=\nabla\nabla\nabla F$. Notice that $F$ is determined up to an arbitrary quadratic function. 

Let $(M,\cdot,e,E,\eta)$ be an $n$-dimensional Frobenius manifold, where $e$ is the unit element of the Frobenius algebra, $E$ is the Euler vector field and $\eta$ is a flat metric. By definition, in \textit{flat coordinates} $\{t^1,\dots,t^n\}$ -- whose indices will be denoted by $a,b,c,\dots$ -- the components of the metric $\eta^{ab}$ are constant and, by condition 2., we can set $e=\dfrac{\partial}{\partial t^1}$. This implies $\eta_{ab}=c_{1ab}$ by condition 3. We shall use the notation $\partial_a\equiv\dfrac{\partial}{\partial t^a}$.

Condition $4.$ implies that there exists a constant matrix $Q=(q^a_b)$ and a constant vector $r^a$ such that $E=(q^a_b t^b+r^a)\partial_a$. If $Q$ is diagonalisable with eigenvalues $d_1,\dots,d_n$, then we can perform a linear transformation to rewrite the Euler vector field as
\begin{equation}
E=\sum_ad_a t^a \partial_a+\sum_{a|d_a=0}r^a\partial_a
\end{equation}
without changing the above properties. By condition 4., $d_1=1$.
%We can always rescale the Euler vector field by a constant so that condition $4.$ is satisfied, namely, we suppose $d_1\not=0$ and rescale to fix $d_1=1$.

Since $E$ is a conformal Killing vector (condition $4.$) and $\eta_{ab}=c_{1ab}=\partial_1\partial_a\partial_b F$, we conclude that $F$ should rescale with a fixed weight $d_F$, namely
\begin{equation}\label{resF}
\mathcal L_E F(t)=d_F F(t)+A_{ab}t^a t^b+B_a t^a+C,
\end{equation}
where $A_{ab}$, $B_a$ and $C$ are arbitrary constants. The quadratic terms play no role when we take third order derivatives.
 
The one parameter group of transformations generated by $E$ acts on $\eta_{ab}dt^a dt^b=\partial_1\partial_a\partial_b F dt^a dt^b$ as $\eta_{ab}dt^a dt^b\mapsto \eta_{ab}dt^a dt^b+\epsilon (d_F-1)\eta_{ab}dt^adt^b+O(\epsilon^2)$, so that
\begin{equation}
\mathcal L_E \eta_{ab}=(d_F-1)\eta_{ab}.
\end{equation}

Associativity of $\cdot$ means that $F$ should satisfy the so called \textit{associativity equation} in flat coordinates
\begin{equation}\label{aseq}
0=c^c_{ab}c_{cde}-c^c_{ad}c_{cbe}=\partial_1\partial_f\partial_b F \, \eta^{fc} \,\partial_c\partial_d\partial_e F-\partial_f\partial_1\partial_d F\, \eta^{fc}\, \partial_c\partial_b\partial_e F.
\end{equation}

Equations (\ref{resF}) and (\ref{aseq}) together are known as the \textit{WDVV equations}.

\begin{defn}[Dubrovin connection]
The \textit{Dubrovin connection} is the following deformation of $\nabla$
$$
\tilde\nabla^{(z)}_u v=\nabla_u v+z u\cdot v,
$$
for arbitrary $z\in\mathbb C$. 
\end{defn}
Flatness of this connection for any $z$ is equivalent to associativity of $\cdot$ (coming from the $z^2$ terms of the curvature) and condition $3.$ (from the $z$ terms of the curvature).

The flat coordinates $\tilde t^a$ of $\tilde\nabla^{(z)}$ are the solutions to the overdetermined system of differential equations $\tilde\nabla^{(z)}_{\partial_{\tilde t^a}}d\tilde t=0$, i.e.,
\begin{equation}\label{flatdefcoord}
\partial_a\partial_b\tilde t=z c^c_{ab}\partial_c\tilde t.
\end{equation}

The consistency conditions of this system correspond to the vanishing of the curvature of $\tilde\nabla^{(z)}$, so that it forms a Lax pair of (\ref{aseq}) with spectral parameter $z$. Once the consistency conditions are satisfied, there exists a fundamental set of solutions $\tilde t^1_a=\partial_a\tilde t^1,\dots,\tilde t^n_a=\partial_a\tilde t^n$ of (\ref{flatdefcoord}).

\section{Intersection form}

We can define another metric $g$ on $M$ by the inner product
$$
(\omega,\sigma)^*\equiv\iota_E(\omega\cdot \sigma).
$$
The ${}^*$ indicates that it is a product defined on $T^*M$ and $\iota_E$ is the contraction of the vector field $E$ with the $1$-form $\omega\cdot\sigma$. The product $\cdot$ of two $1$-forms is induced by the product of their dual vector fields through $\eta$. In flat coordinates, the components of $(\; ,\; )$ are
\begin{equation}\label{intformflat}
g^{ab}=(dt^a,dt^b)^*=E^c c_c^{ab}.
\end{equation}
It is worth pointing out here that we still raise and lower indices using $\eta$, unless explicitly stated.

\begin{defn}
$g$ is called the \textit{intersection form}.
\end{defn}

Until the end of this section, the indices $i,j,k,\dots$ will denote abstract indices.

\begin{defn}[Contravariant Levi-Civita connection]
If $g$ is a metric and its Levi-Civita connection is given by the Christoffel symbols $\Gamma^i_{jk}$ then the contravariant Levi-Civita connection are the symbols $\Gamma^{ji}_k$ given by
$$
\Gamma^{ji}_k\equiv -g^{jl}\Gamma^i_{lk}.
$$
\end{defn}

\begin{defn}[Flat pencil]
We say that two metrics $g_1$ and $g_2$ form a \textit{flat pencil} if the metric $g^{ij}=g^{ij}_1+\lambda g^{ij}_2$ is flat for arbitrary $\lambda$ and its contravariant Levi-Civita connection $\Gamma^{ij}_k$ is given by
$$
\Gamma^{ij}_k=\Gamma^{ij}_{1k}+\lambda \Gamma^{ij}_{2k},
$$
where $\Gamma^{ij}_{\alpha k}=-g_\alpha^{il}\Gamma^j_{\alpha lk},\, \alpha=1,2,$ and $\Gamma^j_{\alpha lk}$ is the Levi-Civita connection of the metric $g_\alpha$.
\end{defn}
Let the covariant derivatives corresponding to $g$ and $g_\alpha$ be $\nabla$ and $\nabla_\alpha$, respectively and let $\nabla^i\equiv g^{ij}\nabla_j$ and $\nabla^j_{\alpha}\equiv g^{ij}_\alpha\nabla_{\alpha j}$; then forming a flat pencil amounts to saying that $\nabla^i _{\alpha}=\nabla^i _{1}+\lambda \nabla^i_{2}$.

\begin{lemma}\label{lemmainteta}\cite{Dubrovin1996}
The intersection form and the flat metric $\eta$ form a flat pencil.
\end{lemma}

\section{Hierarchy of hydrodynamic type}

Define the quantities $h_a(t,z)=\eta_{ab}\tilde t^b$ from the flat coordinates of the Dubrovin connection. In order to recover the flat coordinates $t$ when $z=0$, we choose the following normalisation $h_a(t,0)=t_a$. We formally expand these functions as power series in $z$
$$
h_a(t,z)=\sum_{p\geq 0} h_{a,p}(t)z^p,
$$
so that the components $h_{a,p}$ are determined recursively from (\ref{flatdefcoord}):
\begin{equation}\label{rech}
\partial_b\partial_c h_{a,p+1}(t)=c_{bc}^d\partial_d h_{a,p}(t)
\end{equation}
along with 
\begin{equation}\label{hin}
h_{a,0}(t)=t_a.
\end{equation}

Now we consider the coordinates $t^a$ as functions in the loop space of $M$, $\mathcal L(M)=\{t:S^1\to M\}$. We parametrise $S^1$ by a variable $X\in [0,2\pi]$. Thus, $t^a=t^a(X)$.

From the flat metric $\eta^{ab}$, we define the Poisson bracket of hydrodynamic type by 
$$
\{t^a(X),t^b(Y)\}=\eta^{ab}\delta^\prime(X-Y).
$$
\begin{prop}\label{propPencil}
Two flat metrics $g_1$ and $g_2$ define, as above, compatible Poisson brackets of hydrodynamic type $\{\, ,\,\}_1$ and $\{\, ,\,\}_2$ if and only if they form a flat pencil.
\end{prop}

Consider the functionals 
$$
H_{a,p}(t)=\int_0^{2\pi} h_{a,p+1}(t(X))dX.
$$
We introduce an infinite sequence of ``times" $T^{a,p}$ in order to define the flows of hydrodynamic type
\begin{equation}\label{hierHT}
\partial_{T^{a,p}}t^b=\{t^b(X),H_{a,p}\}.
\end{equation}

From (\ref{rech}) and (\ref{hin}), it is natural to mark the variable $T^{1,0}=X$ so that the flow $(a,p)=(1,0)$ is trivial.

By using associativity (\ref{aseq}) one can show \cite{Dub93}
\begin{equation}\label{flow0}
\partial_{T^{a,0}}t^b=c^b_{ac}(t)\partial_X t^c
\end{equation}
and
\begin{equation}
\partial_{T^{a,p}}t^b=\nabla^ch_{a,p}\partial_{T^{c,0}}t^b.
\end{equation}
This relation implies that if all the flows for $p=0$ admit the same Riemann invariants, then all the Hierarchy admits the same Riemann invariants. In particular, this is true for systems with $n=2$ components.

\begin{lemma}\cite{Dub93}
$\{H_{a,p},H_{b,q}\}=0$.
\end{lemma}

\begin{theo}\cite{Dub98}
The hydrodynamic-type systems (\ref{flow0}) from Fro\-benius manifolds are bihamiltonian and the corresponding hydro\-dy\-na\-mic-type Poisson brackets are induced by the flat metric $\eta$ and by the intersection form.
\end{theo}
\begin{proof}
From Lemma \ref{lemmainteta} and Proposition \ref{propPencil}, we just need to show that there exists a Hamiltonian associated to the Poisson bracket $\{\,,\,\}_2$ of the intersection form. It can be checked explicitly, by using associativity of $c$, that this Poisson bracket satisfies the conditions of Lemma 1 of \cite{Ferapontov1991} and thus admits a Hamiltonian.
\end{proof}

\section{2D Frobenius manifolds}

In $2$ dimensions, the prepotential $F$ automatically satisfies the associativity condition (\ref{aseq}) and can be determined solely from (\ref{resF}). It is possible to classify all Frobenius manifolds according to $6$ families.

\begin{theo}\label{thm2dFM}\cite{Dubrovin1996}
A $2$-dimensional Frobenius manifold with diagonalisable $Q$ admits flat coordinates such that its prepotential is one of the following
\begin{align*}
F(t_1,t_2)&=\frac{1}{2}t_1^2 t_2+K t_2^k, \; k=\frac{3-d}{1-d}=\frac{3-d}{d_2},\; d\not=-1,1,3,\\ 
F(t_1,t_2)&=\frac{1}{2}t_1^2 t_2+K t_2^2\ln t_2, \; d_2=2 , \\
F(t_1,t_2)&=\frac{1}{2}t_1^2 t_2+K \ln t_2, \; d_2=-2 , \\
F(t_1,t_2)&=\frac{1}{2}t_1^2 t_2+K e^{\frac{2}{r}t_2}, \; d_2=0 ,\;r\not=0, \\
F(t_1,t_2)&=\frac{1}{2}t_1^2 t_2, \; d_2=0  ,\;r=0, \\
F(t_1,t_2)&=\frac{1}{2}t_1^2 t_2+\frac{c}{6}t_1^3+\frac{K}{6}t_2^3, \; d_2=1 ,
\end{align*}
where in the fourth and fifth cases the Euler vector field is $E=t_1\partial_1+r\partial_2$ and in the others, $E=t_1\partial_1+d_2 t_2\partial_2$. $K$ and $c$ are arbitrary constants.
\end{theo}
In such coordinates $\eta_{ab}=\delta_{a+b, 3}$, except in the last case, in which $\eta_{11}=c$.
Let us construct the first flow of the hierarchy (\ref{hierHT}). Let us write the prepotentials of the first five cases in the general form $F(t_1,t_2)=\frac{1}{2}t_1^2 t_2+f(t_2)$. After suppressing the index ${}^{(2,0)}$, the flow $(a,p)=(2,0)$ is given by,
\begin{align}\label{HTsysFrob}
\partial_T t^1=f'''(t^2)\partial_X t^2 \nonumber\\
\partial_T t^2=\partial_X t^1.
\end{align}
The characteristic velocities are $\lambda^1=\sqrt{f'''(t^2)}$ and $\lambda^2=-\sqrt{f'''(t^2)}$ and the Riemann invariants, 
\begin{equation}\label{Rieminv2D}
R^i= t^1+\int\lambda^i dt^2.
\end{equation} 
Using Theorems \ref{theo_1} and \ref{theo_2}, it is possible to show that such systems are trihamiltonian. The most general metric giving rise to Hamiltonian structures is given by the general solution of the Killing equations for the connection defined in Theorem \ref{theo_2},
\begin{align}\label{trimetric}
&\left[\left(C_1+C_2 R^1 +C_3 \left(R^1\right)^2\right)\lambda^1\right]^{-1}dR^1dR^1+\nonumber \\
+&\left[\left(C_1+C_2 R^2+C_3 \left(R^2\right)^2\right)\lambda^2\right]^{-1} dR^2dR^2,
\end{align}
where $C_i$ are arbitrary constants. One can check that the three fundamental metrics in this solution form a flat pencil. 

The metrics given by $C_2=C_3=0$ and $C_1=C_3=0$ are $\eta$ and the intersection form, respectively. In fact, the intersection form is given by (\ref{intformflat}). %, $g^{11}=E^2 f'''(t^2)$, $g^{12}=E^1=t^1$,  $g^{22}=E^2$. 
Notice that we can choose Riemann invariants such that the Euler vector field is $E=R^1 \partial_{R^1}+R^2 \partial_{R^2}$ (c.f. Lemma \ref{lemmaRiemInv} along with the fact that these Riemann invariants are canonical coordinates). In such coordinates, the intersection form reads $g^{ij}=2\lambda^i R^i \delta^{ij}$, which is the metric corresponding to $C_1=C_3=0$ above, while the inner product $\eta$ becomes $\eta^{ij}=4\lambda^i\delta^{ij}$, corresponding to $C_2=C_3=0$.

\textit{Remark}: we have not studied hydrodynamic-type systems coming from the two last Frobenius manifolds of Theorem \ref{thm2dFM} because the flows constructed from the fifth one have two identical characteristic velocities while the last one has vanishing $A$ and $B$ for every flow%\footnote{Recall that commuting flows mean that $A$ and $B$ are the same.}
. These cases are not taken into account in the present work.

\section{The third Hamiltonian structure}

The metric corresponding to $C_1=C_2=0$ in (\ref{trimetric}) is given in terms of the Frobenius data by $h^{ab}=g^{ac}g^b_c=g^{ac}g^{bd}\eta_{dc}$. As mentioned below equation (\ref{trimetric}), this metric is flat in that context. 

However, given an arbitrary Frobenius manifold with flat metric $\eta^{ab}$ and intersection form $g^{ab}$, it still makes sense to define another metric by $h^{ab}$ as above and whose curvature we briefly analyse now. Its contravariant Levi-Civita connection $\tilde\Gamma^{bc}_a$ is the solution to
\begin{align*}
\partial_a h^{bc}=\tilde \Gamma^{bc}_a+\tilde \Gamma^{cb}_a\\
h^{ad}\tilde\Gamma^{bc}_d=h^{bd}\tilde\Gamma^{ac}_d,
\end{align*}
which reads
\begin{equation}\label{contconnh}
\tilde\Gamma^{bc}_a=g^c_d\Gamma^{bd}_a+g^b_d\Gamma^{dc}_a=E^e\left(c^c_{ef}\mathcal R^f_g+c^d_{eg}\mathcal R^c_d\right)c_a^{bg},
\end{equation}
where $\Gamma^{bc}_a=\mathcal R^c_d c^{bd}_a$ is the Levi-Civita connection of the intersection form $g^{ab}$ and $\mathcal R^c_b\equiv\left(\frac{d-1}{2}\delta^c_b+\nabla_bE^c\right)$. We have used associativity in the second equality.

The curvature of this connection does not vanish in general. In fact, it does not vanish in $2$ dimensions when the identity $e$ is not null, i.e. $\eta_{11}\not=0$ (sixth case of Theorem \ref{thm2dFM}), but it vanishes otherwise and forms a flat pencil with $\eta$ and $g$ yielding the trihamiltonian structure mentioned in the previous section. Let us understand why.

By raising the last three indices of its Riemann tensor $\tilde R$ with $h$, we can write it in terms of the contravariant Levi-Civita connection
$$
\tilde R_a\,^{bcd}=h^{be}\left(\partial_e\tilde\Gamma^{dc}_a-\partial_a\tilde\Gamma^{dc}_e\right)+\tilde\Gamma^{bd}_e\tilde\Gamma^{ec}_a-\tilde\Gamma^{bc}_e\tilde\Gamma^{ed}_a.
$$
From (\ref{contconnh}) and associativity the last two terms cancel out, however the first two terms
\begin{equation}\label{fistterms}
\partial_e\tilde\Gamma^{dc}_a-\partial_a\tilde\Gamma^{dc}_e
\end{equation} 
do not necessarily vanish. For instance, for the last Frobenius manifold of Theorem \ref{thm2dFM}, we have 
$$
\partial_2\tilde\Gamma_1^{12}-\partial_1\tilde\Gamma_2^{12}=K c,
$$
so $h^{ab}$ is not flat. But for the other five cases, $h^{ab}$ is flat and forms a flat pencil with $\eta^{ab}+\lambda g^{ab}$ for any $\lambda$.

Keeping the same notation, we state the following theorem.
\begin{theo}
The tensor $h^{ab}+\lambda \eta^{ab}+\mu g^{ab}$ does not degenerate in any open set and its contravariant Levi-Civita connection is $\tilde\Gamma^{ab}_c+\mu \Gamma^{ab}_c$, for any $\lambda, \mu$.

The metric $h^{ab}$ is flat (and forms a flat pencil with $\lambda \eta^{ab}+\mu g^{ab}$) if and only if (\ref{fistterms}) vanishes. This is precisely the case for $2$-dimensional Frobenius manifolds of prepotential $F(t_1,t_2)=\frac{1}{2}t_1^2 t_2+f(t_2)$.
\end{theo}
\begin{proof}
Let $P^{ab}=h^{ab}+\lambda \eta^{ab}+\mu g^{ab}$. Let us first prove that $P^{ab}$ is non-degenerate on an open dense subset of the Frobenius manifold. We can write $g^{ab}=E^c c_c^{ab}=t_1 \eta^{ab}+ \tilde g^{ab}(t_2,\dots,t_n)$, for some symmetric tensor $\tilde g^{ab}$. Thus
$$
h^{ab}+\lambda \eta^{ab}+\mu g^{ab}=\eta^{ab}\left(t_1^2+\mu t_1+\lambda\right)+2 t_1\tilde g^{ab}+\mu\tilde g^{ab}+\lambda \tilde g^{ac}\tilde g^{ad}\eta_{cd}.
$$
Since it is a polynomial in $t_1$ and $\eta^{ab}$ is non-degenerate, it does not degenerate in an open set for any fixed values of $\lambda$ and $\mu$.

Now, by the definition of the contravariant connections of $\eta^{ab}$ and $g^{ab}$,
$$
\partial_c P^{ab}=\tilde\Gamma^{ab}_c+\mu \Gamma^{ab}_c+\tilde\Gamma^{ba}_c+\mu \Gamma^{ba}_c.
$$
The relation
$$
P^{dc}\left(\tilde\Gamma^{ab}_c+\mu \Gamma^{ab}_c\right)=P^{ac}\left(\tilde\Gamma^{db}_c+\mu \Gamma^{db}_c\right)
$$
is obtained by associativity. This concludes the proof of the first statement.

Flatness of $h^{ab}$ was discussed above. The fact that it implies that $h^{ab}$ and $\lambda \eta^{ab}+\mu g^{ab}$ form a flat pencil is a matter of calculation similar to what was done above, bearing in mind the associativity property (\ref{aseq}) and flatness of $\lambda \eta^{ab}+\mu g^{ab}$ (c.f. Lemma \ref{lemmainteta}).
\end{proof}

\section{Riemann invariants}\label{AppeRiemInv}
%\appendix
We explain here how to calculate Riemann invariants for a hydrodyna\-mic-type system. This will also prove that hydrodynamic-type systems with $2$ components always admit Riemann invariants.

Consider the hydrodynamic-type system
$$
w^i_T=v^i_jw^j_X.
$$
Under a change of coordinates $w^i=w^i(t)$ it transforms as
$$
t^i_T=\left(J^{-1}\right)^i_j v^j_k J^k_l t^l_X.
$$
where $J^i_j=\partial_{t^j}w^i$,

We immediately conclude that for a system to admit Riemann invariants it must have a diagonalisable matrix $v^i_j$ and the so called characteristic velocities $\lambda^i$ are its eigenvalues. Now suppose that $v$ is diagonalisable so that $P^i_j v^j_k \left(P^{-1}\right)^k_l$ is a diagonal matrix. The Riemann invariants $R^i$, if they exist, are given as solutions of 
\begin{equation}\label{calcR}
\partial_{w^j}R^i=P^i_j
\end{equation}
up to redefining $P$ in a way that we explain now. In fact, if we define the $1$-forms $\omega^i=P^i_j dw^j$, then it is necessary and sufficient that
\begin{equation}\label{Riemanncond}
\omega^i\wedge d\omega^i=0
\end{equation} 
in order to be able to define local coordinates $R^i$ such that $\omega^i=f^i(R) dR^i$, for some non-vanishing function $f^i(R)$. Notice that we have the freedom to rescale the lines of the matrix $P$ so that $f^i=1$ and thus we find (\ref{calcR}).

The existence of Riemann invariants for hydrodynamic-type systems of $n=2$ components is trivial since (\ref{Riemanncond}) is identically satisfied. 

\textit{Remark}: define the vector fields $V_i=\left(P^{-1}\right)^k_i\partial_{w_k}$, which are the eigenvectors of $v$. The existence of Riemann invariants means that the distribution defined by any set of $n-1$ vector fields $V_i$ is integrable. In other words, the flow of any such distribution generates a submanifold $R^i(w)=\text{constant}$.

\subsection*{Canonical coordinates}
We say that a point $t\in M$ in a Frobenius manifold $M$ is \textit{semisimple} if the Frobenius algebra $T_tM$ is semisimple (i.e. has no nilpotents). This is an open property: every semisimple point admits a neighbourhood of semisimple points.

\begin{lemma}\cite{Dubrovin1996}
In a neighbourhood of a semisimple point, there exist local coordinates $u^1,\dots,u^n$ such that
$$
\partial_i\cdot \partial_j=\delta_{ij}\partial_i\, , \quad \text{ where }\quad \partial_i=\partial_{u^i}.
$$
\end{lemma}
Such coordinates are called \textit{canonical} and can be found as independent solutions of the system
\begin{equation}\label{cancoord}
\partial_d u c^d_{ab}=\partial_au\partial_bu.
\end{equation}
This equation means that $\partial_au$ is an eigenvalue of the matrix $\left(c_{ab}^d\right)_{b,d=1,\dots,n}$ with eigenvector $\left(\partial_d u\right)_{d=1,\dots,n}$. From what has been said in the first part of this appendix, we conclude that the canonical coordinates are actually the Riemann invariants for the hydrodynamic-type system (\ref{flow0}) and the characteristic velocities are $\partial_a u^i$. Notice that this is consistent with (\ref{Rieminv2D}).

\begin{lemma}\label{lemmaRiemInv}\cite{Dubrovin1996}
There exist canonical coordinates in a neighbourhood of a semisimple point such that the Euler vector field $E$ reads $E=\sum_{i=1}^n u^i \partial_i$.
\end{lemma}

\bibliography{BiblioArticle}

\begin{thebibliography}{10}

\bibitem{ARS1980}
M.J. Ablowitz, A.~Ramani, and H.~Segur.
\newblock {A connection between nonlinear evolution equations and ordinary
  differential equations of P-type. I}.
\newblock {\em Journal of Mathematical Physics}, 21:715--721, 1980.

\bibitem{AmbOl1988}
J.~Ambj{\o}rn and P.~Olesen.
\newblock Anti-screening of large magnetic fields by vector bosons.
\newblock {\em Physics Letters B}, 214:565 -- 569, 1988.

\bibitem{BabBor99}
M.V. Babich and L.A. Bordag.
\newblock {Projective differential geometrical structure of the Painlevé
  equations}.
\newblock {\em Journal of Differential Equations}, 157:452 -- 485, 1999.

\bibitem{Baptista2014}
J.M. Baptista.
\newblock Vortices as degenerate metrics.
\newblock {\em Letters in Mathematical Physics}, 104:731--747, 2014.

\bibitem{Bogomolny1976}
E.B. Bogomolny.
\newblock {Stability of classical solutions}.
\newblock {\em Sov. J. Nucl. Phys.}, 24:449, 1976.

\bibitem{Bryant2009}
R.~Bryant, M.~Dunajski, and M.~Eastwood.
\newblock Metrisability of two-dimensional projective structures.
\newblock {\em J. Differential Geom.}, 83:465--500, 2009.

\bibitem{Cald2014}
D.M.J. Calderbank.
\newblock Integrable background geometries.
\newblock {\em SIGMA: Symmetry, Integrability and Geometry: Methods and
  Applications}, 10, 2014.

\bibitem{casale07}
G.~Casale.
\newblock {The Galois groupoid of Picard-Painlev{\'e} VI equation}.
\newblock {\em {Research Institute for Mathematical Sciences}}, B2:15--20,
  2007.
\newblock Proceedings of the French-Japanese Conference Algebraic, Analytic and
  Geometric Aspect of Complex Differential Equations and their Deformations.
  Painlev{\'e} Hierarchies.

\bibitem{Contatto2017}
F.~Contatto.
\newblock {Integrable Abelian vortex-like solitons}.
\newblock {\em Physics Letters B}, 768:23 -- 29, 2017.

\bibitem{ConDor2015}
F.~Contatto and D.~Dorigoni.
\newblock {Instanton solutions from Abelian sinh-Gordon and Tzitzeica
  vortices}.
\newblock {\em Journal of Geometry and Physics}, 98:429 -- 445, 2015.

\bibitem{FCMD2015}
F.~Contatto and M.~Dunajski.
\newblock {First integrals of affine connections and Hamiltonian systems of
  hydrodynamic type}.
\newblock {\em Journal of Integrable Systems}, 1:xyw009, 2015.

\bibitem{ConDun2017}
F.~Contatto and M.~Dunajski.
\newblock Manton’s five vortex equations from self-duality.
\newblock {\em Journal of Physics A: Mathematical and Theoretical}, 50:375201,
  2017.

\bibitem{CD16}
F.~Contatto and M.~Dunajski.
\newblock {Metrisability of Painlevé equations}.
\newblock {\em Journal of Mathematical Physics}, 59:023507, 2018.

\bibitem{Conte2013}
R.~Conte and M.~Musette.
\newblock {Introduction to the Painlevé property, test and analysis}.
\newblock {\em AIP Conference Proceedings}, 1562:24--29, 2013.

\bibitem{Darboux1887}
G.~Darboux.
\newblock {\em Leçons sur la théorie générale des surfaces}.
\newblock Gauthier-Villars, 1887.

\bibitem{DorDunMan2013}
D.~Dorigoni, M.~Dunajski, and N.S. Manton.
\newblock Vortex motion on surfaces of small curvature.
\newblock {\em Annals of Physics}, 339:570 -- 587, 2013.

\bibitem{Dub93}
B.~Dubrovin.
\newblock Topological conformal field theory from the point of view of
  integrable systems.
\newblock In L.~Bonora, G.~Mussardo, A.~Schwimmer, L.~Girardello, and
  M.~Martellini, editors, {\em {Integrable Quantum Field Theories}}, volume 310
  of {\em NATO ASI Series}, pages 283--302. Springer US, 1993.

\bibitem{Dubrovin1996}
B.~Dubrovin.
\newblock {\em {Geometry of 2D topological field theories}}, pages 120--348.
\newblock Springer Berlin Heidelberg, 1996.

\bibitem{Dub98}
B.~Dubrovin.
\newblock {Flat pencils of metrics and Frobenius manifolds}.
\newblock In {\em Proceedings of 1997 Taniguchi Symposium ``Integrable Systems
  and Algebraic Geometry", editors M.-H Saito, Y. Shimizu and K. Ueno}, pages
  47--72. World Scientific, 1998.

\bibitem{DubNov1983}
B.A. Dubrovin and S.P. Novikov.
\newblock {Hamiltonian formalism of one-dimensional systems of the hydrodynamic
  type and the Bogolyubov--Whitham averaging method}.
\newblock {\em Dokl. Akad. Nauk SSSR}, 270:781--785, 1983.

\bibitem{DunajskiBook}
M.~Dunajski.
\newblock {\em {Solitons, Instantons and Twistors}}.
\newblock Oxford Graduate Texts in Mathematics. Oxford University Press, 2010.

\bibitem{Dunajski2012}
M.~Dunajski.
\newblock {Abelian vortices from sinh-Gordon and Tzitzeica equations}.
\newblock {\em Physics Letters B}, 710:236 -- 239, 2012.

\bibitem{DunEast2016}
M.~Dunajski and M.~Eastwood.
\newblock Metrisability of three-dimensional path geometries.
\newblock {\em European Journal of Mathematics}, 2:809--834, 2016.

\bibitem{DunGun2016}
M.~Dunajski and J.~Gundry.
\newblock {Non-relativistic twistor theory and Newton--Cartan Geometry}.
\newblock {\em Communications in Mathematical Physics}, 342:1043--1074, 2016.

\bibitem{East2008}
M.~Eastwood and V.~Matveev.
\newblock {Metric connections in projective differential geometry}.
\newblock In M.~Eastwood and W.~Miller~Jr., editors, {\em {Symmetries and
  Overdetermined Systems of Partial Differential Equations}}, volume 144 of
  {\em The IMA Volumes in Mathematics and its Applications}, pages 339--350.
  Springer New York, 2008.

\bibitem{Ferapontov1991}
E.V. Ferapontov.
\newblock {Hamiltonian systems of hydrodynamic type and their realization on
  hypersurfaces of a pseudo-Euclidean space}.
\newblock {\em Journal of Soviet Mathematics}, 55:1970--1995, 1991.

\bibitem{FordyGib1980}
A.P. Fordy and J.~Gibbons.
\newblock {Integrable nonlinear Klein--Gordon equations and Toda lattices}.
\newblock {\em Communications in Mathematical Physics}, 77:21--30, 1980.

\bibitem{ForgMan1980}
P.~Forgacs and N.S. Manton.
\newblock Space-time symmetries in gauge theories.
\newblock {\em Comm. Math. Phys.}, 72:15--35, 1980.

\bibitem{Gambier1910}
B.~Gambier.
\newblock Sur les équations différentielles du second ordre et du premier
  degré dont l'intégrale générale est \`a points critiques fixes.
\newblock {\em Acta Mathematica}, 33:1--55, 1910.

\bibitem{Gromak1999}
V.I. Gromak.
\newblock {B\"acklund transformations of Painlev\'e equations and their
  applications}.
\newblock In R.~Conte, editor, {\em The Painlev\'e Property, One Century
  Later}, CRM Series in Mathematical Physics, pages 687--734. Springer, 1999.

\bibitem{HietDry2002}
J.~Hietarinta and V.~Dryuma.
\newblock {Is my ODE a Painlevé equation in disguise?}
\newblock {\em Journal of Nonlinear Mathematical Physics}, 9:67--74, 2002.

\bibitem{HorvZhang2009}
P.A. Horvathy and P.~Zhang.
\newblock {Vortices in (Abelian) Chern–Simons gauge theory}.
\newblock {\em Physics Reports}, 481:83 -- 142, 2009.

\bibitem{Ince}
E.L. Ince.
\newblock {\em Ordinary Differential Equations}.
\newblock Dover Publications, 1956.

\bibitem{JackPi1990}
R.~Jackiw and S.-Y. Pi.
\newblock {Soliton solutions to the gauged nonlinear Schr\"odinger equation on
  the plane}.
\newblock {\em Phys. Rev. Lett.}, 64:2969--2972, 1990.

\bibitem{JafTaubes}
A.~Jaffe and C.~Taubes.
\newblock {\em {Vortices and Monopoles: Structure of Static Gauge Theories}}.
\newblock Boston: Birkhauser, 1980.

\bibitem{Kit1989}
A.V. Kitaev.
\newblock {Method of isomonodromic deformations for ``de\-generate'' third
  Painlev{\'e} equation}.
\newblock {\em Journal of Soviet Mathematics}, 46:2077--2083, 1989.

\bibitem{Kruglikov2008}
B.~Kruglikov.
\newblock {Invariant characterization of Liouville metrics and polynomial
  integrals}.
\newblock {\em Journal of Geometry and Physics}, 58:979 -- 995, 2008.

\bibitem{KruSpeight2010}
S.~Krusch and J.M. Speight.
\newblock Exact moduli space metrics for hyperbolic vortex polygons.
\newblock {\em Journal of Mathematical Physics}, 51:022304, 2010.

\bibitem{LeMason2002}
C.~Lebrun and L.J. Mason.
\newblock {Zoll manifolds and complex surfaces}.
\newblock {\em J. Differential Geom.}, 61:453--535, 2002.

\bibitem{Levine1948}
J.~Levine.
\newblock Invariant characterizations of two-dimensional affine and metric
  spaces.
\newblock {\em Duke Math. J.}, 15:69--77, 1948.

\bibitem{Liouville1886}
R.~Liouville.
\newblock Sur une classe d'équations différentielles non-linéaires.
\newblock {\em {Comptes rendus hebdomadaires des séances de l'Académie des
  Sciences}}, 103:457--460, 1886.

\bibitem{Liouville1889}
R.~Liouville.
\newblock Sur les invariants de certaines \'equations diff\'erentielles et sur
  leurs applications.
\newblock {\em Journal de l’Ecole Polytechnique}, 59:7--76, 1889.

\bibitem{Lohe1981}
M.A. Lohe.
\newblock Generalized noninteracting vortices.
\newblock {\em Phys. Rev. D}, 23:2335--2339, 1981.

\bibitem{LoheVDH1983}
M.A. Lohe and J.~van~der Hoek.
\newblock Existence and uniqueness of generalized vortices.
\newblock {\em Journal of Mathematical Physics}, 24:148--153, 1983.

\bibitem{MalMan2015}
R.~Maldonado and N.S. Manton.
\newblock Analytic vortex solutions on compact hyperbolic surfaces.
\newblock {\em Journal of Physics A: Mathematical and Theoretical}, 48:245403,
  2015.

\bibitem{ManSutbook}
N.~Manton and P.~Sutcliffe.
\newblock {\em Topological solitons}.
\newblock Cambridge monographs on mathematical physics. Cambridge University
  Press, 2004.

\bibitem{Manton2013}
N.S. Manton.
\newblock {Vortex solutions of the Popov equations}.
\newblock {\em Journal of Physics A: Mathematical and Theoretical}, 46:145402,
  2013.

\bibitem{Manton2017}
N.S. Manton.
\newblock Five vortex equations.
\newblock {\em Journal of Physics A: Mathematical and Theoretical}, 50:125403,
  2017.

\bibitem{ManRink2010}
N.S. Manton and N.A. Rink.
\newblock Vortices on hyperbolic surfaces.
\newblock {\em Journal of Physics A: Mathematical and Theoretical}, 43:434024,
  2010.

\bibitem{MasonWoodBook}
L.J. Mason and N.M.J. Woodhouse.
\newblock {\em {Integrability, Self-duality, and Twistor Theory}}.
\newblock London Mathematical Society Monographs. Clarendon Press, 1996.

\bibitem{Mik1979}
A.~Mikhailov.
\newblock {Integrability of a two-dimensional generalization of the Toda
  chain}.
\newblock {\em Soviet Phs. JETP Lett.}, 30:414--418, 1979.

\bibitem{Mik1981}
A.~Mikhailov.
\newblock The reduction problem and the inverse scattering method.
\newblock {\em Physica}, 3D:73--117, 1981.

\bibitem{OlNut1988}
P.~Olver and Y.~Nutku.
\newblock Hamiltonian structures for systems of hyperbolic conservation laws.
\newblock {\em J. Math. Phys}, 29:1610--1619, 1988.

\bibitem{Painleve1902}
P.~Painlevé.
\newblock {Sur les équations différentielles du second ordre et d'ordre
  supérieur dont l'intégrale générale est uniforme}.
\newblock {\em Acta Mathematica}, 25:1--85, 1902.

\bibitem{Painleve1900}
P.~Painlev\'e.
\newblock {Mémoire sur les équations différentielles dont l'intégrale
  générale est uniforme}.
\newblock {\em Bulletin de la Société Mathématique de France}, 28:201--261,
  1900.

\bibitem{Palais1979}
R.S. Palais.
\newblock The principle of symmetric criticality.
\newblock {\em Communications in Mathematical Physics}, 69:19--30, 1979.

\bibitem{PenRind1984}
R.~Penrose and W.~Rindler.
\newblock {\em Spinors and Space-Time}, volume~1.
\newblock Cambridge University Press, 1984.

\bibitem{Popov2009}
A.D. Popov.
\newblock {Integrability of vortex equations on Riemann surfaces}.
\newblock {\em Nuclear Physics B}, 821:452 -- 466, 2009.

\bibitem{Popov2012}
A.D. Popov.
\newblock Integrable vortex-type equations on the two-sphere.
\newblock {\em Phys. Rev. D}, 86:105044, 2012.

\bibitem{Popov1993}
A.G. Popov.
\newblock {Exact formulas for constructing solutions of the Liouville equation
  $\Delta_2u=e^u$ from solutions of the Laplace equation $\Delta_2v=0$.
  (Russian)}.
\newblock {\em Russian Akad. Sci. Math.}, 48:570--572, 1994.

\bibitem{Samols1992}
T.M. Samols.
\newblock Vortex scattering.
\newblock {\em Comm. Math. Phys.}, 145:149--179, 1992.

\bibitem{Shiff1991}
J.~Schiff.
\newblock {Integrability of Chern–Simons–Higgs and Abelian Higgs vortex
  equations in a background metric}.
\newblock {\em Journal of Mathematical Physics}, 32:753--761, 1991.

\bibitem{Sheftel1986}
M.B. Sheftel.
\newblock {Integration of Hamiltonian systems of hydrodynamic type with two
  dependent variables with the aid of the Lie---B{\"a}cklund group}.
\newblock {\em Functional Analysis and its Applications}, 20:227--235, 1986.

\bibitem{Strachan1992}
I.A.B. Strachan.
\newblock {Low velocity scattering of vortices in a modified Abelian Higgs
  model}.
\newblock {\em J. Math. Phys.}, 33:102--110, 1992.

\bibitem{Taubes1980}
C.H. Taubes.
\newblock {Arbitrary $N$-vortex solutions to the first order Ginzburg-Landau
  equations}.
\newblock {\em Comm. Math. Phys.}, 72:277--292, 1980.

\bibitem{Thomas1925}
T.Y. Thomas.
\newblock {On the projective and equi-projective geometries of paths}.
\newblock {\em Proceedings of the National Academy of Sciences}, 11:199--203,
  1925.

\bibitem{Thompson1999}
G.~Thompson.
\newblock Killing's equations in dimension two and systems of finite type.
\newblock {\em Mathematica Bohemica}, 124:401--420, 1999.

\bibitem{Tsarev1985}
S.P. Tsarev.
\newblock {Poisson brackets and one-dimensional Hamiltonian systems of
  hydrodynamic type}.
\newblock {\em Soviet Math. Dokl.}, 31:488, 1985.

\bibitem{WTC1983}
J.~Weiss, M.~Tabor, and G.~Carnevale.
\newblock {The Painlevé property for partial differential equations}.
\newblock {\em Journal of Mathematical Physics}, 24:522--526, 1983.

\bibitem{Witten1977}
E.~Witten.
\newblock {Some exact multipseudoparticle solutions of classical Yang-Mills
  theory}.
\newblock {\em Phys. Rev. Lett.}, 38:121--124, 1977.

\end{thebibliography}
\bibliographystyle{plain}

%\begin{thebibliography}{99}
 
% An alphabetical listing of all references must be used.
%\end{thebibliography}

\end{document}